%% file: Manuscript_Trigonometric_NH.tex
\documentclass[aip,cha,reprint,floatfix]{revtex4-2}
\usepackage{graphicx}
\usepackage{dcolumn}
\usepackage{bm}
\usepackage[utf8]{inputenc}
\usepackage[T1]{fontenc}
\usepackage{algpseudocode}
\usepackage{mathptmx}
\usepackage{amsthm}
\usepackage{enumerate}
\usepackage{booktabs}
\usepackage{tabularx}
\usepackage[mathscr]{eucal}
\usepackage{amsmath,amssymb}
\usepackage{mathtools}
\usepackage{subfigure}
\usepackage[english]{babel}
\usepackage{lmodern}
\newtheorem{theorem}{Theorem}

\newtheorem{lemma}{Lemma}
\newtheorem{remark}{Remark}

\newcounter{algorithm}

\input{mathdef.tex}

\begin{document}
\mathtoolsset{
mathic 
}
\title{document}
\title{Trigonometric Nos\'e--Hoover oscillator: chaos, periodic orbits and integrability}
\author{Wojciech Szumi{\'n}ski}
\email{w.szuminski@if.uz.zgora.pl}
\affiliation{Institute of Physics, University of Zielona G\'ora,
Licealna 9, PL-65-407 Zielona G\'ora, Poland}

\author{Jaume Llibre}
\email{jaume.llibre@uab.cat}
\affiliation{Departament de Matem\`atiques,
Universitat Aut\`onoma de Barcelona,
08193 Bellaterra, Barcelona, Catalonia, Spain}
\affiliation{Reial Acad\`emia de Ci\`encies i Arts de Barcelona,
La Rambla 115, 08002 Barcelona, Catalonia, Spain}
\begin{abstract}
We introduce a trigonometric version of the Nos\'e--Hoover oscillator
in which the quadratic mechanical terms and unbounded thermostat
coupling are replaced by bounded trigonometric functions. This
formulation replaces the harmonic potential by a pendulum-type
potential and confines the thermostat interaction to a bounded periodic
form. The resulting two-parameter system is naturally defined on the
three-dimensional torus and reduces near the origin, to leading order,
to the classical polynomial Nos\'e--Hoover model.

We investigate its global dynamics using Poincar\'e sections,
bifurcation diagrams, Lyapunov spectra, Kaplan--Yorke dimensions, and
the Lyapunov Integrability Test (LIT). The numerical results reveal the
coexistence of regular and chaotic dynamics and characterize changes in
dissipative behavior across the parameter plane.

We then analyze the two limiting cases associated with the parameter
axes. For $a=0$, we construct two functionally independent first
integrals on regular domains, whereas for $b=0$ the dynamics reduces to
a family of two-dimensional systems on invariant tori, which are
analyzed using Darboux polynomials and exponential factors.

First-order averaging near the intersection of these integrable limits
yields periodic solutions bifurcating from unperturbed periodic orbits
and an obstruction to regular $C^1$ first integrals in their
neighborhoods. Independently, differential Galois theory applied to the
normal variational equation, together with the Ayoul--Zung and Li--Shi
criteria, excludes meromorphic $B$-integrability and non-constant
meromorphic first integrals near a particular non-equilibrium phase
curve for $ab\neq0$. Thus, despite retaining the local structure of the
classical Nos\'e--Hoover oscillator, its trigonometric counterpart
exhibits markedly different global dynamics and integrability
properties.
\end{abstract}
\maketitle
\noindent\textbf{The Nos\'e--Hoover oscillator provides a remarkably simple
setting in which regular and chaotic dynamics can coexist. In this
work we ask what remains of this familiar picture when the classical
polynomial model is replaced by a fully periodic one. The resulting
trigonometric oscillator retains the local leading-order structure of
the classical system, but its dynamics now takes place on the compact
three-dimensional torus. We show that this change in global geometry
has substantial consequences: regular and chaotic motion coexist over
broad parameter ranges, period-doubling cascades and phase-space
contraction appear, and the integrable limits of the classical model
are significantly altered. To explore this dynamics, we also use our
recent Lyapunov Integrability Test (LIT), which provides both qualitative
and quantitative information about chaotic behavior and indications of
the possible existence of first integrals. On the analytical side, averaging theory identifies periodic orbits
emerging near the integrable regime and provides a local obstruction
to regular first integrals. Differential Galois theory independently
excludes meromorphic $B$-integrability and non-constant meromorphic
first integrals near a particular non-equilibrium phase curve for
arbitrary nonzero values of both parameters. To the best of
our knowledge, this is one of the first applications in which
non-integrability conditions obtained from averaging theory and
differential Galois theory are combined within the analysis of the same
dynamical system. The trigonometric Nos\'e--Hoover oscillator therefore
provides a simple setting in which numerical diagnostics, periodic
dynamics, and complementary notions of integrability and
non-integrability can be studied together.}

\section{Introduction}
\label{sec:introduction}

The Nos\'e--Hoover oscillator originates from the extended Hamiltonian
formulation introduced by Nos\'e~\cite{Nose:84::}. For a one-dimensional harmonic oscillator, after normalization of the physical constants, the extended Hamiltonian can be written as
\begin{equation}
\label{eq:Nose_Hamiltonian}
\mathcal{H}_{\mathrm{N}}
=
\frac{y^{2}}{2s^{2}}
+\frac{x^{2}}{2}
+\frac{p_s^{2}}{2}
+T\log s,
\end{equation}
where $x$ and $y$ denote the normalized position and momentum,
respectively, while $s$ is the additional scaling variable and $p_s$
is its conjugate momentum. The logarithmic term introduces the prescribed temperature $T$ into
the extended thermostat dynamics.

Hamilton's equations are
\begin{equation}
\label{eq:Nose_Hamilton_equations}
\dot x=\frac{y}{s^{2}},
\quad
\dot y=-x,
\quad
\dot s=p_s,
\quad
\dot p_s=\frac{y^{2}}{s^{3}}-\frac{T}{s}.
\end{equation}

After the Nos\'e time transformation and the corresponding rescaling
of the physical momentum, the variable $s$ can be eliminated from the
physical equations, while the additional thermostat degree of freedom
is represented by a dynamical friction coefficient $\zeta$. To make
the structure of the resulting equations explicit, we first separate
the mechanical part of the dynamics from the thermostat contribution.

In the normalized physical variables, the mechanical subsystem is the
harmonic oscillator with Hamiltonian
\begin{equation}
\label{eq:Hclassical_mechanical}
H_{\mathrm{cl}}(x,y)
=
K_{\mathrm{cl}}(y)+V_{\mathrm{cl}}(x)
=
\frac{y^{2}}{2}+\frac{x^{2}}{2},
\end{equation}
where
\begin{equation}
\label{eq:classical_KV}
K_{\mathrm{cl}}(y)=\frac{y^{2}}{2},
\qquad
V_{\mathrm{cl}}(x)=\frac{x^{2}}{2}.
\end{equation}
The Nos\'e--Hoover thermostat augments this mechanical dynamics by
introducing the friction term $-\zeta\,\partial H_{\mathrm{cl}}/\partial y$
in the momentum equation, while the evolution of $\zeta$ is governed
by the deviation of the instantaneous kinetic energy from its prescribed
reference value. The reduced equations thus take the form
\begin{equation}
\label{eq:classical_general_thermostat}
\begin{aligned}
\dot x
&=
\frac{\partial H_{\mathrm{cl}}}{\partial y},
\\
\dot y
&=
-\frac{\partial H_{\mathrm{cl}}}{\partial x}
-\zeta
 \frac{\partial H_{\mathrm{cl}}}{\partial y},
\\
\dot\zeta
&=
2\alpha\left(K_{\mathrm{cl}}(y)-K_{0}\right),
\qquad
K_{0}=\frac{T}{2}.
\end{aligned}
\end{equation}
Substituting~\eqref{eq:classical_KV} into
system~\eqref{eq:classical_general_thermostat} yields the standard
normalized Nos\'e--Hoover oscillator
\begin{equation}
\label{eq:NH_standard_physical}
\begin{cases}
\dot x=y,\\
\dot y=-x-\zeta y,\\
\dot\zeta=\alpha(y^{2}-T),
\end{cases}
\end{equation}
where $\alpha$ determines the response rate of the thermostat.

This form makes the role of the thermostat transparent. The variable
$\zeta$ acts as a time-dependent friction coefficient, while its
evolution is controlled by the deviation of the instantaneous kinetic
energy from the reference value $K_{0}$. Although
system~\eqref{eq:NH_standard_physical} is not Hamiltonian in
$(x,y,\zeta)$, it originates from the extended Hamiltonian
construction of Nos\'e.

The same representation gives a simple balance law for the mechanical
energy. Along solutions of~\eqref{eq:classical_general_thermostat},
\begin{equation*}
\label{eq:Hclassical_balance_standard}
\begin{split}
\frac{\rmd H_{\mathrm{cl}}}{\rmd t}
=
\frac{\partial H_{\mathrm{cl}}}{\partial x}\dot x
+
\frac{\partial H_{\mathrm{cl}}}{\partial y}\dot y
=
-\zeta
\left(
\frac{\partial H_{\mathrm{cl}}}{\partial y}
\right)^2
=
-\zeta y^{2}.
\end{split}
\end{equation*}
Thus, $\zeta>0$ corresponds to instantaneous energy extraction from
the mechanical subsystem, whereas $\zeta<0$ corresponds to energy
injection.

The dynamics of the Nos\'e--Hoover oscillator has been studied
extensively. Posch, Hoover and Vesely showed that even a harmonic
oscillator with a single mechanical degree of freedom, when coupled
to the Nos\'e--Hoover thermostat, can exhibit coexisting regular and
chaotic trajectories. They also found that the resulting dynamics is
not sufficiently chaotic to provide a fully ergodic description of
the canonical ensemble~\cite{PoschHooverVesely:86::}. In the slow-thermostat regime,
Nos\'e showed that the dynamics displays long recurrence cycles and slow
energy modulation, which can be described perturbatively by an
integrable averaged Hamiltonian~\cite{Nose:93::}. Related studies have also shown that Nos\'e--Hoover thermostated
systems can exhibit Lyapunov instability, limit cycles, and phase-space
contraction~\cite{HooverPoschHolian1987}, while numerical investigations of
thermostated harmonic oscillators revealed regular oscillations for
weak coupling and increasingly irregular dynamics as the thermostat
coupling is increased~\cite{GoloSalnikovShaitan:04::}. These features,
together with the coexistence of invariant tori and chaotic motion,
remain characteristic of the classical Nos\'e--Hoover oscillator and
its variants~\cite{Sprott:20::}.

A particularly useful two-parameter representation
was considered in a previous study~\cite{Llibre:21::}. Setting $T=1$ and, for $a\neq0$, introducing
the rescaling
\begin{equation*}
\label{eq:NH_parameter_transformation}
\zeta=az,
\qquad
\alpha=ab,
\end{equation*}
transforms~\eqref{eq:NH_standard_physical} into the family
\begin{equation}
\label{eq:modified_NH}
\begin{cases}
\dot x=y,\\
\dot y=-x-ayz,\\
\dot z=b(y^{2}-1),
\end{cases}
\qquad
a,b\in\mathbb R.
\end{equation}
For $a\neq0$, this is therefore a reparametrized form of the standard
Nos\'e--Hoover oscillator, with $\alpha=ab$. At the same time,
representation~\eqref{eq:modified_NH} separates the friction coupling,
controlled by $a$, from the thermostat feedback, controlled by $b$,
and extends naturally to the limiting cases in which one of these
mechanisms is switched off.

These limiting cases are especially important for the present work.
System~\eqref{eq:modified_NH} is completely integrable when either
$a=0$ or $b=0$.  The complete integrability of a three-dimensional
autonomous vector field means the existence of two functionally
independent first integrals on an open dense subset of the phase
space.

For $a=0$, two functionally
independent first integrals are
\begin{equation*}
\label{eq:NH_integrals_a0}
H_{1}=H_{\mathrm{cl}}(x,y),
\end{equation*}
and
\begin{equation*}
\label{eq:NH_integral_F1}
F_{1}
=
2z-bxy
-b\left(x^{2}+y^{2}-2\right)
\arctan\left(\frac{x}{y}\right),
\end{equation*}
with $y\neq0$. 
The level sets of $H_{1}$ are invariant circular cylinders
\[
\mathcal C_c
=
\left\{(x,y,z)\in\mathbb R^3:
x^{2}+y^{2}=c^2\right\},
\qquad c>0,
\]
whose intersections with the planes $z=\mathrm{const}$ are circles
of radius $c$.

For $b=0$, two functionally independent first integrals are
\begin{equation*}
\label{eq:NH_integral_H2}
H_{2}=z,
\end{equation*}
and
\begin{equation*}
\begin{split}
\label{eq:NH_integral_F2}
F_{2}
&=
\left(x^{2}+y^{2}+axyz\right)
\\
&\quad\times
\exp\left[
\frac{2az}{\sqrt{4-a^{2}z^{2}}}
\arctan\left(
\frac{2x+ayz}
{y\sqrt{4-a^{2}z^{2}}}
\right)
\right],
\end{split}
\end{equation*}
defined for $y\neq0$ and $z\neq\pm2/a$. In this case, the level
sets of $H_{2}$ are the invariant planes $z=c$. Hence, the
two-parameter polynomial family contains two distinguished completely
integrable limits embedded in a system whose generic dynamics may be
considerably more complicated.

The classical Nos\'e--Hoover oscillator is built around the harmonic
oscillator, but the mechanical part need not be restricted to a
quadratic potential. Periodic potentials appeared in Nos\'e--Hoover
dynamics already in early studies of nonequilibrium systems. Hoover
et al.~\cite{HooverPoschHolian1987} considered a thermostatted particle
moving in the periodic pendulum potential
\[
V(x)=\varepsilon(1-\cos x),
\]
and subject to a constant external field. In dimensionless variables
their model takes the form
\[
\dot x=p,\qquad
\dot p=F-\varepsilon\sin x-\zeta p,\qquad
\dot\zeta=\alpha(p^2-1).
\]
It may be viewed as a one-particle limit of the Frenkel--Kontorova
model coupled to a Nos\'e--Hoover thermostat~\cite{HooverPoschHolian1987}. For small oscillations
around a minimum of the periodic potential,
$
\varepsilon\sin x \approx \varepsilon x,
$
so the harmonic restoring force is recovered locally. Thus, a
periodic mechanical potential is already compatible with the basic
Nos\'e--Hoover construction.

This observation suggests a natural extension. Instead of making only
the potential periodic, we impose periodicity on the entire
three-dimensional system. The quadratic kinetic and potential energies
are replaced by bounded trigonometric functions, while the thermostat
coupling is periodicized accordingly. The thermostat feedback retains
the same dependence on the deviation of the kinetic energy from a
prescribed reference value.

Thus, we introduce the trigonometric Nos\'e--Hoover oscillator
\begin{equation}
\label{eq:periodic_NH}
\begin{cases}
\dot x=\sin y,\\
\dot y=-\sin x-a\sin y\sin z,\\
\dot z=b(1-2\cos y),
\end{cases}
\qquad
a,b\in\mathbb R.
\end{equation}
The vector field is $2\pi$-periodic in each phase variable, and hence
system~\eqref{eq:periodic_NH} naturally defines an autonomous flow on
the three-dimensional torus $\mathbb T^3$. Its relation to the
normalized classical Nos\'e--Hoover oscillator and the motivation for
the trigonometric construction are discussed in
Sec.~\ref{sec:periodic_generalization}.

Three-dimensional trigonometric systems on periodic domains are common
in mathematical physics. Well-known examples include the
Arnold--Beltrami--Childress (ABC) flow and the Roberts
flow~\cite{Arnold1965,Beltrami1889,Childress1970,Roberts1972}. Their
vector fields are periodic in the spatial variables and are naturally
studied on compact periodic domains. The ABC flow, in particular,
provides a standard example of a simple autonomous system on the torus
with a mixed phase-space structure, where regular and chaotic motions
may coexist. These examples place system~\eqref{eq:periodic_NH} in the
broader class of three-dimensional trigonometric flows on compact phase
spaces.

The main question is how the global periodicity changes the dynamics and integrability
of the Nos\'e--Hoover oscillator. Of particular interest are the two
integrable limits of the polynomial family. The classical system is
completely integrable for $a=0$ and for $b=0$, and it is not clear a
priori whether both integrable families survive after
trigonometricization. More generally, it is natural to ask to what
extent the characteristic dynamical features of the classical
Nos\'e--Hoover oscillator persist in the trigonometric model. Although
the two systems have the same leading local structure, their global
phase spaces are fundamentally different: the polynomial model evolves
in an unbounded phase space, whereas the trigonometric system is defined
on the compact torus $\mathbb T^3$. This change may affect not only the
existence and form of first integrals, but also the organization of
regular and chaotic regions, the bifurcation structure, and the
phase-space contraction properties of the flow.

We investigate these questions using complementary analytical and
numerical methods. After introducing the physical interpretation of
the trigonometric construction, we describe the basic structural
properties of system~(8) and examine the limiting cases corresponding
to the parameter axes. The global dynamics is then explored through
Poincar\'e sections, bifurcation diagrams, Lyapunov spectra,
Kaplan--Yorke dimensions, and Lyapunov maps, which reveal the
coexistence of regular and chaotic motion and characterize
phase-space contraction across the parameter space. The integrability
problem is subsequently studied analytically. A tangent half-angle
transformation and Darboux theory are used to investigate the reduced
dynamics in the limit $b=0$, while first-order averaging near the
intersection of the parameter axes yields periodic solutions and
non-integrability conditions. Finally, meromorphic non-integrability
is studied by means of Morales--Ramis theory and the differential
Galois groups of the corresponding normal variational equations.
Together, these approaches allow us to distinguish the dynamical
features inherited locally from the classical Nos\'e--Hoover
oscillator from those produced by the global trigonometric geometry.

\section{Construction of the trigonometric model}
\label{sec:periodic_generalization}

From this point onward, all variables and energies are understood in
the normalized dimensionless form used for the classical
Nos\'e--Hoover oscillator. In particular, $x$ and $y$ are dimensionless
canonical variables of the mechanical subsystem, whose classical
Hamiltonian is given by in~\eqref{eq:Hclassical_mechanical}.

The construction introduced in system~\eqref{eq:periodic_NH} should
therefore be understood as a trigonometric counterpart of this
normalized model. In particular, it does not amount to replacing a
dimensional physical momentum by a trigonometric function.

The quadratic kinetic and potential terms are replaced by
\[
K_{\mathrm{tr}}(y)=1-\cos y,
\qquad
V_{\mathrm{tr}}(x)=1-\cos x.
\]
These functions are dimensionless, bounded, and $2\pi$-periodic. 
The corresponding mechanical Hamiltonian is
\begin{equation}
\label{eq:Hperiodic}
H_{\mathrm{tr}}(x,y)
=
K_{\mathrm{tr}}(y)+V_{\mathrm{tr}}(x)
=
2-\cos y-\cos x.
\end{equation}
Its derivatives are
\[
\frac{\partial H_{\mathrm{tr}}}{\partial y}
=
\sin y,
\qquad
-\frac{\partial H_{\mathrm{tr}}}{\partial x}
=
-\sin x.
\]
Thus, the first equation of system~\eqref{eq:periodic_NH} is generated
by the periodic kinetic term, while the conservative part of the
second equation is generated by the periodic potential.

The thermostat coupling is modified in the same spirit. In the
polynomial model, the effective friction coefficient is proportional
to $az$. Here it is replaced by the bounded periodic function
\[
\zeta_{\mathrm{tr}}(z)=a\sin z.
\]
The corresponding friction term becomes
\[
-\zeta_{\mathrm{tr}}(z)
\frac{\partial H_{\mathrm{tr}}}{\partial y}
=
-a\sin z\sin y.
\]

The feedback equation is obtained by preserving the same dependence
on the deviation of the kinetic term from its reference value. We set
\begin{equation}
\label{eq:periodic_feedback_energy}
\dot z
=
2b\left(K_{\mathrm{tr}}(y)-K_0\right),
\qquad
K_0=\frac12.
\end{equation}
The value $K_0=1/2$ corresponds to the normalization $T=1$ used in the
polynomial Nos\'e--Hoover system. Since
\[
2K_{\mathrm{tr}}(y)-1
=
1-2\cos y,
\]
equation~\eqref{eq:periodic_feedback_energy} gives precisely the third
equation of system~\eqref{eq:periodic_NH}.

The trigonometric system~\eqref{eq:periodic_NH} can therefore be written
in the structural form
\begin{equation}
\label{eq:periodic_general_thermostat}
\begin{aligned}
\dot x
&=
\frac{\partial H_{\mathrm{tr}}}{\partial y},
\\
\dot y
&=
-\frac{\partial H_{\mathrm{tr}}}{\partial x}
-\zeta_{\mathrm{tr}}(z)
 \frac{\partial H_{\mathrm{tr}}}{\partial y},
\\
\dot z
&=
2b\left(K_{\mathrm{tr}}(y)-K_0\right).
\end{aligned}
\end{equation}
This representation makes the relation with the classical
Nos\'e--Hoover mechanism explicit. The mechanical subsystem is coupled
to the thermostat variable through the periodic friction coefficient
$\zeta_{\mathrm{tr}}(z)$, while the thermostat responds to the
difference between the instantaneous kinetic term and the prescribed
reference value $K_0$.

The parameters $a$ and $b$ have distinct roles. The parameter $a$
sets the amplitude of the thermostat-dependent friction,
\[
-|a|
\leq
\zeta_{\mathrm{tr}}(z)
\leq
|a|,
\]
whereas $b$ controls the rate at which the thermostat responds to the
deviation of $K_{\mathrm{tr}}$ from $K_0$. Unlike in the polynomial
family, these parameters cannot in general be combined by a linear
rescaling of the thermostat variable. Indeed, setting $\zeta=az$
gives
\[
a\sin z
=
a\sin\left(\frac{\zeta}{a}\right),
\qquad
\dot\zeta
=
ab(1-2\cos y).
\]
Hence, $a$ remains explicitly in the nonlinear periodic coupling and
cannot be absorbed into the product $ab$ as in the classical case. The parameters $a$ and $b$
therefore remain independent parameters of the trigonometric model.

The mechanical energy balance follows directly
from~\eqref{eq:periodic_general_thermostat}. Along its solutions,
\begin{equation}
\label{eq:Hperiodic_balance}
\begin{split}
\frac{\rmd H_{\mathrm{tr}}}{\rmd t}
&=
\frac{\partial H_{\mathrm{tr}}}{\partial x}\dot x
+
\frac{\partial H_{\mathrm{tr}}}{\partial y}\dot y
\\
&=
-\zeta_{\mathrm{tr}}(z)
\left(
\frac{\partial H_{\mathrm{tr}}}{\partial y}
\right)^2
=
-a\sin z\,\sin^2 y.
\end{split}
\end{equation}
Hence, the sign of $a\sin z$ determines the instantaneous direction of
energy exchange. If $a\sin z>0$, the mechanical energy decreases,
whereas if $a\sin z<0$, the thermostat injects energy into the
mechanical subsystem.

The local relation with the polynomial Nos\'e--Hoover
family~\eqref{eq:modified_NH} follows directly from the Taylor
expansions
\[
\sin u=u+O(u^3),
\qquad
1-\cos u=\frac{u^2}{2}+O(u^4).
\]
Consequently, the Taylor expansion of
system~\eqref{eq:periodic_NH} about the origin takes the form
\[
\begin{cases}
\dot x=y+O(y^3),\\
\dot y=-x-ayz+O\!\left(\|(x,y,z)\|^3\right),\\
\dot z=b(y^2-1)+O(y^4).
\end{cases}
\]
Thus, the leading-order terms reproduce the modified polynomial
Nos\'e--Hoover family~\eqref{eq:modified_NH}.

The construction should not be interpreted as an exact canonical
transformation of the original extended Nos\'e Hamiltonian.
System~\eqref{eq:periodic_NH} is instead a trigonometric realization
of the normalized Nos\'e--Hoover dynamics. It preserves the structural
form of the mechanical energy exchange and thermostat feedback, while
its global dynamics is defined on the compact phase space
$\mathbb T^3$. Whether this thermostat admits an invariant statistical
measure analogous to the canonical measure of the classical
Nos\'e--Hoover system is a separate question and is not assumed here.

\section{Structural properties}
\label{sec:basic_properties}

Before proceeding to the numerical and analytical study of
system~\eqref{eq:periodic_NH}, we describe some of its basic structural
properties. In particular, we discuss its natural phase space,
symmetries, divergence, and equilibrium structure.
\subsection{Periodicity and phase space}

The vector field associated with system~\eqref{eq:periodic_NH} is
$2\pi$-periodic in each variable. More precisely, if
$
\boldsymbol{x}=(x,y,z)
$
and the vector field is denoted by $\vv=\vv(\vx)$, then
\begin{equation*}
\label{eq:periodicity_vector_field}
\vv(x+2k_1\pi,y+2k_2\pi,z+2k_3\pi)
=
\vv(x,y,z),
\end{equation*}
with $k_1,k_2,k_3\in\mathbb Z$. 
Consequently, system~\eqref{eq:periodic_NH} naturally defines a smooth
vector field on the three-dimensional torus
\begin{equation*}
\label{eq:phase_space_torus}
\mathbb T^3
=
(\mathbb R/2\pi\mathbb Z)^3,
\end{equation*}
which will be regarded as its natural phase space.

When the system is represented in the covering space $\mathbb R^3$,
every orbit and invariant set is reproduced under translations by
integer multiples of $2\pi$ in each coordinate direction. This
periodic phase-space geometry constitutes a fundamental global
difference from the modified polynomial Nos\'e--Hoover
system~\eqref{eq:modified_NH}, whose natural phase space is
$\mathbb R^3$.

All components of the vector field are globally bounded. Indeed, from
system~\eqref{eq:periodic_NH} one immediately obtains
\begin{equation*}
\label{eq:vector_field_bounds}
|\dot x|\leq 1,
\qquad
|\dot y|\leq 1+|a|,
\qquad
|\dot z|\leq 3|b|.
\end{equation*}
Thus, in contrast to the polynomial Nos\'e--Hoover oscillator, the
nonlinear interactions remain bounded throughout the phase space.
Moreover, since the vector field is smooth and $\mathbb T^3$ is
compact, the corresponding flow is complete. Hence, for every initial
condition $\boldsymbol{x}_0\in\mathbb T^3$, the solution
$\Phi^t(\boldsymbol{x}_0)$ is defined for all $t\in\mathbb R$.

\subsection{Symmetries}

System~\eqref{eq:periodic_NH} possesses both an ordinary symmetry and
a reversing symmetry. First, consider the involution
\begin{equation*}
\label{eq:symmetry_R}
R:\mathbb T^3\longrightarrow\mathbb T^3,
\qquad
R(x,y,z)=(-x,-y,z).
\end{equation*}
The vector field satisfies
\begin{equation*}
\label{eq:equivariance_R}
\vv(R\boldsymbol{x})
=
DR\,\vv(\boldsymbol{x}),
\end{equation*}
where
\[
DR=\operatorname{diag}(-1,-1,1).
\]
Thus, $R$ is a symmetry of the system. In particular, if
\[
\boldsymbol{x}(t)=(x(t),y(t),z(t))
\]
is a solution of~\eqref{eq:periodic_NH}, then
\[
R\boldsymbol{x}(t)
=
(-x(t),-y(t),z(t))
\]
is also a solution. Consequently, invariant sets and periodic orbits
which are not themselves invariant under $R$ occur in
symmetry-related pairs.

The system is also reversible with respect to the involution
\begin{equation*}
\label{eq:reversing_symmetry}
S:\mathbb T^3\longrightarrow\mathbb T^3,
\qquad
S(x,y,z)=(x,-y,-z).
\end{equation*}
Indeed,
\begin{equation*}
\label{eq:reversibility_S}
\vv(S\boldsymbol{x})
=
-DS\,\vv(\boldsymbol{x}),
\end{equation*}
where
\[
DS=\operatorname{diag}(1,-1,-1).
\]
Therefore, if $\boldsymbol{x}(t)$ is a solution, then
\begin{equation*}
\label{eq:reversed_solution}
S\boldsymbol{x}(-t)
=
\bigl(x(-t),-y(-t),-z(-t)\bigr)
\end{equation*}
is also a solution. This reversibility imposes an additional
constraint on the organization of periodic orbits and other invariant
sets in the phase space.

\subsection{Divergence and phase-space volume}
\label{subsec:divergence}

The divergence of the vector field $\vv(\vx)$ associated with
system~\eqref{eq:periodic_NH} is
\begin{equation}
\label{eq:divergence_periodic_NH}
\nabla\cdot\vv
=
\frac{\partial\dot x}{\partial x}
+
\frac{\partial\dot y}{\partial y}
+
\frac{\partial\dot z}{\partial z}
=
-a\cos y\sin z.
\end{equation}
Hence, the system is not volume preserving in general. The local rate
of phase-space expansion or contraction depends periodically on $y$
and $z$ and satisfies
\begin{equation*}
\label{eq:divergence_bound}
|\nabla\cdot\vv|
\leq |a|.
\end{equation*}
The divergence vanishes identically if and only if $a=0$.
Consequently, for $a=0$ the flow preserves the standard volume form
$
dx\wedge dy\wedge dz
$
on $\mathbb T^3$. For $a\neq0$, the divergence takes both positive and negative values
on $\mathbb T^3$, while vanishing on the hypersurfaces
$\cos y=0$ or $\sin z=0$.

More precisely, let $\Phi^t$ denote the flow generated by
system~\eqref{eq:periodic_NH}. By Liouville's formula
\cite{KatokHasselblatt1995},
\begin{equation}
\label{eq:Liouville_formula}
\det D_{\vx}\Phi^t(\boldsymbol{x}_0)
=
\exp\left[
\int_0^t
\nabla\cdot\vv\bigl(\Phi^s(\boldsymbol{x}_0)\bigr)\,ds
\right].
\end{equation}
Consequently, whenever the Lyapunov exponents exist, their sum satisfies~\cite{Benettin:80::,Wolf:85::, Pikovsky:16::}
\begin{equation}
\begin{split}
\label{eq:Lyapunov_sum_divergence}
\mathrm{\Lambda}_{\Sigma}(\boldsymbol{x}_0)
&:=
\sum_{i=1}^{3}\lambda_i(\boldsymbol{x}_0)
\\ & =
\lim_{T\to\infty}
\frac{1}{T}
\int_0^T
\nabla\cdot\vv\bigl(\Phi^t(\boldsymbol{x}_0)\bigr)\,dt,
\end{split}
\end{equation}
Thus, $\mathrm{\Lambda}_{\Sigma}$ measures the asymptotic mean
phase-space contraction or expansion along a trajectory. In particular,
for $a=0$ the divergence vanishes identically and
$\mathrm{\Lambda}_{\Sigma}=0$. For $a\neq0$, the divergence takes both
signs in phase space, and its long-time average depends on the
trajectory. This quantity will be examined numerically below through
the computed Lyapunov spectrum.

\subsection{Equilibrium set}
\label{subsec:equilibria}

The equilibrium set of system~\eqref{eq:periodic_NH} is determined by
\begin{equation*}
\label{eq:equilibrium_conditions}
\sin y=0,
\quad
\sin x+a\sin y\sin z=0,
\quad
b(1-2\cos y)=0.
\end{equation*}
For $b\neq0$, these conditions admit no solution on $\mathbb T^3$,
since $\sin y=0$ and $1-2\cos y=0$ cannot hold simultaneously.
Therefore,
\begin{equation*}
\label{eq:no_equilibria}
\mathcal E=\varnothing,
\qquad b\neq0.
\end{equation*}
In particular, as in the modified polynomial Nos\'e--Hoover
system~\eqref{eq:modified_NH}, the generic thermostat regime is
equilibrium-free.

For $b=0$, the angular coordinate $z$ is constant along every
trajectory. Consequently, the phase space is foliated by the
invariant tori 
\begin{equation*}
\label{eq:invariant_tori_b0}
\mathcal T_c
=
\left\{
(x,y,z)\in\mathbb T^3:z=c
\right\},
\qquad
c\in\mathbb S^1.
\end{equation*}
On each $\mathcal T_c$, the dynamics is governed by
\begin{equation*}
\label{eq:reduced_b0}
\begin{cases}
\dot x=\sin y,\\
\dot y=-\sin x-a\sin c\,\sin y.
\end{cases}
\end{equation*}
Its equilibria are given by
\[
(x,y)\in\{0,\pi\}\times\{0,\pi\}.
\]
Consequently, the full equilibrium set is
\begin{equation*}
\label{eq:equilibrium_b0}
\mathcal E
=
\bigcup_{(\xi,\eta)\in\{0,\pi\}^2}
\left\{(\xi,\eta,z):z\in\mathbb S^1\right\},
\end{equation*}
and hence consists of four disjoint equilibrium circles.

Thus, $b=0$ is a singular parameter limit from the dynamical point of
view: the equilibrium-free three-dimensional flow for $b\neq0$ is
replaced by a foliation into invariant two-dimensional tori, each
containing four equilibria. This reduction will also play an important
role in the analysis of the integrability properties of the system.

\section{Global dynamics}
\label{sec:numerical_dynamics}

We now investigate the global dynamics of the trigonometric
Nos\'e--Hoover oscillator~\eqref{eq:periodic_NH}. We begin near the
integrable limit $a=0$ and examine how its regular phase-space
organization changes as the coupling parameter $a$ is increased.

Poincar\'e sections are first used to visualize the transition from
regular to increasingly complex dynamics and to compare the
trigonometric system with the classical Nos\'e--Hoover oscillator in
the weakly perturbed regime. We then examine the bifurcation structure
as the parameter $a$ is varied. The observed transition to chaotic
dynamics is quantified through the Lyapunov spectrum, the
Kaplan--Yorke dimension, and the sum of the Lyapunov exponents, which
characterize chaotic instability and phase-space contraction.
Lyapunov maps over the initial-condition plane are subsequently used
to reveal the coexistence and spatial organization of regular,
chaotic, and contracting dynamics for representative parameter
values. Finally, the analysis is extended to the $(a,b)$ parameter
plane using the Lyapunov Integrability Test (LIT), together with
ensemble-based diagnostics obtained from the computed Lyapunov
spectra. To facilitate reproducibility of the numerical results
the implementation of LIT and its extensions used
throughout this analysis is publicly available -- details and access
information are provided in the Code Availability section.

\subsection{Integrable and near-integrable dynamics}
\label{sec:near_integrable_dynamics}

To obtain a first qualitative picture of the dynamics of the trigonometric
Nos\'e--Hoover oscillator, we consider the Poincar\'e sections shown in
Figs.~\ref{fig:integrable}--\ref{fig:NH_a01_comparison}. They were
computed for
\[
b=\frac{1}{2},\qquad a\in\{0,10^{-3},0.1\},
\]
using several distinct initial conditions distributed over the
fundamental domain. The sections were constructed from the intersections
of numerically integrated trajectories with the surface
\[
z=0\pmod{2\pi}.
\]
Both transverse crossing directions are retained. The resulting intersection points are represented in the $(x,y)$
plane, with both angular coordinates restricted to the fundamental
interval $[-\pi,\pi)$ for visualization.

\begin{figure}[t]
\centering
\includegraphics[width=0.8\linewidth]{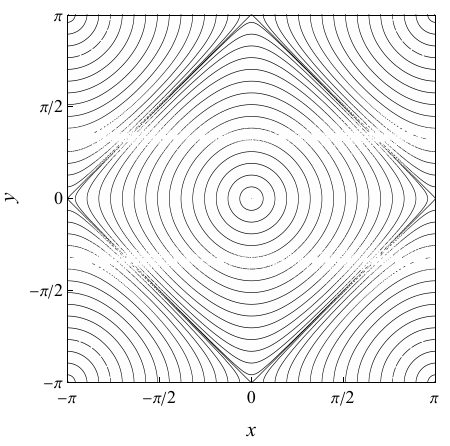}
\caption{
Poincar\'e section of the trigonometric Nos\'e--Hoover system~\eqref{eq:periodic_NH} for
$a=0$ and $b=1/2$. The section exhibits a completely regular
phase-space structure formed by invariant curves, with no visible
chaotic regions. This numerical picture is consistent with the
integrable character of the volume-preserving limit $a=0$.
}
\label{fig:integrable}
\end{figure}

The periodic orbit located at the center of the section belongs to a
family of particular solutions that can be obtained directly from the
equations of motion. Indeed, setting
\begin{equation*}
x=k\pi,
\qquad
y=s\pi,
\qquad
k,s\in\mathbb Z,
\label{eq:points}
\end{equation*}
gives the corresponding family of solutions on the torus
\begin{equation}
\label{eq:particular_periodic}
\vvarphi_{k,s}(t)
=
\left(
k\pi,\,
s\pi,\,
z_0+b\left[1-2(-1)^s\right]t
\right)
\pmod{2\pi}.
\end{equation}
For $s$ even, the motion in the $z$ direction has angular velocity
$-b$, whereas for $s$ odd it has angular velocity $3b$. In particular,
the orbit represented at the center of Fig.~\ref{fig:integrable}
corresponds to $k=s=0$ and is given by
\[
\boldsymbol{\varphi}_{0,0}(t)
=
(0,0,z_0-bt)
\pmod{2\pi}.
\]

This orbit is surrounded by nested invariant curves corresponding to
regular motion. The central region is bounded by a separatrix
 passing through the periodic solutions represented at
$(0,\pm\pi)$ and $(\pm\pi,0)$. The structures visible near the corners
of the square also belong to the family
$\vvarphi_{k,s}(t)$. Since opposite sides of the
fundamental domain are identified, these apparently distinct boundary
and corner structures must be understood as different representations
of the corresponding periodic orbits on the torus.

The completely regular organization of the section, together with the
absence of visible chaotic regions, is consistent with the integrable
character of the $a=0$ system established analytically in
Sec.~\ref{sec:integrability}. Moreover, the particular solutions
$\vvarphi_{k,s}(t)$ will be useful later in the
differential-Galois analysis of Sec.~\ref{subsec:galois_nonintegrability}, where the variational
equations along selected members of this family will be investigated.

The natural question is how the regular phase-space organization of the
integrable limit changes when the parameter $a$ is allowed to depart
slightly from zero. We therefore begin with small positive values of
$a$ and compare the resulting dynamics of the trigonometric model with that
of the classical modified Nos\'e--Hoover oscillator.
\begin{figure*}[t]
\centering
\subfigure[ClassicalNos\'e--Hoover system~\eqref{eq:modified_NH}]{
\includegraphics[width=0.4\linewidth]{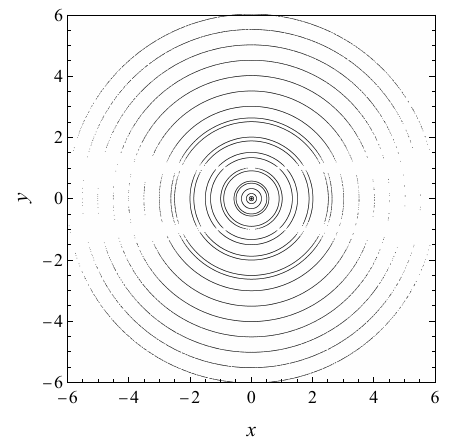}}
\hspace{1cm}
\subfigure[Trigonometric Nos\'e--Hoover system~\eqref{eq:periodic_NH}]{
\includegraphics[width=0.4\linewidth]{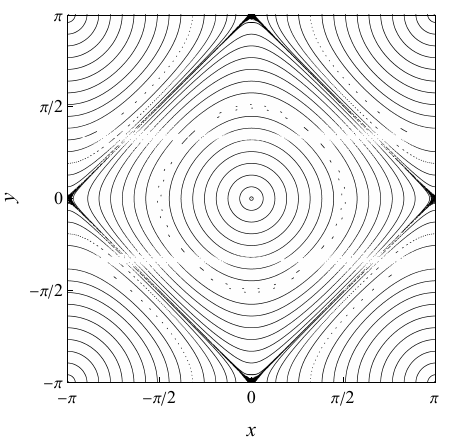}}
\caption{Poincar\'e sections of the classical and trigonometric Nos\'e--Hoover
systems for a weak thermostat coupling, $a=0.001$, with $b=1/2$.
For this small perturbation from the $a=0$ limit, both systems retain a predominantly regular phase-space structure, although thin chaotic layers are already visible in the trigonometric model. The comparison illustrates the close dynamical
correspondence between the classical and trigonometric models in the
weak-coupling regime.}
\label{fig:NH_weak_coupling}
\end{figure*}
Fig.~\ref{fig:NH_weak_coupling} shows the corresponding Poincar\'e
sections for $a=0.001$ and $b=1/2$. At this weak coupling, both systems
remain predominantly regular. In each case, the central periodic orbit
is surrounded by nested invariant curves corresponding to regular
quasiperiodic motion. There is, however, a clear difference in the
geometry of these curves. In the classical system,
Fig.~\ref{fig:NH_weak_coupling}(a), the invariant curves surrounding the
central orbit are nearly circular. This is consistent with the fact
that, for $a\ll1$, the classical Nos\'e--Hoover oscillator represents
only a weak perturbation of the harmonic oscillator.

The geometry of the trigonometric system is already noticeably different,
see Fig.~\ref{fig:NH_weak_coupling}(b). Close to the central periodic
orbit, the invariant curves remain smooth and regular, but they are
more elongated than in the classical case. Further away from the
center, they become increasingly deformed and gradually acquire the
pendulum-like geometry characteristic of the $a=0$ trigonometric system.
This structure is closely related to the nonlinear $(x,y)$ dynamics of
the integrable limit, whose solutions will be expressed in terms of
Jacobi elliptic functions in the next section.

The first traces of chaotic dynamics are also already visible in the
trigonometric model. They occur near the separatrix passing through
$(0,\pm\pi)$ and $(\pm\pi,0)$, where the smooth invariant curves
of the integrable limit begin to break up and are replaced by thin
layers of scattered points. This is precisely the region in which the
invariant curves of the integrable limit are expected to be most
sensitive to perturbations. Thus, even though the perturbation
$a=0.001$ is very small, the trigonometric model already develops
local chaotic layers, whereas no comparable chaotic region is visible
in the classical section at the same parameter values.

The difference between the two systems becomes much more pronounced
when the coupling is increased to $a=0.1$, while keeping $b=1/2$.
The corresponding sections are shown in
Fig.~\ref{fig:NH_a01_comparison}. The central parts of both sections
still exhibit a similar regular organization. In particular, a
$1\!:\!4$ resonant structure appears in both models, manifested by a
chain of four islands surrounding the central region.

\begin{figure*}[t]
\centering
\subfigure[Classical Nos\'e--Hoover system~\eqref{eq:modified_NH}]{
\includegraphics[width=0.4\linewidth]{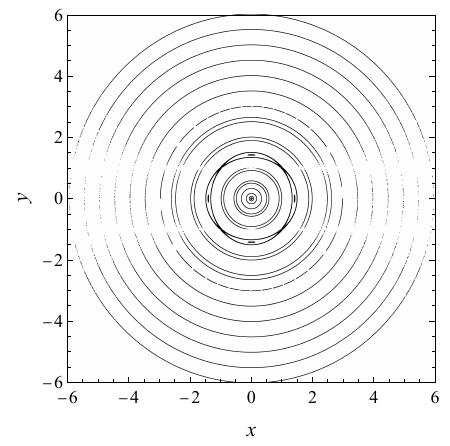}}
\hspace{1cm}
\subfigure[Trigonometric Nos\'e--Hoover system~\eqref{eq:periodic_NH}]{
\includegraphics[width=0.4\linewidth]{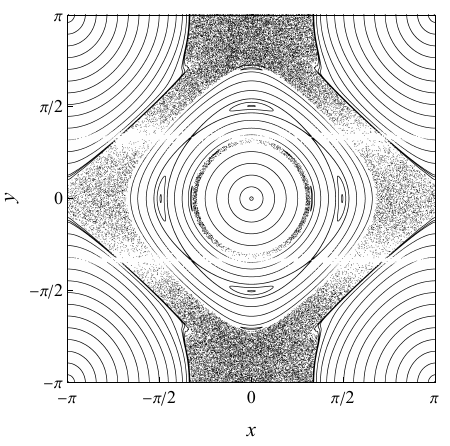}}
\caption{Poincar\'e sections of the classical and trigonometric Nos\'e--Hoover
systems for $a=0.1$ and $b=1/2$. Both systems exhibit a $1\!:\!4$
resonant structure, while the trigonometric system shows highly chaotic dynamics.}
\label{fig:NH_a01_comparison}
\end{figure*}

Away from this regular core, however, the dynamics of the trigonometric
system changes substantially. In Fig.~\ref{fig:NH_a01_comparison}(b),
large regions of the section are occupied by irregularly distributed
points, indicating chaotic motion. The invariant curves that confined
the dynamics in the near-integrable regime are progressively destroyed,
allowing trajectories to explore much larger portions of the accessible
phase space. In contrast, the classical Nos\'e--Hoover oscillator,
Fig.~\ref{fig:NH_a01_comparison}(a), retains a predominantly regular
structure for the same parameter values.

These observations show that the trigonometric modification is considerably
more sensitive to the perturbation parameter $a$. Although both models
remain close to their common regular limit near the central periodic
orbit, the surrounding invariant structures of the trigonometric oscillator
are destroyed much more rapidly. The extensive chaotic regions visible
already for $a=0.1$ therefore provide strong numerical evidence for the
loss of integrability in the trigonometric model.

\subsection{From conservative-like to dissipative-like dynamics}
\label{sec:conservative_dissipative}
\begin{figure*}[t]
\centering

\subfigure[Periodic]{
\includegraphics[width=0.31\linewidth]{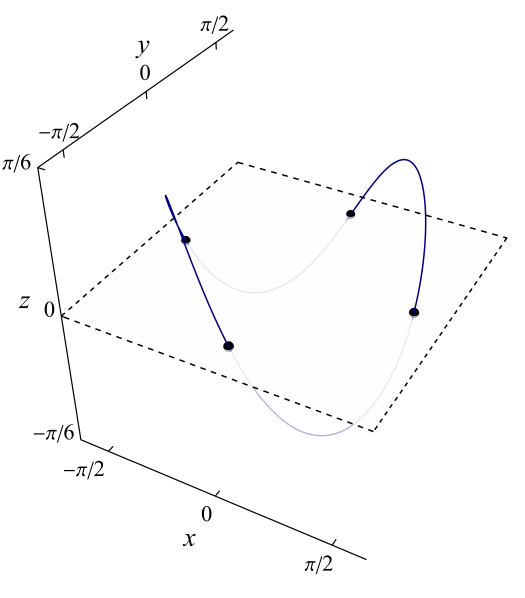}
}
\hfill
\subfigure[Quasi-periodic]{
\includegraphics[width=0.31\linewidth]{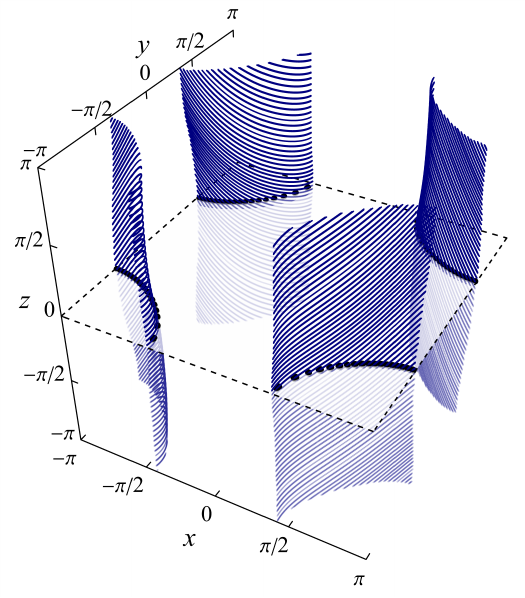}
}
\hfill
\subfigure[Chaotic]{
\includegraphics[width=0.31\linewidth]{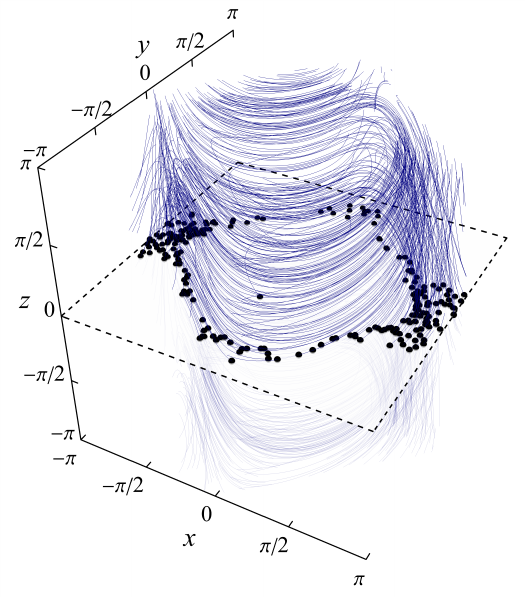}
}

\caption{(Color online)  Periodic, quasi-periodic, and chaotic phase trajectories of the
trigonometric Nos\'e--Hoover system for $a=0.1$ and $b=1/2$. The initial
conditions were selected from the Poincar\'e sections shown in
Fig.~\ref{fig:NH_a01_comparison}: (a) $\vx_0=(0,\pi/2,0)$,
(b) $\vx_0=(3\pi/4,\pi/2,0)$, and
(c) $\vx_0=(3\pi/4,0,0)$. Points of intersection with the
Poincar\'e section $z=0 \pmod{2\pi}$ are indicated by black dots.}
\label{fig:periodic_quasiperiodic_chaotic}
\end{figure*}
\begin{figure}[t]
\centering
\includegraphics[width=0.8\linewidth]{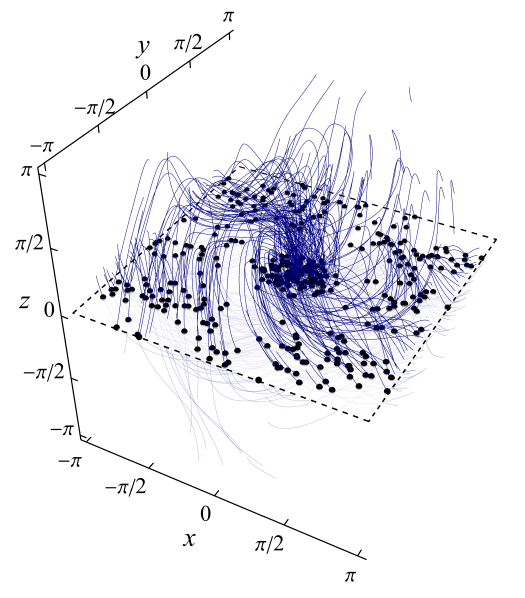}
\caption{(Color online)  Representative three-dimensional chaotic trajectory of the
trigonometric Nos\'e--Hoover oscillator for $a=2$ and $b=1/2$, with initial condition $\vx_0=(10^{-3},10^{-3},0)$, after
discarding the transient. Black dots indicate intersections with the
Poincar\'e section $z=0\pmod{2\pi}$. The chaotic motion extends over a
large part of the accessible phase space.}
\label{fig:trajectory_a3}
\end{figure}
The Poincar\'e sections discussed above show that, for small values of
$a$, regular and chaotic motion coexist in the phase space. Smooth
invariant curves and resonance islands are accompanied by localized
chaotic layers. To complement this two-dimensional picture, we now
examine representative trajectories in the full three-dimensional
phase space.

Fig.~\ref{fig:periodic_quasiperiodic_chaotic} shows three trajectories
of the trigonometric Nos\'e--Hoover oscillator for $a=0.1$ and $b=1/2$,
selected from qualitatively different regions of the Poincar\'e section
in Fig.~\ref{fig:NH_a01_comparison}(b). They represent periodic,
quasiperiodic, and chaotic motion, respectively. The periodic trajectory
forms a closed curve, whereas the quasiperiodic trajectory remains
confined to a smooth torus-like structure. The chaotic trajectory has a
more complicated geometry but remains restricted to a relatively
localized region of the phase space. The black points mark successive
intersections with the section $z=0\pmod{2\pi}$ and establish a direct
correspondence with the Poincar\'e geometry discussed above.

The character of the dynamics changes considerably as the coupling
parameter $a$ is increased. A representative trajectory for $a=2$ and
$b=1/2$, with initial condition 
chosen as
\begin{equation}
\label{eq:ini}
x_0=10^{-3},\qquad
y_0=10^{-3},\qquad
z_0=0,
\end{equation}
 is shown in
Fig.~\ref{fig:trajectory_a3}. After the initial transient has been
discarded, the trajectory develops a strongly folded and geometrically
irregular structure. Its intersections with the Poincar\'e section
spread over a large fraction of the accessible $(x,y)$ domain, in
contrast to the localized chaotic layers observed for small $a$.

The three-dimensional trajectories and Poincar\'e sections alone,
however, do not establish whether the observed dynamics is accompanied
by phase-space contraction. We therefore use the term
\emph{dissipative-like} only to describe the observed geometry. The
dissipative character of the dynamics will be quantified below through
the Lyapunov spectrum, in particular through the sum of the Lyapunov
exponents.

Before turning to these quantitative diagnostics, we examine how the
transition from regular to chaotic dynamics develops When $a$ is varied.
For this purpose, we construct a bifurcation diagram, which reveals the
onset of chaos as well as regular windows embedded in the chaotic
regime. One of these windows exhibits a clear period-doubling cascade,
analyzed in detail below.

\subsection{Bifurcation diagram and period-doubling route to chaos}
\label{sec:bifurcation}
\begin{figure*}[t]
\centering
\subfigure[Global view]{
\includegraphics[width=0.45\linewidth]
{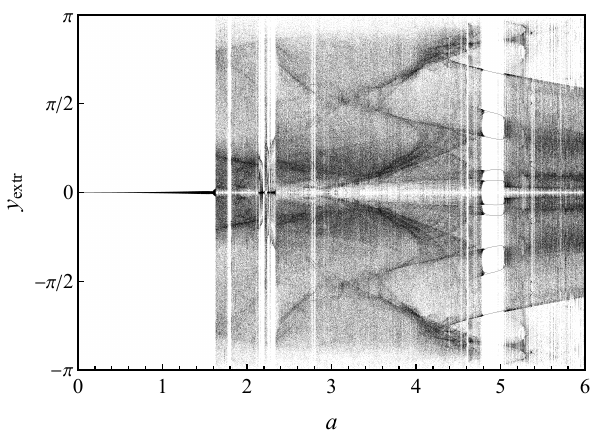}}
\hspace{1cm}
\subfigure[Magnification]{
\includegraphics[width=0.452\linewidth]
{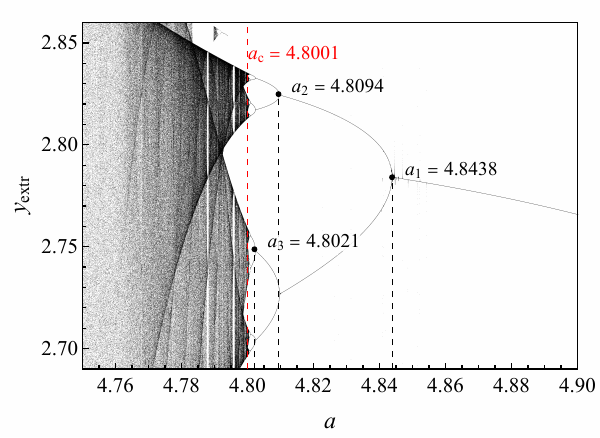}}
\caption{Bifurcation diagram of the trigonometric Nos\'e--Hoover system
versus $a$: (a) global view; (b) magnification of a regular window
showing a period-doubling route to chaos, with the successive
bifurcation values $a_1$, $a_2$, and $a_3$, and the estimated
accumulation point $a_c\simeq4.8001$.}
\label{fig:biff}
\end{figure*}

To follow the transition from regular to chaotic dynamics more
systematically, we construct a bifurcation diagram with $a$ as the
control parameter and $b$ fixed. The initial condition is chosen
according to~\eqref{eq:ini}.

After discarding the initial transient, we record the extremal values
$y_{\mathrm{extr}}$ satisfying
$
\dot y(t)=0
$
and plot them as a function of $a$.

The global bifurcation diagram is shown in
Fig.~\ref{fig:biff}(a). For small values of $a$, the asymptotic motion
is close to periodic and the extrema form a tight smooth branch. This regime
corresponds to the dynamics around the stable periodic solution $\boldsymbol{\varphi}_{0,0}(t)$ located
at the center of the Poincar\'e sections discussed above. When $a$ is
increased, this regular branch persists up to approximately
$a\simeq1.6$. Beyond this value, the bifurcation structure changes
    rapidly: the single periodic branch is replaced by a broad irregular
band extending over almost the entire accessible range
$y_{\mathrm{extr}}\in[-\pi,\pi)$. No sequence of successive period
doublings is visible at this first transition. Instead, the diagram
indicates a rapid transition from periodic to chaotic dynamics, which
will be confirmed below by the Lyapunov spectrum.

The chaotic regime is not uniform over the whole parameter interval.
For larger values of $a$, several regular windows appear within the
chaotic background. In these windows, the broad chaotic bands collapse
into a finite number of branches corresponding to periodic motion.
Thus, the global diagram reveals an alternating sequence of chaotic
regions and regular windows rather than a simple monotonic transition
from periodic to chaotic dynamics.

A particularly clear period-doubling cascade is visible within one of
these regular windows, as shown in the magnification in
Fig.~\ref{fig:biff}(b). When $a$ is decreased, a stable periodic branch
undergoes successive period-doubling bifurcations at approximately
\[
a_1=4.8438,\qquad
a_2=4.8094,\qquad
a_3=4.8021.
\]
The decreasing intervals between successive bifurcations suggest
convergence toward a finite accumulation point~$a_c$.

For a period-doubling cascade exhibiting asymptotic Feigenbaum
scaling, the ratios of the intervals between successive bifurcation
values approach the universal Feigenbaum constant
\cite{Feigenbaum1978,Feigenbaum1979},
\begin{equation}
\label{eq:feigenbaum_constant}
\delta
=
\lim_{n\to\infty}
\frac{a_{n-1}-a_n}{a_n-a_{n+1}}
=
4.6692016\ldots .
\end{equation}
For the three bifurcation values determined numerically above, the
corresponding ratio is
\[
\frac{a_1-a_2}{a_2-a_3}
=
\frac{4.8438-4.8094}{4.8094-4.8021}
\simeq 4.7123,
\]
which is close to the universal Feigenbaum value.

Since the cascade develops toward decreasing values of $a$,
the bifurcation sequence satisfies
\[
a_1>a_2>a_3>\cdots>a_c.
\]
Using the asymptotic Feigenbaum scaling, the accumulation point can
be estimated as
\begin{equation}
\label{eq:feigenbaum_accumulation}
a_c
\simeq
a_n-\frac{a_{n-1}-a_n}{\delta-1}.
\end{equation}
Using the last two numerically determined bifurcation values,
$a_2$ and $a_3$, gives
\[
a_c
\simeq
a_3-\frac{a_2-a_3}{\delta-1}
\simeq 4.8001.
\]
This critical value is marked in Fig.~\ref{fig:biff}(b) and agrees very
well with the accumulation of the bifurcation branches observed directly
in the diagram. As $a$ decreases toward $a_c$, successive period
doublings occur on progressively smaller parameter scales and eventually
give way to chaotic dynamics. The magnified diagram also reveals
smaller-scale branching patterns resembling the main bifurcation
structure, indicating a characteristic self-similar organization near
the transition to chaos. These observations provide further numerical
evidence that the transition follows the classical period-doubling
route to chaos.

The bifurcation diagram therefore provides a clear picture of how
periodic dynamics loses stability and gives way to chaos as the control
parameter is varied. In the next subsection, we quantify this transition
using the Lyapunov spectrum and compare the resulting instability
thresholds with the bifurcation structure found here.

\subsection{Lyapunov spectrum, phase-space contraction, and Kaplan--Yorke dimension}
\label{sec:lyapunov_analysis}

The bifurcation analysis presented above provides a geometrical picture
of the transition from regular to chaotic dynamics. To quantify this
transition, we now consider the complete Lyapunov spectrum. According to chaos theory, the Lyapunov
exponents measure the average exponential rates of growth or decay of
infinitesimal perturbations along a trajectory in phase space~\cite{Pikovsky:16::}. In
particular, a positive largest Lyapunov exponent provides a standard
numerical signature of chaotic dynamics.

Numerous algorithms for computing Lyapunov exponents have been developed
and refined over the years
\cite{Wolf:85::,Rangarajan:98::,Carbonell:02::,Lu:05::,Chen:06::,
Stachowiak:11::,Balcerzak:18::,Balcerzak:20::,Balcerzak:20b::}.
In the present work, we use the standard method of Benettin
\emph{et al.}~\cite{Benettin:80::}, based on the simultaneous integration
of the equations of motion and the corresponding variational equations,
combined with periodic Gram--Schmidt re-orthonormalization of the tangent
vectors. A convenient Mathematica implementation of this approach was
proposed by Sandri~\cite{Sandri:96::}.

Let $\Delta t$ denote the re-orthonormalization interval and let
$r_i^{(k)}$ be the norm of the $i$-th orthogonalized tangent vector at
the $k$-th re-orthonormalization step. After $N$ successive steps, the
Lyapunov exponents are computed as
\begin{equation}
\label{eq:benettin}
\lambda_i 
=
\frac{1}{N\Delta t}
\sum_{k=1}^{N}\ln r_i^{(k)},
\quad
t=N\Delta t,
\quad
i=1,\ldots,n.
\end{equation}
Periodic re-orthonormalization prevents the tangent vectors from aligning
with the most unstable direction and thus enables the computation of the
complete Lyapunov spectrum.

The equations of motion and the associated variational equations are
integrated using the built-in Mathematica \texttt{NDSolve} routine with
an explicit Runge--Kutta scheme and adaptive error control. A maximum
integration step size of $0.01$ is imposed throughout the computations.
The re-orthonormalization interval is fixed at $\Delta t=5$, and
$N=3000$ re-orthonormalization steps are performed, corresponding to the
total integration time
\[
t=N\Delta t=15000.
\]
These settings were selected after extensive numerical testing with
different integration times and re-orthonormalization intervals. The
choice $\Delta t=5$ and $N=3000$ was found to provide stable and accurate
Lyapunov spectra at a reasonable computational cost and was therefore
adopted throughout the numerical analysis.
Throughout this subsection, the Lyapunov exponents are ordered as
\[
\lambda_1\geq\lambda_2\geq\lambda_3.
\]
For a three-dimensional autonomous flow, one Lyapunov exponent is
expected to vanish along a bounded nonstationary trajectory, reflecting
the neutral direction tangent to the flow.

\begin{figure}[t]
\centering
\includegraphics[width=0.95\linewidth]{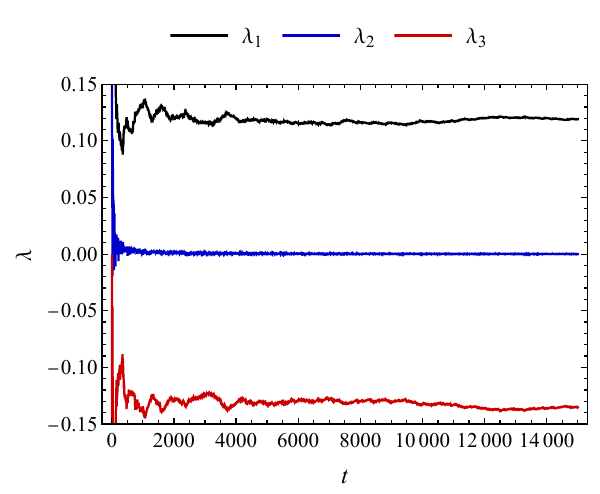}
\caption{(Color online)  Evolution of the complete Lyapunov spectrum
$\lambda_1\geq\lambda_2\geq\lambda_3$ along the chaotic trajectory
shown in Fig.~\ref{fig:trajectory_a3}, computed for $a=2$ and $b=1/2$,
with convergence toward $\lambda_1>0$, $\lambda_2\simeq0$, and
$\lambda_3<0$.}
\label{fig:lyapunov_spectrum_t}
\end{figure}\begin{figure*}[t]
\centering
\subfigure[Global view]{
\includegraphics[width=0.45\linewidth]{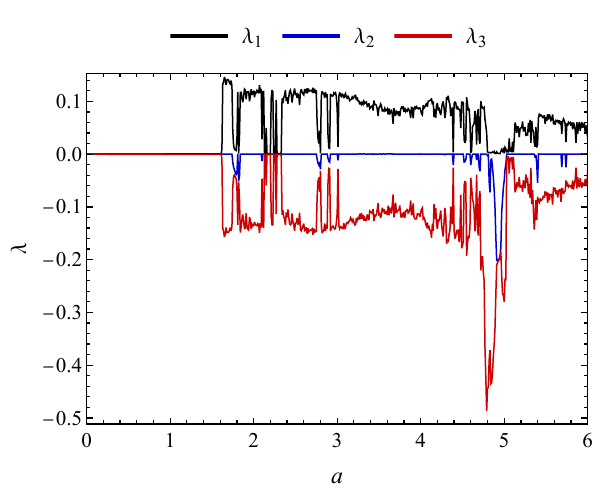}}\hspace{1cm}
\subfigure[Magnification]{
\includegraphics[width=0.45\linewidth]{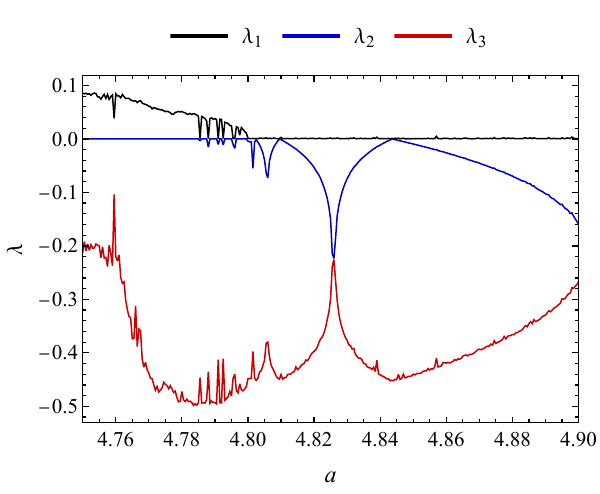}}
\caption{(Color online)  Lyapunov spectrum of the trigonometric Nos\'e–Hoover system
versus $a$ for fixed $b=1/2$: (a) global view; (b) magnification of the
regular window associated with the period-doubling route to chaos shown
in Fig.~\ref{fig:biff}. The curves represent the ordered
Lyapunov exponents $\lambda_1$, $\lambda_2$, and $\lambda_3$.}
\label{fig:lyapunov_spectrum_a}
\end{figure*}

As a representative example, Fig.~\ref{fig:lyapunov_spectrum_t}
shows the evolution of the complete Lyapunov spectrum computed along
the same trajectory as that shown in Fig.~\ref{fig:trajectory_a3},
for $a=2$, $b=1/2$, and the initial condition~\eqref{eq:ini}.
After an initial transient, the exponents gradually stabilize, reaching
the values
\[
\lambda_1\simeq0.1191,\qquad
\lambda_2\simeq3.8\times10^{-5},\qquad
\lambda_3\simeq-0.1352.
\]
Thus, the resulting spectrum has the characteristic form
$\lambda_1>0$, $\lambda_2\simeq0$, $\lambda_3<0$, providing clear
numerical evidence of chaotic dynamics for this trajectory.

Having verified the convergence of the Lyapunov exponents for a
representative chaotic trajectory, we now examine how the complete
spectrum changes with the control parameter $a$.
Fig.~\ref{fig:lyapunov_spectrum_a} shows the complete Lyapunov
spectrum computed for the initial condition~\eqref{eq:ini} as a
function of $a$, with $b=1/2$ fixed. The behavior of the spectrum
agrees remarkably well with the bifurcation diagram shown in
Fig.~\ref{fig:biff}. For small values of $a$, all three Lyapunov
exponents remain close to zero, consistently with the regular,
conservative-like dynamics observed in the Poincar\'e sections and
in the bifurcation diagram. Near $a\simeq1.6$, however, the spectrum
changes rapidly. The largest Lyapunov exponent becomes positive and
reaches values of approximately $\lambda_1\simeq0.15$, indicating the
onset of chaotic dynamics. At the same time, $\lambda_3$ becomes
negative, while $\lambda_2$ remains close to zero.
\begin{figure*}[t]
\centering
\subfigure[Kaplan--Yorke dimension]{
\includegraphics[width=0.44\linewidth]{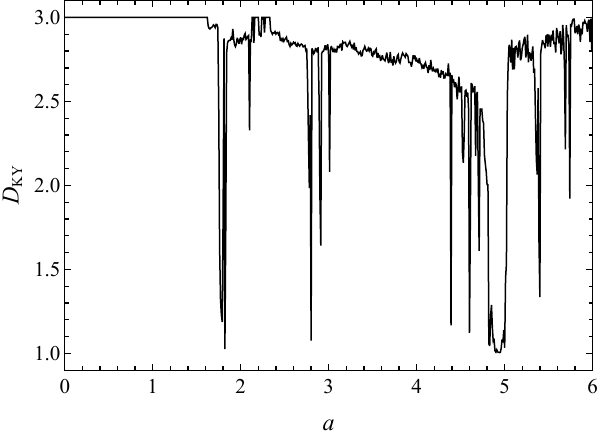}}
\hspace{1cm}
\subfigure[Sum of the Lyapunov exponents]{
\includegraphics[width=0.45\linewidth]{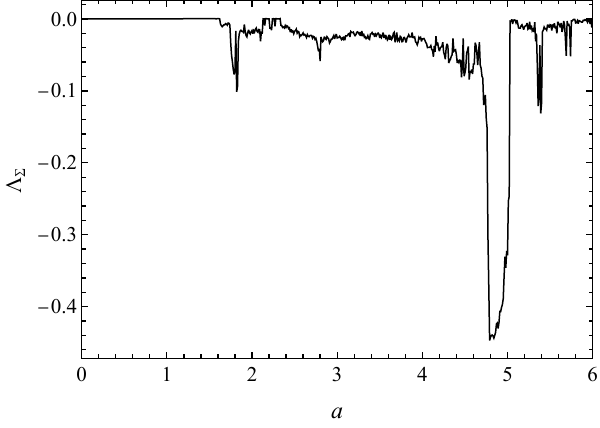}}
\caption{Dynamical characteristics derived from the Lyapunov spectrum
of the trigonometric Nos\'e--Hoover system as functions of the parameter $a$
for fixed $b=1/2$: (a) Kaplan--Yorke dimension $D_{\mathrm{KY}}$;
(b) sum of the Lyapunov exponents
$\mathrm{\Lambda}_{\Sigma}=\lambda_1+\lambda_2+\lambda_3$.}
\label{fig:KY_LambdaSigma}
\end{figure*}

The magnification in Fig.~\ref{fig:lyapunov_spectrum_a}(b) resolves
the regular window associated with the period-doubling cascade shown
in Fig.~\ref{fig:biff}(b). When $a$ is decreased through this window,
the Lyapunov spectrum undergoes the characteristic changes associated
with successive period doublings: the largest exponent approaches
zero at the bifurcation points, while the remaining transverse
exponent stays negative. Beyond the accumulation region, the largest
Lyapunov exponent becomes positive, consistently with the transition
to chaos observed in the bifurcation diagram.

For larger values of $a$, the Lyapunov spectrum exhibits numerous
narrow intervals in which $\lambda_1$ approaches zero. These intervals
correspond closely to the regular windows visible in the global
bifurcation diagram. In particular, the pronounced regular window near
$a\simeq4.8$, in which the period-doubling cascade discussed above
takes place, is clearly reflected in the Lyapunov spectrum. Thus, the
Lyapunov analysis independently confirms the alternation between
chaotic and periodic regimes revealed by the bifurcation structure.

Having determined the complete Lyapunov spectrum, we can also evaluate
the Kaplan--Yorke dimension, also referred to as the Lyapunov
dimension~\cite{Kaplan:79::}. It is defined as
\begin{equation}
\label{eq:kaplan_yorke}
D_{\mathrm{KY}}
=
j+
\frac{1}{|\lambda_{j+1}|}\displaystyle\sum_{i=1}^{j}\lambda_i,
\end{equation}
where $j$ is the largest integer such that
$
\sum_{i=1}^{j}\lambda_i\geq0.
$

Fig.~\ref{fig:KY_LambdaSigma}(a) shows $D_{\mathrm{KY}}$ as a
function of $a$. For $a\lesssim1.6$, the Kaplan--Yorke dimension
computed from the Lyapunov spectrum remains close to
\[
D_{\mathrm{KY}}=3.
\]
This reflects the nearly vanishing sum of the computed Lyapunov
exponents and is consistent with the predominantly regular dynamics
observed in the Poincar\'e sections. Immediately after the transition
near $a\simeq1.6$, $D_{\mathrm{KY}}$ decreases below three as chaotic
dynamics and phase-space contraction develop. For larger values of
$a$, the dimension varies considerably and assumes noninteger values
over broad parameter intervals, reflecting changes in the effective
dimensionality of the chaotic dynamics.

The regular window containing the period-doubling cascade is
particularly well resolved in this representation. Near
$a\simeq4.8$, the Kaplan--Yorke dimension approaches
\[
D_{\mathrm{KY}}\simeq1,
\]
consistently with the presence of a stable attracting periodic orbit.
As chaotic dynamics reappears beyond the regular window,
$D_{\mathrm{KY}}$ increases again and assumes noninteger values
between one and three. The Kaplan--Yorke dimension therefore provides
an additional quantitative characterization of the transitions
between periodic and chaotic dynamics already observed in the
bifurcation diagram and in the Lyapunov spectrum.

To complement the largest Lyapunov exponent, we also consider
$\mathrm{\Lambda}_{\Sigma}$ defined in
Eq.~\eqref{eq:Lyapunov_sum_divergence}. This quantity characterizes the
asymptotic mean phase-space contraction or expansion along a trajectory.
Thus, $\mathrm{\Lambda}_{\Sigma}\simeq0$ indicates no appreciable
average phase-space contraction, whereas
$\mathrm{\Lambda}_{\Sigma}<0$ indicates net contraction.

As shown in Eq.~\eqref{eq:divergence_periodic_NH}, the divergence of
the vector field is not sign-definite. Therefore, the sign of the
instantaneous divergence alone does not determine the asymptotic
behavior; instead, the long-time average along a trajectory determines
$\mathrm{\Lambda}_{\Sigma}$ through
Eq.~\eqref{eq:Lyapunov_sum_divergence}.

The dependence of $\mathrm{\Lambda}_{\Sigma}$ on $a$ is shown in
Fig.~\ref{fig:KY_LambdaSigma}(b). For $a\lesssim1.6$, the sum remains
very close to zero. This is consistent with the conservative-like
geometry observed in the Poincar\'e sections for small $a$, where
regular invariant structures dominate and no appreciable average
phase-space contraction is detected.

The situation changes     rapidly near $a\simeq1.6$. In the same
parameter region in which the largest Lyapunov exponent becomes
positive, $\mathrm{\Lambda}_{\Sigma}$ becomes negative, showing that the onset
of chaos is accompanied by average phase-space contraction. For larger
values of $a$, $\mathrm{\Lambda}_{\Sigma}$ remains predominantly negative,
although its magnitude varies strongly and approaches zero within some
regular parameter intervals.

The strongest contraction is observed in the vicinity of the regular
window near $a\simeq4.8$, where $\mathrm{\Lambda}_{\Sigma}$ reaches its most
negative values. This is precisely the parameter region in which the
period-doubling route to chaos was identified in
Fig.~\ref{fig:biff}(b). The combination of the bifurcation diagram,
the Lyapunov spectrum, the Kaplan--Yorke dimension, and
$\mathrm{\Lambda}_{\Sigma}$ therefore provides a consistent picture of the
transition from conservative-like regular dynamics to chaotic,
phase-space-contracting behavior.

\subsection{Lyapunov maps and phase-space contraction}
\label{sec:lyapunov_maps}

The analysis presented above characterizes the dynamics along a single
trajectory as the control parameter $a$ is varied. A complementary
picture can be obtained by fixing the parameters and computing the
Lyapunov spectrum over a two-dimensional set of initial conditions.
This allows us to determine how regularity, chaos, and phase-space
contraction are distributed throughout the phase space and how this
distribution changes with the coupling parameter $a$.

For this purpose, we construct Lyapunov maps on the same
two-dimensional domain used for the Poincar\'e sections. The initial
conditions $(x_0,y_0)$ are sampled on a uniform grid in
\[
(x_0,y_0)\in[-\pi,\pi)\times[-\pi,\pi),\qquad \text{with}\qquad z_0=0.
\]
 For each
initial condition, the complete Lyapunov spectrum
$\{\lambda_1,\lambda_2,\lambda_3\}$ is computed after discarding the
initial transient. In this way, every point of the initial-condition
plane is assigned a quantitative measure of the corresponding
asymptotic dynamics.

We first consider maps of the largest Lyapunov exponent
$\lambda_1(x_0,y_0)$. Positive values of $\lambda_1$ are associated
with chaotic motion, whereas values close to zero correspond to regular
dynamics. The resulting maps can therefore be compared directly with
the corresponding Poincar\'e sections: invariant curves and resonance
islands correspond to regions where $\lambda_1$ remains close to zero,
while chaotic regions are characterized by positive values of
$\lambda_1$.

To improve the contrast of the Lyapunov maps, the remaining values are
displayed using the nonlinear color transformation
\begin{equation}
\label{eq:lyapunov_map_scaling}
\mathcal{S}(\lambda)
:=
\left(
\frac{\lambda-\lambda_{\min}}
{\lambda_{\max}-\lambda_{\min}}
\right)^{p},
\qquad p>0,
\end{equation}
where $\lambda_{\min}$ and $\lambda_{\max}$ denote the lower and upper
limits of the displayed range. The case $p=1$ corresponds to a linear
color scale. For $p>1$, the transformation compresses the lower part
of the normalized range and enhances the visual separation of larger
Lyapunov exponents, making structures within strongly chaotic regions
more pronounced. In contrast, for $0<p<1$, the transformation expands
the lower part of the range and is therefore useful for revealing
variations among small positive Lyapunov exponents associated with
weakly chaotic dynamics. The parameter $p$ thus controls which part
of the Lyapunov range is visually emphasized. This transformation
affects only the visualization and does not modify the computed
Lyapunov exponents or their numerical classification.

Figs.~\ref{fig:map1}--\ref{fig:map2} show the  representative
Poincar\'e sections with the corresponding maps of the largest Lyapunov
exponent for fixed $b=1/2$ and increasing values of $a$. The clear
correspondence between the two representations is immediately visible.
The Poincar\'e sections provide the geometrical organization of the
dynamics, whereas the Lyapunov maps quantify the stability of the
trajectories associated with different initial conditions.
\begin{figure*}[htp]
\centering
\subfigure[$a=1/2, \quad b=1/2$]{
\includegraphics[width=0.42\linewidth]{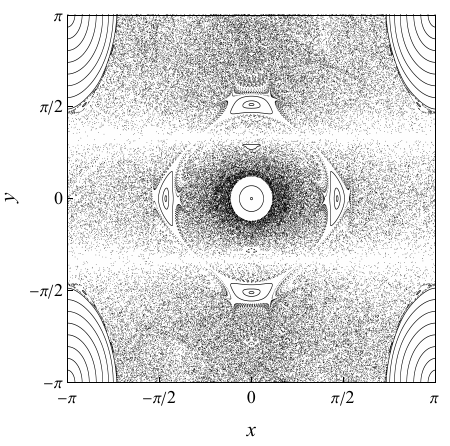}\hspace{0.5cm}
\includegraphics[width=0.55\linewidth]{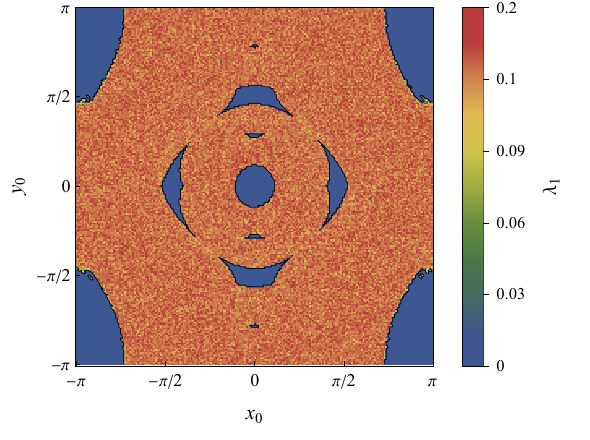}}
\subfigure[$a=2, \quad b=1/2$]{
\includegraphics[width=0.42\linewidth]{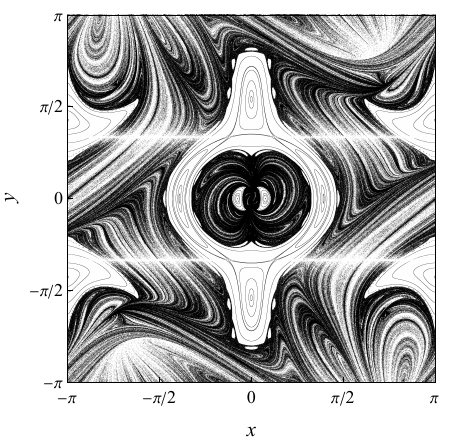}\hspace{0.5cm}
\includegraphics[width=0.55\linewidth]{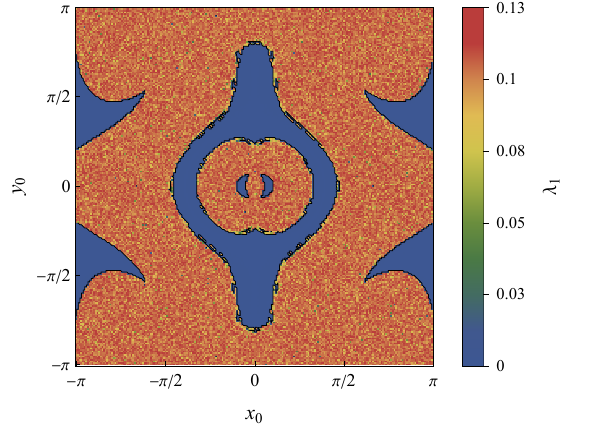}}
\caption{(Color online)  Poincar\'e sections (left) and the corresponding maps of the
largestLyapunov exponent $\lambda_1$ (right) for the trigonometric
Nos\'e--Hoover oscillator with $b=1/2$ and selected values of the
parameters. The color scale is defined by Eq.~\eqref{eq:lyapunov_map_scaling}
with $p=1$. Blue regions correspond to $\lambda_1$ below the adopted
chaos threshold, whereas positive values of $\lambda_1$ indicate
chaotic dynamics. \label{fig:map1}}
\end{figure*}
\begin{figure*}[htp]
\centering
\subfigure[$a=3, \quad b=1/2$]{
\includegraphics[width=0.42\linewidth]{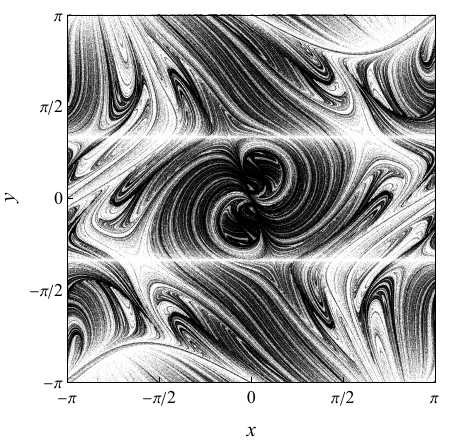}\hspace{0.5cm}
\includegraphics[width=0.55\linewidth]{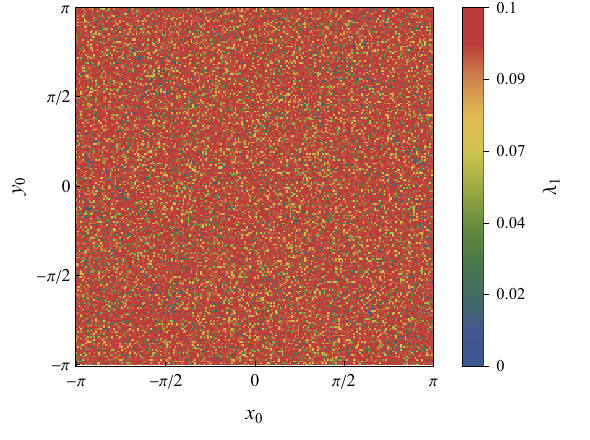}}
\subfigure[$a=6, \quad b=1/2$]{
\includegraphics[width=0.42\linewidth]{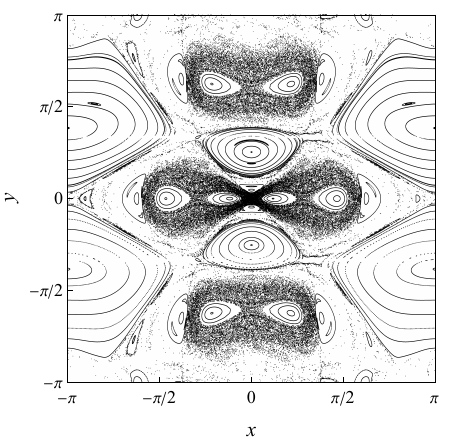}\hspace{0.5cm}
\includegraphics[width=0.55\linewidth]{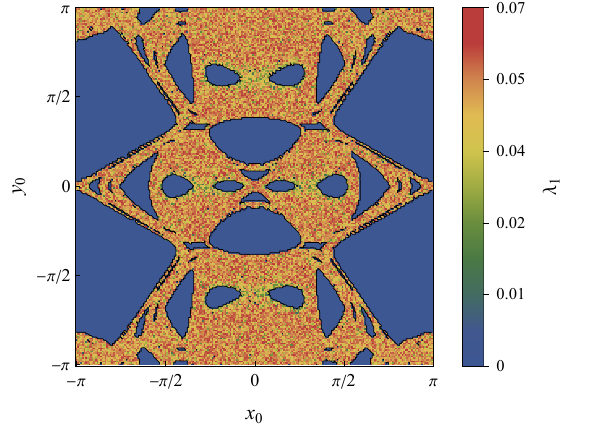}}
\caption{(Color online)  Poincar\'e sections (left) and the corresponding maps of the
largestLyapunov exponent $\lambda_1$ (right) for the trigonometric
Nos\'e--Hoover oscillator with $b=1/2$ and selected values of the
parameters. The color scale is defined by
Eq.~\eqref{eq:lyapunov_map_scaling} with $p=1$. Blue regions correspond
to $\lambda_1$ below the adopted chaos threshold, whereas positive
values of $\lambda_1$ indicate chaotic dynamics.\label{fig:map2}}
\end{figure*}
For $a=1/2$, the chaotic component of the phase space is considerably
larger than for $a=0.1$ shown in Fig.~\ref{fig:NH_a01_comparison}.
Most of the accessible region is now occupied by chaotic trajectories,
while regular motion survives mainly in a small central region and near
the outer boundaries. This picture is consistent with the Lyapunov map,
where positive values of $\lambda_1$ cover most of the phase space and
regions with $\lambda_1\simeq0$ are restricted to these remaining
regular domains. Despite the substantial increase in chaotic motion,
the phase space still retains a conservative-like structure, with
regular islands embedded in and surrounded by chaotic regions.

When $a$ is increased further to $a=2$, the Poincar\'e section becomes
increasingly dominated by chaotic dynamics. The remaining regular
structures are substantially reduced, while irregular trajectories
occupy most of the accessible region. The chaotic part of the section
develops intricate, strongly folded and filamentary patterns, indicating
a further loss of the conservative-like organization observed for smaller
values of $a$.

It is important to emphasize that the dense, folded patterns visible
in the chaotic regions are generated by a single initial condition.
A single chaotic trajectory repeatedly intersects the Poincar\'e plane
at different locations and, after sufficiently long integration,
produces the extended irregular structures observed in the section.
This is fundamentally different from the regular part of the section,
where each family of invariant curves is obtained from a separate
initial condition and represents a trajectory confined to its own
invariant torus.

The corresponding Lyapunov map reproduces this dynamical organization
remarkably well. The regular invariant tori are represented by extended
blue regions with $\lambda_1\simeq0$, whereas the folded and irregular
structures observed in the Poincar\'e section correspond to regions
with positive values of $\lambda_1$. In particular, the large chaotic
domains are characterized by relatively high values of the largest
Lyapunov exponent, confirming their strongly chaotic character.

For $a=3$, the phase-space organization changes qualitatively. The
regular structures that were still clearly visible for $a=2$ have
disappeared, and the Poincar\'e section is now dominated entirely by
chaotic dynamics. The smooth invariant curves and regular islands are
replaced by strongly distorted and folded patterns extending throughout
the section plane. In particular, a pronounced attractor-like structure
with a complex folded geometry develops around the central region of
the section.

The corresponding Lyapunov map provides a clear quantitative
confirmation of this picture. The largest Lyapunov exponent is positive
over essentially the entire $(x_0,y_0)$ plane: initial conditions
sampled on the dense grid lead to positive values of $\lambda_1$,
with values reaching approximately $\lambda_1\simeq0.1$. In contrast
to the case $a=2$, no extended regions with $\lambda_1\simeq0$ remain
visible. Thus, the disappearance of the regular invariant structures
from the Poincar\'e section is accompanied by the spreading of chaotic
dynamics over essentially the whole sampled set of initial conditions.
This result is fully consistent with the parameter-dependent Lyapunov
spectrum discussed above, which identifies this parameter range as a
strongly chaotic regime.

Interestingly, this trend does not continue monotonically When $a$ is
increased. For instance, for larger $a=6$, a rich organization of regular domains reappears.
The Poincar\'e section again contains numerous invariant curves and
island-like structures embedded in chaotic regions, and these structures
are accurately reproduced by the corresponding domains with
$\lambda_1\simeq0$ in the Lyapunov map. Thus, increasing the coupling
parameter does not simply produce progressively stronger chaos.
Instead, the system undergoes a non-monotonic reorganization in which
strongly chaotic regimes alternate with parameter ranges containing
substantial regular structures.

\begin{figure*}[htp]
\centering
\subfigure[$a=1/2, \quad b=1/2$]{
\includegraphics[width=0.32\linewidth]{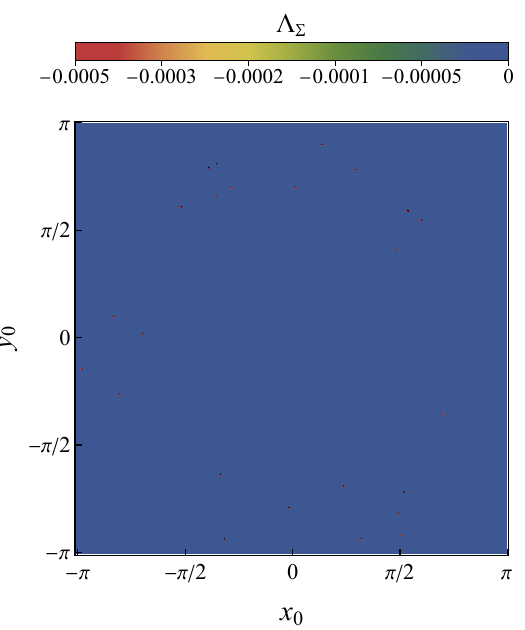}}
\subfigure[$a=2, \quad b=1/2$]{
\includegraphics[width=0.32\linewidth]{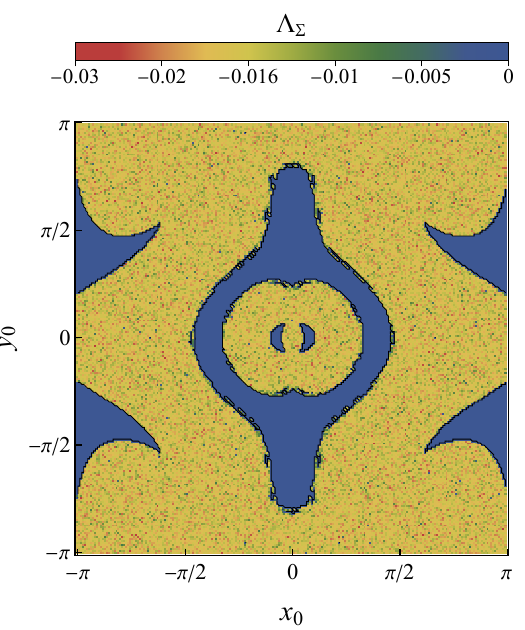}}
\subfigure[$a=3, \quad b=1/2$]{
\includegraphics[width=0.32\linewidth]{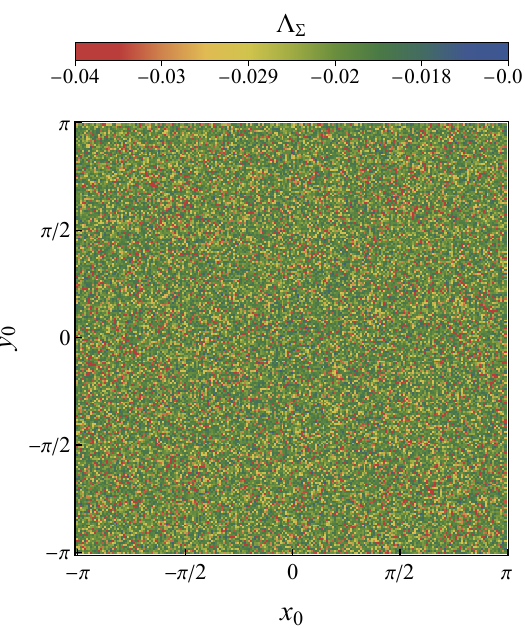}}
\caption{(Color online)  Maps of the Lyapunov-sum quantity $\Lambda_\Sigma$, defined
by Eq.~\eqref{eq:Lyapunov_sum_divergence}, characterizing the asymptotic mean phase-space contraction
rate, for $b=1/2$ and selected values of the parameters corresponding
to the previous figures. The color scale is defined by
Eq.~\eqref{eq:lyapunov_map_scaling} with $p=1$.
\label{fig:lambda_sigma_maps}}
\end{figure*}
The largest Lyapunov exponent characterizes local exponential
instability, but it does not by itself distinguish between chaotic
motion with and without average phase-space contraction. To resolve
this additional property over the initial-condition plane, we also
construct maps of the sum $\mathrm{\Lambda}_{\Sigma}(x_0,y_0)$ of the complete Lyapunov spectrum  over a dense grid of initial conditions $(x_0,y_0)$,
As discussed above, this quantity is related to the long-time averaged
divergence along the corresponding trajectory. Hence,
$\mathrm{\Lambda}_{\Sigma}\simeq0$ identifies initial conditions for which no
appreciable average phase-space contraction is detected, whereas
$\mathrm{\Lambda}_{\Sigma}<0$ indicates contracting dynamics.

Fig.~\ref{fig:lambda_sigma_maps} shows representative maps of
$\mathrm{\Lambda}_{\Sigma}$. For $a=1/2$, the sum remains extremely
close to zero over almost the entire initial-condition plane. This
provides a quantitative confirmation of the conservative-like character
of the phase-space organization observed in the corresponding
Poincar\'e section.

At $a=2$, a pronounced spatial structure develops. Regions with
$\mathrm{\Lambda}_{\Sigma}\simeq0$ coexist with domains in which
$\mathrm{\Lambda}_{\Sigma}<0$. Remarkably, their boundaries closely reproduce
the structures identified independently in the Poincar\'e section and
in the map of $\lambda_1$. The initial-condition plane can therefore
be separated not only into regular and chaotic regions, but also into
regions characterized by markedly different rates of average
phase-space contraction.

For $a=3$, the picture changes again. Negative values of
$\mathrm{\Lambda}_{\Sigma}$ extend over essentially the whole sampled domain,
showing that the widespread chaotic dynamics identified by the
$\lambda_1$ map is accompanied by phase-space contraction. This
provides a quantitative counterpart of the dissipative-like geometry
observed previously in the three-dimensional trajectories and
Poincar\'e sections.

Taken together, the Poincar\'e sections and the two complementary
Lyapunov maps provide a detailed picture of the phase-space
reorganization induced by the parameter $a$. The maps of $\lambda_1$
identify the distribution of chaotic instability, whereas the maps of
$\Lambda_{\Sigma}$ reveal the corresponding distribution of average
phase-space contraction. Their combination thus characterizes the
transition from conservative-like to dissipative-like dynamics both
geometrically and quantitatively.

As mentioned, this transition is clearly non-monotonic. Increasing $a$
initially leads to the progressive destruction of regular structures
and to widespread chaotic, contracting dynamics, whereas organized
regular domains reappear for larger values of the parameter. This
behavior motivates a systematic exploration of the full $(a,b)$
parameter plane.

\subsection{Global organization of the parameter space: LIT analysis}
\label{sec:LIT_analysis}

We now employ the Lyapunov Integrability Test
(LIT)~\cite{SzuminskiLIT2026} to extend the preceding analysis from
selected parameter values and one-parameter families to a sampled
two-parameter domain.

The LIT was originally introduced as an ensemble-based numerical
procedure for identifying parameter values that may admit additional
first integrals. Its principal indicator is the maximum of the largest
Lyapunov exponent over a prescribed set of initial conditions.
Parameter values for which this quantity remains numerically close to
zero over the entire sampled ensemble are selected as candidates for
further analytical investigation of regularity or integrability.
Since the LIT computes the complete Lyapunov spectrum for every sampled
initial condition, the same numerical data can also provide additional
information about the global organization of the dynamics. Here, besides
the principal LIT indicator, we consider the percentage of chaotic
initial conditions  and two measures of phase-space
contraction based on the sum of the Lyapunov exponents.

\subsubsection{Maximum largest Lyapunov exponent $\lambda_{1,\max}$}
\label{sec:LIT_lambda1max}

\begin{figure*}[htp]
\centering
\includegraphics[width=0.49\linewidth]{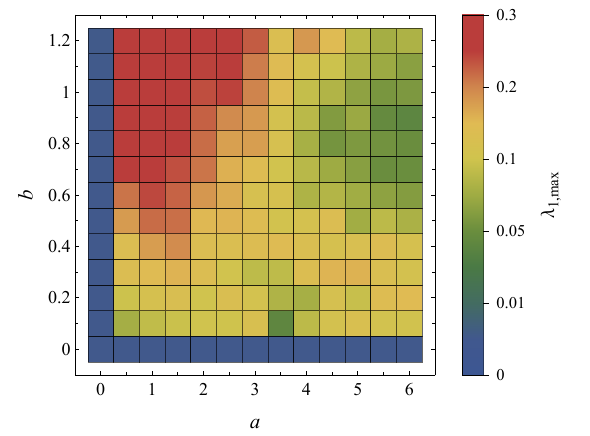}
\includegraphics[width=0.49\linewidth]{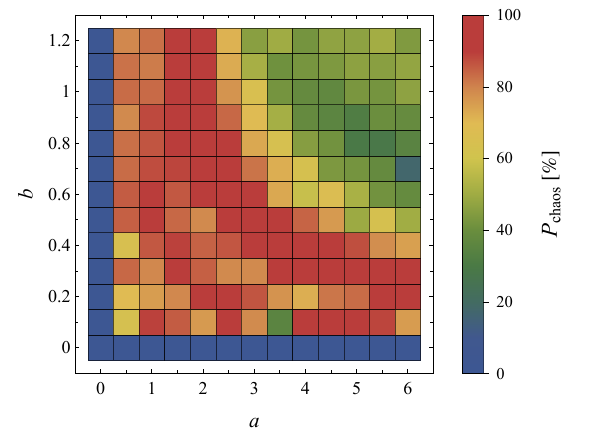}
\caption{(Color online) 
LIT analysis of system~\eqref{eq:periodic_NH} in the parameter plane
$(a,b)$. Left: maximum largest Lyapunov exponent
$\lambda_{1,\max}$ computed over the considered set of initial
conditions. Right: percentage $P_{\mathrm{chaos}}$ of initial
conditions classified as chaotic. The color transformation
in~\eqref{eq:lyapunov_map_scaling} is applied with $p=0.5$ and $p=1$
in the left and right panels, respectively.
}
\label{fig:LIT_parameter_map}
\end{figure*}

Following the LIT procedure, we introduce a uniform grid of the scanning
parameters over the range
\[
\valpha=(a,b)\in[0,6]\times[0,1.2],
\]
with step sizes
\[
h_a=\frac{1}{2},
\qquad
h_b=0.1,
\]
respectively. This defines the discrete parameter set
$P_{h_{\valpha}}$. For each parameter pair
$(a,b)\in P_{h_{\valpha}}$, we consider a uniform grid of initial
conditions on the plane $z_0=0$,
\[
I_{h_{\vx}}
=
\left\{
\vx_0=(x_0,y_0,0):
(x_0,y_0)\in[-\pi,\pi)\times[-\pi,\pi)
\right\},
\]
with spacing $h_{\vx}=0.2$.

For every $\vx_0\in I_{h_{\vx}}$, we compute the complete Lyapunov
spectrum
\[
\lambda_1(\vx_0;a,b)
\geq
\lambda_2(\vx_0;a,b)
\geq
\lambda_3(\vx_0;a,b).
\]
The principal LIT indicator is the maximum of the largest Lyapunov
exponent over the sampled initial conditions,
\begin{equation}
\label{eq:LIT_lambda1max}
\lambda_{1,\max}(a,b)
=
\max_{\vx_0\in I_{h_{\vx}}}
\lambda_1(\vx_0;a,b).
\end{equation}

If $\lambda_{1,\max}(a,b)>0$, at least one trajectory within the sampled
region exhibits chaotic dynamics. The corresponding parameter pair can
therefore be excluded as a candidate for globally regular or integrable
dynamics within the sampled domain. Conversely, parameter regions for
which $\lambda_{1,\max}(a,b)$ remains numerically close to zero over the
entire initial-condition grid are identified by the LIT as candidates
for additional regular or integrable cases.

Fig.~\ref{fig:LIT_parameter_map}(a) shows the map
\[
(a,b)\longmapsto\lambda_{1,\max}(a,b)
\]
for the trigonometric Nos\'e--Hoover oscillator~\eqref{eq:periodic_NH}.
The color scale represents the magnitude of $\lambda_{1,\max}$, with
the square-root transformation defined in
Eq.~\eqref{eq:lyapunov_map_scaling} applied to enhance the visibility
of weakly chaotic regions.

Two distinct lines with $\lambda_{1,\max}\simeq0$ are clearly visible,
namely $a=0$ and $b=0$. These coincide with the parameter cases
identified above as candidates for the presence of additional first
integrals. Their integrability properties will be examined analytically
in Sec.~\ref{sec:integrability}.

At every sampled parameter pair with $a>0$ and $b>0$, a positive value
of $\lambda_{1,\max}$ is detected. Its magnitude, however, varies
considerably across the $(a,b)$ plane, revealing a strongly nonuniform
distribution of chaotic instability. Relatively large values occur over
broad regions of the parameter plane, whereas other regions exhibit
much weaker chaotic instability.

\subsubsection{Percentage of chaotic initial conditions $P_{\rm chaos}$}
\label{sec:Pchaos}

Although $\lambda_{1,\max}$ detects the presence of chaos within the
sampled ensemble and quantifies the strongest detected instability, it
does not indicate how widespread chaotic dynamics is. A positive value
of $\lambda_{1,\max}$ may result either from chaos occurring for almost
all sampled initial conditions or from only a small subset of them.
To distinguish between these situations, we introduce the percentage
of chaotic initial conditions,
\begin{equation}
\label{eq:Pchaos}
P_{\mathrm{chaos}}(a,b)
=
100\,
\frac{\left|I_{h_{\vx}}^{\mathrm{chaos}}(a,b)\right|}
     {\left|I_{h_{\vx}}\right|}\%.
\end{equation}
Here
\[
I_{h_{\vx}}^{\mathrm{chaos}}(a,b)
=
\left\{
\vx_0\in I_{h_{\vx}}:
\lambda_1(\vx_0;a,b)>\lambda_{\mathrm{thr}}
\right\}
\]
denotes the subset of initial conditions classified as chaotic, with
the numerical threshold
\[
\lambda_{\mathrm{thr}}=5\times10^{-4}.
\]
The threshold accounts for small positive finite-time values of
$\lambda_1$ that may also occur for regular trajectories. Thus,
$P_{\mathrm{chaos}}$ ranges from $0$ to $100\%$, with values close to
$0\%$ corresponding to predominantly regular dynamics and values close
to $100\%$ indicating that chaotic trajectories dominate the sampled
initial-condition plane.

Fig.~\ref{fig:LIT_parameter_map}(b) shows the corresponding map
\[
(a,b)\longmapsto P_{\mathrm{chaos}}(a,b).
\]
The resulting distribution is consistent with the dynamical picture
obtained from the Poincar\'e sections and the Lyapunov maps discussed
above. Along the two parameter lines $a=0$ and $b=0$, the chaotic
percentage vanishes, $P_{\mathrm{chaos}}=0$, in agreement with the
corresponding near-zero values of $\lambda_{1,\max}$. Thus, both
ensemble-based indicators single out the same parameter families as
candidates for the presence of additional first integrals.

A markedly different picture emerges in the interior of the parameter
domain, where $a>0$ and $b>0$. Chaotic initial conditions are detected
for every investigated parameter pair. Moreover, chaos is generally
not confined to a small subset of the sampled initial-condition plane.
Approximately $70\%$ of the investigated parameter pairs satisfy
$P_{\mathrm{chaos}}>60\%$, while about $50\%$ satisfy
$P_{\mathrm{chaos}}>80\%$. The mean value of $P_{\mathrm{chaos}}$ over
this part of the parameter plane is approximately $70\%$, with a median
close to $80\%$. Thus, for a substantial part of the parameter space,
the majority of the sampled initial conditions lead to chaotic
trajectories.

Parameter pairs for which the entire sampled initial-condition grid is
classified as chaotic are less common. The condition
$P_{\mathrm{chaos}}=100\%$ occurs for approximately $6\%$ of the
investigated parameter pairs with $a>0$ and $b>0$. Nevertheless, the
overall distribution shows that chaotic dynamics is not only detected
throughout the interior of the investigated parameter domain but is
also widespread over the sampled initial-condition plane for a large
fraction of parameter values.

\subsubsection{Phase-space contraction:
$\mathrm{\Lambda}_{\Sigma,\min}$ and
$\mathrm{\Lambda}_{\Sigma,\mathrm{mean}}$}
\label{sec:LIT_contraction}

\begin{figure*}[htp]
\centering
\includegraphics[width=0.49\linewidth]{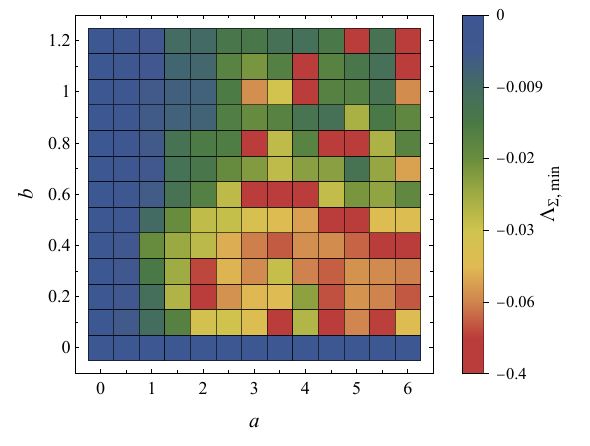}
\includegraphics[width=0.49\linewidth]{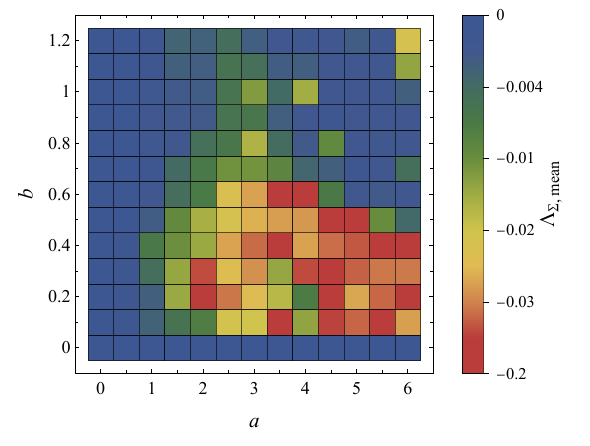}
\caption{(Color online) 
LIT maps of the sum of the Lyapunov exponents
$\mathrm{\Lambda}_{\Sigma}=\sum_i\lambda_i$ in the parameter plane
$(a,b)$. Left: minimum value
$\mathrm{\Lambda}_{\Sigma,\min}$; right: mean value
$\mathrm{\Lambda}_{\Sigma,\mathrm{mean}}$, both evaluated over the set
of initial conditions on the considered Poincar\'e section for each
pair $(a,b)$. The color scale is defined by
Eq.~\eqref{eq:lyapunov_map_scaling} with $p=10$.
}
\label{fig:LIT_Lyapunov_sum}
\end{figure*}

The complete Lyapunov spectra also provide information about
phase-space contraction. For each initial condition and parameter pair
$(a,b)$, we define
\begin{equation}
\label{eq:LIT_LambdaSigma}
\mathrm{\Lambda}_{\Sigma}(\vx_0;a,b)
=
\sum_{i=1}^{3}\lambda_i(\vx_0;a,b).
\end{equation}
Negative values of $\mathrm{\Lambda}_{\Sigma}$ indicate average
phase-space contraction along the corresponding trajectory, whereas
values close to zero indicate the absence of appreciable average
contraction.

To characterize its distribution over the sampled initial-condition
grid, we introduce the minimum
\begin{equation}
\label{eq:LIT_LambdaSigmaMin}
\mathrm{\Lambda}_{\Sigma,\min}(a,b)
=
\min_{\vx_0\in I_{h_{\vx}}}
\mathrm{\Lambda}_{\Sigma}(\vx_0;a,b),
\end{equation}
which identifies the strongest contraction detected among the sampled
initial conditions, and the mean
\begin{equation}
\label{eq:LIT_LambdaSigmaMean}
\mathrm{\Lambda}_{\Sigma,\mathrm{mean}}(a,b)
=
\frac{1}{|I_{h_{\vx}}|}
\sum_{\vx_0\in I_{h_{\vx}}}
\mathrm{\Lambda}_{\Sigma}(\vx_0;a,b),
\end{equation}
which measures the average contraction over the entire ensemble.

Fig.~\ref{fig:LIT_Lyapunov_sum} extends to the $(a,b)$ parameter plane
the dynamical picture inferred earlier from the representative
Poincar\'e sections and the corresponding Lyapunov maps. In particular,
along the previously analyzed section $b=1/2$, the parameter dependence
of the phase-space contraction is consistent with the local phase-space
organization discussed in Figs.~\ref{fig:map1}--\ref{fig:map2}.

For small values of $a$, in particular for $a<1$, both
$\mathrm{\Lambda}_{\Sigma,\min}$ and
$\mathrm{\Lambda}_{\Sigma,\mathrm{mean}}$ remain close to zero. This
agrees with the predominantly conservative-like dynamics observed for
$a=1/2$ in Fig.~\ref{fig:map1}(a) and
Fig.~\ref{fig:lambda_sigma_maps}(a), where extended regular structures
coexist with chaotic layers and no pronounced average phase-space
contraction is detected.

As $a$ increases, both contraction measures become progressively more
negative over a substantial part of the parameter plane. Along
$b=1/2$, this trend is already visible for $a=2$ and becomes more
pronounced for $a=3$, in agreement with the phase-space reorganization
shown in Figs.~\ref{fig:map1}(b) and \ref{fig:map2}(a). In these cases,
the growth of chaotic regions in the Poincar\'e sections and Lyapunov
maps is accompanied by increasingly strong average phase-space
contraction.

For larger values of $a$, the contraction remains substantial, although
its dependence on the parameters is clearly non-monotonic. In the range
$a\gtrsim4$, both $\mathrm{\Lambda}_{\Sigma,\min}$ and
$\mathrm{\Lambda}_{\Sigma,\mathrm{mean}}$ reach some of their most
negative values within the displayed scale. For still larger values,
approximately $a\gtrsim5.5$, the magnitude of both contraction measures
decreases again over part of the investigated range. This weakening of
the average contraction is consistent with the reappearance of extended
regular domains observed for $a=6$ in Fig.~\ref{fig:map2}(b).

Summarizing the results of this section, the numerical analysis reveals
a complex and strongly non-monotonic organization of the dynamics of
the trigonometric Nos\'e--Hoover oscillator. The Poincar\'e sections,
bifurcation diagrams, Lyapunov spectra, and initial-condition Lyapunov
maps show the coexistence and evolution of regular, chaotic, and
contracting dynamics, while the LIT analysis extends this picture to
the full $(a,b)$ parameter plane.

Despite this complexity, the ensemble-based diagnostics consistently
distinguish the two parameter lines $a=0$ and $b=0$. Along these lines,
chaotic dynamics is not detected within the sampled initial-condition
set, while the Lyapunov-sum quantities remain numerically close to
zero. Their exceptional character is therefore not restricted to a
particular trajectory or choice of initial conditions, but persists
over the entire sampled ensemble.

This numerical distinction suggests that the limits $a=0$ and $b=0$
are associated with structural changes in the differential system
itself. We therefore turn from the numerical characterization of the
dynamics to an analytical examination of these two limiting cases and
their integrability properties.


\section{Integrability and partial integrability}
\label{sec:integrability}

The two limiting cases $a=0$ and $b=0$ are fundamentally different.
For $a=0$, the divergence of the vector field vanishes identically and
the dynamics acquires a volume-preserving structure. In contrast, for
$b=0$  the third variable becomes constant along every trajectory,
yielding an immediate first integral and reducing the effective
dimension of the system. We analyze these two cases separately below.
 
\subsection{The volume-preserving limit $a=0$}
\label{subsec:a0}

For $a=0$, system~\eqref{eq:periodic_NH} reduces to
\begin{equation}
\label{eq:periodic_a0}
\begin{cases}
\dot x=\sin y,\\
\dot y=-\sin x,\\
\dot z=b(1-2\cos y).
\end{cases}
\end{equation}
The dynamics of the $(x,y)$ variables is decoupled from $z$ and
constitutes a one-degree-of-freedom Hamiltonian system. It admits the
first integral
\begin{equation}
\label{eq:first_integral_a0}
H(x,y)
=
\cos x+\cos y.
\end{equation}
Thus, the dynamics of the reduced $(x,y)$--subsystem is confined to the
level sets
\begin{equation}
\label{eq:set}
H(x,y)=h,\qquad -2\leq h\leq 2.
\end{equation}
The maximum $h=2$ is attained at $(x,y)=(0,0)$, whereas the minimum
$h=-2$ is attained at $(x,y)=(\pi,\pi)$ modulo $2\pi$. In
particular, the origin is an elliptic equilibrium of the reduced
Hamiltonian subsystem. In its neighborhood,
\begin{equation*}
\label{eq:H0_expansion}
H(x,y)
=
2-\frac{x^2+y^2}{2}
+
O\left(\|(x,y)\|^4\right).
\end{equation*}
Hence, for $h<2$ sufficiently close to $2$, the level sets $H=h$
are closed curves surrounding the origin and are locally
approximated by
\begin{equation*}
x^2+y^2=2(2-h).
\end{equation*}
As $h$ decreases from $2$, these curves deform according to the
global $2\pi$--periodic geometry of the system. The regular levels
$0<h<2$ surround the elliptic equilibrium $(0,0)$, while $h=0$ is
the separatrix level passing through the hyperbolic equilibria
$(0,\pi)$ and $(\pi,0)$ modulo $2\pi$. For $-2<h<0$, the regular
level curves surround the elliptic equilibrium $(\pi,\pi)$ modulo
$2\pi$.

Fig.~\ref{fig:integral} shows representative level curves
$H(x,y)=h$ in the fundamental domain
$(x,y)\in[-\pi,\pi)^2$, covering the full range
$-2\leq h\leq2$. The separatrix level $h=0$ divides the two
families of regular level curves. Their geometry also provides a
direct interpretation of the Poincar\'e section shown in
Fig.~\ref{fig:integrable}. Since the $(x,y)$ dynamics is confined
to a level set of $H$, the intersections of a trajectory with the
section $z=0\pmod{2\pi}$ necessarily belong to the same invariant
level. Consequently, the nested curves observed in the Poincar\'e
section correspond to the invariant level sets of the reduced
Hamiltonian system and provide a direct numerical manifestation of
its integrable structure.
\begin{figure}[t]
\centering
\includegraphics[width=0.85\linewidth]{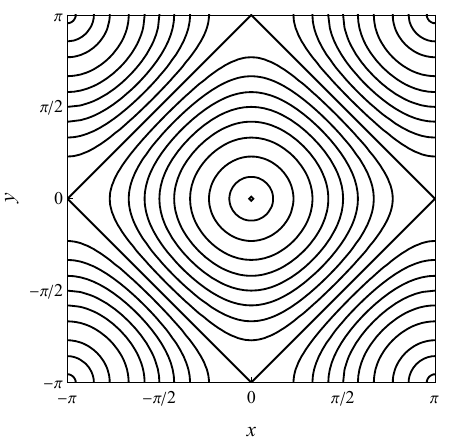}
\caption{Level curves $H(x,y)=h$, with $-2\leq h\leq2$, in the
fundamental domain $(x,y)\in[-\pi,\pi)\times[-\pi,\pi)$. The critical
level $h=0$ is the separatrix level separating the regular level curves
surrounding $(0,0)$ from those surrounding $(\pi,\pi)$ modulo $2\pi$.}
\label{fig:integral}
\end{figure}
 
A second first integral can be obtained directly from the
characteristic equations. Let $F(x,y,z)$ be a first integral of
system~\eqref{eq:periodic_a0}. The condition $\dot F=0$ leads to the
first-order partial differential equation
\[
-\sin x\,F_y+\sin y\,F_x
+b(1-2\cos y)F_z=0.
\]
The associated characteristic equations are
\begin{equation}
\label{eq:characteristics}
\frac{\rmd x}{\sin y}
=
\frac{\rmd y}{-\sin x}
=
\frac{\rmd z}{b(1-2\cos y)}.
\end{equation}
The first two characteristic equations immediately recover the first
integral~\eqref{eq:first_integral_a0}. 

We now restrict the analysis to
the regular level~\eqref{eq:set} 
Combining the second and third equations in~\eqref{eq:characteristics}
and choosing the branch $\sin x>0$ gives
\[
\rmd z
=
-b\frac{1-2\cos y}
{\sqrt{1-(h-\cos y)^2}}\,\rmd y.
\]
Therefore, a second first integral is
\begin{equation}
\label{eq:H2_periodic}
F
=
z+
b\int^y
\frac{1-2\cos u}
{\sqrt{1-(h-\cos u)^2}}\,\rmd u .
\end{equation}
The integral appearing in~\eqref{eq:H2_periodic} is elliptic and can be
expressed as a linear combination of incomplete elliptic integrals of
the first and third kinds. We use the standard notation
\[
J(\phi\mid m)
=
\int_0^\phi
\frac{\rmd\theta}
{\sqrt{1-m\sin^2\theta}}
\]
for the incomplete elliptic integral of the first kind, where $\phi$
is the amplitude and $m$ is the elliptic parameter. Similarly,
\[
\Pi(n;\phi\mid m)
=
\int_0^\phi
\frac{\rmd\theta}
{(1-n\sin^2\theta)
\sqrt{1-m\sin^2\theta}}
\]
denotes the incomplete elliptic integral of the third kind, where
$\phi$ is the amplitude, $m$ is the elliptic parameter, and $n$ is
the characteristic.

For the family of regular levels $0<h<2$ surrounding the elliptic
equilibrium $(0,0)$, the elliptic integral can be written entirely
in real form. We introduce
\[
m_h=\frac{h^2-4}{h^2},
\qquad
n_h=\frac{h-2}{h},
\]
and the real amplitude
\[
\psi_h(y)
=
\arcsin\left(
\sqrt{\frac{h}{2-h}}\tan\frac{y}{2}
\right).
\]
Then the second first integral takes the form
\begin{equation}
\label{eq:H2_elliptic}
F_+
=
z+
\frac{b}{h}
\left[
6J\!\left(\psi_h(y)\mid m_h\right)
-
8\Pi\!\left(n_h;\psi_h(y)\mid m_h\right)
\right].
\end{equation}
The representation~\eqref{eq:H2_elliptic} is understood with a
consistent choice of branches of the square roots and elliptic
functions. It corresponds to the branch $\sin x>0$ used in deriving
\eqref{eq:H2_periodic}; for the opposite branch $\sin x<0$, the sign
of the elliptic contribution in~\eqref{eq:H2_elliptic} is reversed.

The second family of regular levels, $-2<h<0$, surrounds the elliptic
equilibrium $(\pi,\pi)$. To obtain an analogous real representation,
we set
\[
\tilde h=-h\in(0,2),
\qquad
\xi=x-\pi,
\qquad
\eta=y-\pi.
\]
Then the level condition becomes
\[
\cos\xi+\cos\eta=\tilde h,
\]
whereas system~\eqref{eq:periodic_a0} takes the form
\[
\dot\xi=-\sin\eta,
\qquad
\dot\eta=\sin\xi,
\qquad
\dot z=b(1+2\cos\eta).
\]
On the branch $\sin\xi>0$, we therefore obtain
\begin{equation}
\label{eq:H2_periodic_negative}
\frac{\rmd z}{\rmd\eta}
=
b\,
\frac{1+2\cos\eta}
{\sqrt{1-(\tilde h-\cos\eta)^2}}.
\end{equation}

Introducing
\[
m_{\tilde h}
=
\frac{\tilde h^2-4}{\tilde h^2},
\qquad
n_{\tilde h}
=
\frac{\tilde h-2}{\tilde h},
\]
and the real amplitude
\[
\psi_{\tilde h}(\eta)
=
\arcsin\left(
\sqrt{\frac{\tilde h}{2-\tilde h}}
\tan\frac{\eta}{2}
\right),
\]
integration of~\eqref{eq:H2_periodic_negative} yields
\begin{equation}
\label{eq:H2_elliptic_negative}
F_-
=
z+
\frac{b}{\tilde h}
\left[
2J\!\left(
\psi_{\tilde h}(\eta)\mid m_{\tilde h}
\right)
-
8\Pi\!\left(
n_{\tilde h};
\psi_{\tilde h}(\eta)\mid m_{\tilde h}
\right)
\right].
\end{equation}
As for $F_+$, the expression is understood with a consistent choice
of branches, and the sign of the elliptic contribution is reversed
on the opposite branch $\sin\xi<0$.

The singular level $h=0$, which separates the two families of regular
levels, requires a separate treatment, since the elliptic
representations~\eqref{eq:H2_elliptic} and
\eqref{eq:H2_elliptic_negative} are not applicable there. Setting
$h=0$ directly in~\eqref{eq:H2_periodic} gives, on a fixed branch,
\[
F_0
=
z+
b\int^y
\frac{1-2\cos u}
{\sqrt{1-\cos^2u}}\,\rmd u.
\]
On the branch where the square root is represented by $\sin y$, the
integral is elementary and yields
\begin{equation}
\label{eq:H2_separatrix}
F_0
=
z-b\log\left|\sin y\,(1+\cos y)\right|.
\end{equation}
For the opposite choice of the square-root branch, the sign of the
logarithmic contribution is reversed. The first integral
\eqref{eq:H2_separatrix} is defined locally on each regular branch of
the separatrix $H=0$ and has logarithmic singularities at the saddle
points of the reduced $(x,y)$ dynamics. It therefore does not extend
as a regular first integral through the entire singular level.

Consequently, for $a=0$ and $b\neq0$, the second first integral admits
the real representations $F_+$ and $F_-$ on the regular families
$0<h<2$ and $-2<h<0$, respectively, while $F_0$ provides a local
elementary representation on the regular branches of the singular
level $h=0$. Together with the first integral $H$ given
by~\eqref{eq:first_integral_a0}, $F_+$ and $F_-$ establish complete
integrability on every regular domain where the corresponding
branches are chosen consistently and the two first integrals are
functionally independent.

Fig.~\ref{fig:integrables} shows the contour structure associated
with the real representations $F_+$ and $F_-$ of the second first
integral, separated by the singular level $H=0$. In contrast to the
simple level-set geometry of $H(x,y)=h$ shown in
Fig.~\ref{fig:integral}, the contours associated with $F_+$ and $F_-$
exhibit a considerably more intricate structure, reflecting the
nonlinear and branch-dependent form of the second first integral and
its singular behavior near the separatrix.

\begin{figure}[t]
\centering
\includegraphics[width=0.85\linewidth]{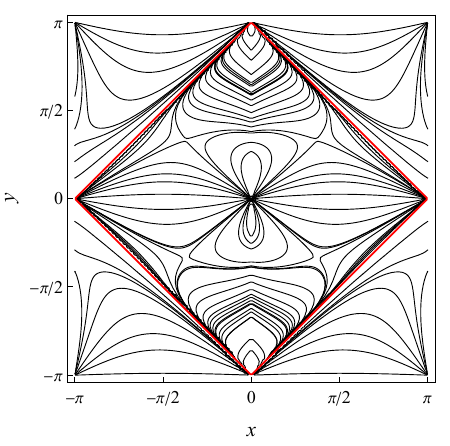}
\caption{(Color online)  Level curves of the second first integral in the fundamental
domain $(x,y)\in[-\pi,\pi)\times[-\pi,\pi)$. The black curves
correspond to the representations $F_+$ and $F_-$ given
by~\eqref{eq:H2_elliptic} and~\eqref{eq:H2_elliptic_negative}
for the regular families $0<h<2$ and $-2<h<0$, respectively.
The red curve represents the separatrix $H=0$, which separates the
two families of regular levels.}
\label{fig:integrables}
\end{figure}

\subsection{The frozen-thermostat limit $b=0$}
\label{subsec:b0}

For $b=0$, system~\eqref{eq:periodic_NH} reduces to
\begin{equation*}
\begin{cases}
\dot x=\sin y,\\
\dot y=-\sin x-a\sin y\sin z,\\
\dot z=0.
\end{cases}
\end{equation*}
Hence,
$
I=z
$
is a first integral, and the phase space is foliated by the invariant
surfaces
\[
z=c,
\qquad c\in\mathbb S^1.
\]
On each of these surfaces, the dynamics reduces to the planar system
\begin{equation}
\label{eq:reduced_b0}
\begin{cases}
\dot x=\sin y,\\
\dot y=-\sin x-\kappa\sin y,
\end{cases}
\qquad
\kappa=a\sin c.
\end{equation}

The corresponding phase trajectories satisfy
\begin{equation}
\label{eq:phase_b0}
\frac{\rmd y}{\rmd x}
=
-\frac{\sin x}{\sin y}-\kappa.
\end{equation}
For $\kappa=0$, this equation can be integrated directly, recovering
the first integral~\eqref{eq:first_integral_a0}. Thus, on the invariant
surfaces $z=k\pi$, the reduced dynamics is integrable.

For $\kappa\neq0$, the same first integral is lost and
\eqref{eq:phase_b0} does not yield an obvious second integral.
We therefore first examine the qualitative dynamics of the reduced
system and then investigate its algebraic integrability using Darboux
theory.

\subsubsection{Qualitative dynamics of the reduced system}
\label{subsubsec:qualitative-b0}

We begin with the vector field associated with
system~\eqref{eq:reduced_b0},
\begin{equation}
\label{eq:Xkappa-reduced-b0}
\mathcal X_\kappa
=
\sin y\,\frac{\partial}{\partial x}
+
\left(-\sin x-\kappa\sin y\right)
\frac{\partial}{\partial y}.
\end{equation}
Its divergence is
\begin{equation}
\label{eq:div-Xkappa-reduced-b0}
\nabla\cdot\mathcal X_\kappa
=
-\kappa\cos y.
\end{equation}
Thus, except for $\kappa=0$, the reduced flow is not
area-preserving. A nonzero divergence alone does not preclude the
existence of a first integral; rather, it measures the local
contraction or expansion of area elements under the flow.

For $\kappa\neq0$, the equilibria of~\eqref{eq:reduced_b0} are
\[
(x,y)=(k\pi,s\pi),
\qquad
k,s\in\{0,1\},
\]
where the angular variables are understood modulo $2\pi$.
 
The Jacobian matrix of~\eqref{eq:reduced_b0} is
\begin{equation*}
\label{eq:Jacobian-reduced-b0}
J(x,y)
=
\begin{pmatrix}
0 & \cos y\\
-\cos x & -\kappa\cos y
\end{pmatrix}.
\end{equation*}

At $(0,0)$ and $(\pi,\pi)$, the eigenvalues of the Jacobian matrix
are, respectively,
\begin{align*}
\rho_{\pm}^{(0,0)}
=
\frac{-\kappa\pm\sqrt{\kappa^2-4}}{2},
\qquad
\rho_{\pm}^{(\pi,\pi)}
=
\frac{\kappa\pm\sqrt{\kappa^2-4}}{2}.
\end{align*}
Consequently, for $0<\kappa<2$, the equilibrium $(0,0)$ is a
stable focus, whereas $(\pi,\pi)$ is an unstable focus. For
$-2<\kappa<0$, their stability is reversed. When $|\kappa|>2$,
the corresponding eigenvalues are real, so the two focal equilibria
are replaced by nodes, with the same stability determined by the
sign of $\kappa$. The limiting values $|\kappa|=2$ correspond to
degenerate nodes.

At the remaining two equilibria, $(0,\pi)$ and $(\pi,0)$, the
determinant of the Jacobian is negative,
\[
\det J(0,\pi)
=
\det J(\pi,0)
=
-1,
\]
and therefore both equilibria are hyperbolic saddles for every
$\kappa\in\mathbb R$.

The case $\kappa=1$ is illustrated in
Fig.~\ref{fig:stream_b0}. The equilibrium $(0,0)$ is a stable focus,
whereas $(\pi,\pi)$, represented by the identified corners of the
fundamental domain, is an unstable focus. The remaining equilibria
$(0,\pi)$ and $(\pi,0)$ are hyperbolic saddles and organize the
separating structures of the phase portrait.

The streamlines are colored according to the divergence of the
reduced vector field. For $\kappa=1$,
\[
\nabla\cdot\mathcal X_{1}=-\cos y.
\]
Thus, negative and positive values correspond to local area
contraction and expansion, respectively. In particular, the
contraction near $(0,0)$ and the expansion near $(\pi,\pi)$ are
consistent with the attracting and repelling character of these
equilibria.

\begin{figure}[t]
\centering
\includegraphics[width=1\linewidth]{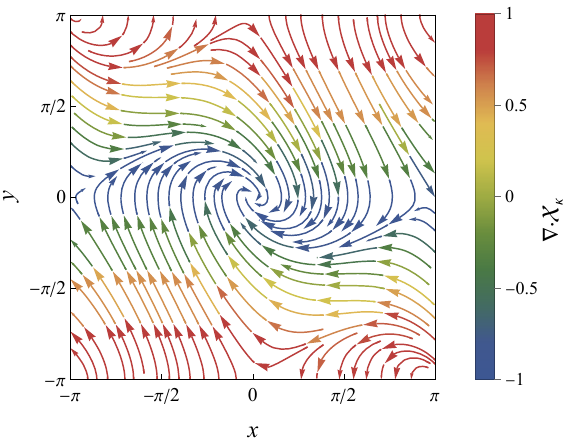}
\caption{(Color online) Stream plot of the reduced system~\eqref{eq:reduced_b0}
for $\kappa=1$ in the fundamental domain
$(x,y)\in[-\pi,\pi)^2$. The equilibrium $(0,0)$ is a stable focus,
$(\pi,\pi)$ is an unstable focus, and $(0,\pi)$ and $(\pi,0)$ are
hyperbolic saddles. The color scale represents the divergence
$\nabla\cdot\mathcal X_{1}=-\cos y$.}
\label{fig:stream_b0}
\end{figure}

The attracting and repelling equilibria also impose a restriction on
regular first integrals of the reduced system. Let $F$ be a continuous
first integral defined on an open set $U$ contained in the basin of
attraction of a stable equilibrium $p_*$. If $\varphi_t$ denotes the
flow of~\eqref{eq:reduced_b0}, then
\[
F(\varphi_t(p))=F(p),
\qquad p\in U.
\]
Since $\varphi_t(p)\to p_*$ as $t\to+\infty$, continuity of $F$
implies
\[
F(p)
=
\lim_{t\to+\infty}F(\varphi_t(p))
=
F(p_*).
\]
Hence $F$ is constant on $U$. The same argument applies to an
unstable equilibrium after reversing time. Thus, the presence of a
focus provides a local obstruction to the existence of a continuous,
locally nonconstant first integral on a full neighborhood contained
in its basin. This obstruction results from the attracting or
repelling dynamics, rather than from the nonzero divergence itself,
and does not exclude singular or branch-dependent first integrals
defined on restricted domains.

This qualitative picture already indicates important restrictions on
the regular integrability of~\eqref{eq:reduced_b0}. To complement this
dynamical description, we now examine the algebraic structure of the
reduced system. In particular, we investigate its invariant algebraic
curves and Darboux polynomials after transforming the system into a
polynomial form.
\subsubsection{Darboux analysis of the reduced system}
\label{subsec:darboux}

To apply Darboux theory, we introduce the tangent half-angle variables
\begin{equation*}
\label{eq:half_angle}
u=\tan\frac{y}{2},
\qquad
v=\tan\frac{x}{2}.
\end{equation*}
In these variables, system~\eqref{eq:reduced_b0} takes the rational form
\begin{equation}
\label{eq:rational_reduced}
\begin{cases}
\displaystyle
\dot u=-\kappa u-
       \frac{(1+u^2)v}{1+v^2},
\\[3mm]
\displaystyle
\dot v=\frac{u(1+v^2)}{1+u^2}.
\end{cases}
\end{equation}

Multiplication of the vector field by
\[
D(u,v)=(1+u^2)(1+v^2)>0
\]
corresponds to a nonsingular reparametrization of time in the real
$(u,v)$-chart and therefore preserves the phase curves and first
integrals. We may thus work with the orbitally equivalent polynomial
system
\begin{equation}
\label{eq:polynomial_reduced}
\begin{cases}
u'=P(u,v),\\
v'=Q(u,v),
\end{cases}
\end{equation}
where
\begin{equation}
\label{eq:PQ}
\begin{split}
P(u,v)
&=
(1+u^2)
\left[
-\kappa u(1+v^2)-v(1+u^2)
\right],
\\
Q(u,v)
&=
u(1+v^2)^2.
\end{split}
\end{equation}
Here the prime denotes differentiation with respect to the
reparametrized time. The associated polynomial vector field is
\[
\mathcal X
=
P\frac{\partial}{\partial u}
+
Q\frac{\partial}{\partial v}.
\]
Recall that a non-constant polynomial $f\in\mathbb C[u,v]$ is a
\emph{Darboux polynomial} of $\mathcal X$ if
\begin{equation}
\label{eq:Darboux_polynomial}
\mathcal X(f)=Kf
\end{equation}
for some polynomial $K\in\mathbb C[u,v]$, called its \emph{cofactor}.
Similarly, an \emph{exponential factor} is a function
\[
E=\exp\left(\frac{g}{h}\right),
\qquad
g,h\in\mathbb C[u,v],
\]
satisfying
\[
\mathcal X(E)=LE
\]
for some polynomial cofactor $L$ of degree at most four.

If Darboux polynomials $f_j$ and exponential factors $E_k$ have
cofactors $K_j$ and $L_k$, respectively, and there exist constants
$\lambda_j,\mu_k\in\mathbb C$, not all zero, such that
\begin{equation}
\label{eq:general_cofactor_relation}
\sum_{j=1}^{p}\lambda_jK_j
+
\sum_{k=1}^{q}\mu_kL_k
=0,
\end{equation}
then
\begin{equation}
\label{eq:general_Darboux_integral}
H=
\prod_{j=1}^{p}f_j^{\lambda_j}
\prod_{k=1}^{q}E_k^{\mu_k}
\end{equation}
is a Darboux first integral.

For system~\eqref{eq:polynomial_reduced}, two Darboux polynomials are
\begin{equation*}
\label{eq:Darboux_f12}
f_1=1+u^2,
\qquad
f_2=1+v^2,
\end{equation*}
with cofactors
\begin{equation*}
\label{eq:Darboux_K12}
\begin{split}
K_1
&=
2u\left[
-\kappa u(1+v^2)-v(1+u^2)
\right],
\\
K_2
&=
2uv(1+v^2).
\end{split}
\end{equation*}
The system also admits the exponential factors
\begin{equation*}
\label{eq:exponential_factors}
E_1=
\exp\left(\frac{1}{1+v^2}\right),
\qquad
E_2=
\exp\left(\frac{v}{1+v^2}\right),
\end{equation*}
with cofactors
\begin{equation*}
\label{eq:exponential_cofactors}
L_1=-2uv,
\qquad
L_2=u(1-v^2).
\end{equation*}

We first consider $\kappa\neq0$. Suppose that the four cofactors
$K_1,K_2,L_1,L_2$ satisfy
\begin{equation*}
\label{eq:cofactor_relation}
\lambda_1K_1+\lambda_2K_2+
\mu_1L_1+\mu_2L_2=0.
\end{equation*}
Comparison of the coefficients of $u^2$, $uv^3$, $u$, and $uv$
shows successively that
\[
\lambda_1=\lambda_2=\mu_2=\mu_1=0.
\]
Thus, for $\kappa\neq0$, the factors $f_1,f_2,E_1,E_2$ do not satisfy
a nontrivial cofactor relation and therefore do not generate a
non-constant Darboux first integral.

The value $\kappa=2$ is nevertheless distinguished by an additional
Darboux polynomial
\begin{equation*}
\label{eq:Darboux_kappa2}
f_3=u+v,
\end{equation*}
with cofactor
\begin{equation*}
\label{eq:Darboux_K_kappa2}
K_3
=
-1-2u^2-u^3v-u^2v^2+uv^3.
\end{equation*}
Hence, the algebraic curve $u+v=0$ is invariant and corresponds
locally to $y=-x\pmod{2\pi}$ in the original variables. Including
$K_3$ in the cofactor relation still gives only the trivial solution,
so this additional Darboux polynomial does not produce a first
integral together with $f_1,f_2,E_1,E_2$. The same reasoning applies
to $\kappa=-2$, for which the additional Darboux polynomial is $u-v$,
corresponding locally to the invariant curve $y=x\pmod{2\pi}$.

For $\kappa=0$, the situation changes. In addition to $f_1$ and $f_2$,
the system admits the Darboux polynomials
\begin{equation*}
\label{eq:Darboux_f34}
f_3=1-uv,
\qquad
f_4=1+uv,
\end{equation*}
with cofactors
\begin{equation*}
\label{eq:Darboux_K34}
\begin{split}
K_3
&=
-(u^2-v^2)(1+uv),
\\
K_4
&=
(u^2-v^2)(1-uv).
\end{split}
\end{equation*}
In this case,
\[
K_1=-2uv(1+u^2),
\qquad
K_2=2uv(1+v^2),
\]
and the four cofactors satisfy the nontrivial relation
\begin{equation*}
\label{eq:cofactor_relation_k0}
K_3+K_4-K_1-K_2=0.
\end{equation*}
It follows that
\begin{equation}
\label{eq:Darboux_integral_k0}
\mathcal H(u,v)
=
\frac{f_3f_4}{f_1f_2}
=
\frac{1-u^2v^2}
     {(1+u^2)(1+v^2)}
\end{equation}
is a Darboux first integral.

Returning to the original variables gives
\[
\cos x+\cos y
=
\frac{2(1-u^2v^2)}
     {(1+u^2)(1+v^2)}
=
2\mathcal H.
\]
Thus, up to a constant factor, the Darboux first
integral~\eqref{eq:Darboux_integral_k0} reproduces the first integral
already obtained for $\kappa=0$ given
in~\eqref{eq:first_integral_a0}.

The Darboux analysis therefore recovers the integrable case
$\kappa=0$ and reveals the additional invariant algebraic curves at
$\kappa=\pm2$. For nonzero $\kappa$, however, the Darboux polynomials
and exponential factors identified above do not satisfy a nontrivial
cofactor relation and hence do not generate a first integral.

This does not establish non-integrability in the complete Darboux
class, since further Darboux polynomials or exponential factors may
exist. A complete result would require their classification or an
independent obstruction to the existence of a Darboux first integral.
In the following sections, we obtain independent obstructions to
regular $C^1$ and meromorphic $B$-integrability using averaging
and differential Galois theory, respectively.

\section{Averaging near  the integrable axes}
\label{sec:periodic-solutions-averaging}

We now investigate periodic solutions of the periodic Nos\'e--Hoover
oscillator~\eqref{eq:periodic_NH} by means of averaging theory. Our analysis is motivated by the approach introduced by Llibre, Messias and
Reinol~\cite{Llibre:21::} in their study of the modified classical
Nos\'e--Hoover oscillator~\eqref{eq:modified_NH}.
We follow their construction as closely as the periodic character of
system~\eqref{eq:periodic_NH} allows.

For the modified classical system~\eqref{eq:modified_NH}, the
unperturbed dynamics in the $(x,y)$ variables reduces to the harmonic
oscillator. Consequently, polar coordinates provide a natural
amplitude--phase parametrization, with the angular variable serving as
the fast variable in the averaging procedure. In contrast, the
corresponding unperturbed subsystem of the periodic
model~\eqref{eq:periodic_NH} is nonlinear because of the trigonometric
terms $\sin x$ and $\sin y$. Therefore, ordinary polar coordinates are
no longer adapted to its trajectories. Instead, we use energy--phase
coordinates associated with the periodic level curves of the
integrable unperturbed subsystem. In this way, the phase variable plays
the same role as the polar angle in the classical construction, while
preserving the intrinsic geometry of the periodic system.

\subsection{The unperturbed periodic family and energy--phase coordinates}
\label{subsec:unperturbed-family-periodic-NH}

We approach the intersection of
the two integrable parameter axes by introducing
\begin{equation}
\label{eq:ab-scaling-periodic-NH}
a=\varepsilon a_1,
\qquad
b=\varepsilon b_1,
\qquad
a_1b_1\neq0,
\qquad
0<\varepsilon\ll1.
\end{equation}
Then system~\eqref{eq:periodic_NH} takes the form
\begin{equation}
\label{eq:periodic-NH-epsilon}
\begin{cases}
\dot x=\sin y,\\[1mm]
\dot y=-\sin x-\varepsilon a_1\sin y\sin z,\\[1mm]
\dot z=\varepsilon b_1(1-2\cos y).
\end{cases}
\end{equation}

For $\varepsilon=0$, the variable $z$ is constant and the
$(x,y)$--subsystem reduces to the one-degree-of-freedom Hamiltonian
system considered in Section~\ref{sec:integrability}, with the first
integral
\[
H(x,y)=\cos x+\cos y;
\]
see~\eqref{eq:first_integral_a0}. For $0<h<2$, the regular level sets
\[
\Gamma_h
=
\left\{
(x,y)\in\mathbb T^2:
H(x,y)=h
\right\}
\]
are periodic orbits surrounding the elliptic equilibrium $(0,0)$.
Hence the unperturbed three-dimensional system possesses the
two-parameter family of periodic orbits
\[
\Gamma_h\times\{z_0\},
\qquad
0<h<2,
\qquad
z_0\in\mathbb S^1.
\]
The limiting levels $h=0$ and $h=2$ correspond, respectively, to the
separatrix containing the saddle equilibria $(0,\pi)$ and $(\pi,0)$,
and to the elliptic equilibrium $(0,0)$.

To parametrize the periodic orbits $\Gamma_h$, we introduce
\[
u=\frac{x+y}{2},
\qquad
v=\frac{x-y}{2}.
\]
Then $x=u+v$, $y=u-v$, and
\begin{equation}
\label{eq:H-uv-periodic}
2\cos u\cos v=h.
\end{equation}
The unperturbed equations give
\begin{equation}
\label{eq:uv-unperturbed}
\dot u=-\cos u\sin v,
\qquad
\dot v=\sin u\cos v.
\end{equation}
Using~\eqref{eq:H-uv-periodic} to eliminate $v$, we obtain
\begin{equation}
\label{eq:udot-square-periodic}
\dot u^{\,2}
=
\cos^2u-\frac{h^2}{4}.
\end{equation}

Let us introduce the elliptic modulus
\begin{equation}
\label{eq:k-def-periodic}
k=k(h)
=
\sqrt{1-\frac{h^2}{4}},
\qquad
k'
=
\sqrt{1-k^2}
=
\frac{h}{2}.
\end{equation}
For $0<h<2$, both $k$ and $k'$ belong to $(0,1)$. Setting
$w=\sin u$, equation~\eqref{eq:udot-square-periodic} becomes
\[
\dot w^{\,2}
=
(1-w^2)(k^2-w^2).
\]
Therefore, up to a translation of the origin of time,
\begin{equation}
\label{eq:u-elliptic-periodic}
\sin u
=
k\,\operatorname{sn}(t,k),
\qquad
\cos u
=
\operatorname{dn}(t,k).
\end{equation}
Using~\eqref{eq:H-uv-periodic} and~\eqref{eq:uv-unperturbed}, the
remaining trigonometric functions are
\begin{equation}
\label{eq:v-elliptic-periodic}
\cos v
=
\frac{k'}{\operatorname{dn}(t,k)},
\qquad
\sin v
=
-\frac{k\,\operatorname{cn}(t,k)}
{\operatorname{dn}(t,k)}.
\end{equation}

We denote by $K(k)$ and $E(k)$ the complete elliptic integrals of the
first and second kinds,
\begin{equation}
\label{eq:complete-elliptic-K}
K(k)
=
\int_0^{\pi/2}
\frac{\rmd\varphi}
{\sqrt{1-k^2\sin^2\varphi}},
\end{equation}
and
\begin{equation}
\label{eq:complete-elliptic-E}
E(k)
=
\int_0^{\pi/2}
\sqrt{1-k^2\sin^2\varphi}\,
\rmd\varphi.
\end{equation}
Throughout this section we use the modulus convention for the Jacobi
elliptic functions and complete elliptic integrals; the corresponding
parameter is $m=k^2$.

The parametrization
\eqref{eq:u-elliptic-periodic}--\eqref{eq:v-elliptic-periodic}
completes one revolution of $\Gamma_h$ over an interval of length
$4K(k)$. Thus
\begin{equation}
\label{eq:T-h-elliptic}
T(h)=4K(k),
\qquad
\Omega(h)
=
\frac{2\pi}{T(h)}
=
\frac{\pi}{2K(k)}.
\end{equation}
Unlike the harmonic oscillator arising in the modified classical
Nos\'e--Hoover model, the present unperturbed system is not
isochronous. Indeed,
\[
T(h)\longrightarrow2\pi,
\qquad
\Omega(h)\longrightarrow1
\qquad
\text{as } h\longrightarrow2^-,
\]
whereas
\[
T(h)\longrightarrow+\infty,
\qquad
\Omega(h)\longrightarrow0
\qquad
\text{as } h\longrightarrow0^+.
\]

We introduce a $2\pi$--periodic phase variable
$\theta\in\mathbb S^1$ along $\Gamma_h$ by
\[
\dot\theta=\Omega(h).
\]
For notational convenience, we define
\begin{equation}
\label{eq:s-theta-periodic}
s=s(h,\theta)
=
\frac{\theta}{\Omega(h)}
=
\frac{2K(k)}{\pi}\theta.
\end{equation}
A shift of the initial phase only changes the origin of the
parametrization along $\Gamma_h$ and will therefore be suppressed.

On every compact annulus contained in $0<h<2$, the variables
$(h,\theta)$ form regular energy--phase coordinates and
$\Omega(h)$ is bounded away from zero. We write
\[
x=X(h,\theta),
\qquad
y=Y(h,\theta)
\]
for the original variables expressed in these coordinates.

Combining $x=u+v$, $y=u-v$ with
\eqref{eq:u-elliptic-periodic}--\eqref{eq:v-elliptic-periodic} gives
\begin{equation}
\label{eq:XY-elliptic}
\begin{aligned}
\cos X(h,\theta)
&=
k'
+
\frac{k^2
\operatorname{sn}(s,k)\operatorname{cn}(s,k)}
{\operatorname{dn}(s,k)},
\\
\sin X(h,\theta)
&=
-k\,\operatorname{cn}(s,k)
+
\frac{kk'\operatorname{sn}(s,k)}
{\operatorname{dn}(s,k)},
\\
\cos Y(h,\theta)
&=
k'
-
\frac{k^2
\operatorname{sn}(s,k)\operatorname{cn}(s,k)}
{\operatorname{dn}(s,k)},
\\
\sin Y(h,\theta)
&=
k\,\operatorname{cn}(s,k)
+
\frac{kk'\operatorname{sn}(s,k)}
{\operatorname{dn}(s,k)}.
\end{aligned}
\end{equation}
Here $k=k(h)$ and $s=s(h,\theta)$ are given by in
\eqref{eq:k-def-periodic} and~\eqref{eq:s-theta-periodic}.
In particular,
\[
\cos X(h,\theta)+\cos Y(h,\theta)=h,
\]
and, for each fixed $h$, the map
\[
\theta
\longmapsto
\bigl(X(h,\theta),Y(h,\theta)\bigr)
\]
parametrizes $\Gamma_h$ with period $2\pi$ in $\theta$. Thus, $(h,\theta)$ play the role of nonlinear energy--phase coordinates,
analogous to the amplitude--phase coordinates used for the harmonic
oscillator in the previous study~\cite{Llibre:21::}.

\subsection{The first-order averaged system}
\label{subsec:averaged-system-periodic-NH}

We next derive the equations governing the slow evolution of $h$ and
$z$. Along system~\eqref{eq:periodic-NH-epsilon},
\begin{equation}
\label{eq:hdot-periodic-NH}
\dot h
=
\varepsilon a_1\sin^2y\sin z,
\end{equation}
whereas
\begin{equation*}
\label{eq:zdot-periodic-NH}
\dot z
=
\varepsilon b_1(1-2\cos y).
\end{equation*}
Thus $h$ and $z$ vary on the slow $O(\varepsilon)$ time scale.

Since the energy--phase transformation is smooth on the compact
annulus under consideration, the perturbed phase equation has the
form
\begin{equation}
\label{eq:thetadot-perturbed}
\dot\theta
=
\Omega(h)
+
\varepsilon G(h,\theta,z,\varepsilon),
\end{equation}
where $G$ is bounded and $2\pi$--periodic in $\theta$. Because
$\Omega(h)$ is bounded away from zero, $\theta$ remains monotone for
all sufficiently small $\varepsilon>0$ and can therefore be used as
the independent variable.

Using~\eqref{eq:hdot-periodic-NH}--\eqref{eq:thetadot-perturbed}, we
obtain
\begin{align}
\frac{\rmd h}{\rmd\theta}
&=
\varepsilon
\frac{a_1}{\Omega(h)}
\sin^2Y(h,\theta)\sin z
+
O(\varepsilon^2),
\label{eq:dh-dtheta-periodic}
\\[1mm]
\frac{\rmd z}{\rmd\theta}
&=
\varepsilon
\frac{b_1}{\Omega(h)}
\bigl(1-2\cos Y(h,\theta)\bigr)
+
O(\varepsilon^2).
\label{eq:dz-dtheta-periodic}
\end{align}
With
\[
\boldsymbol{u}=(h,z)^T,
\]
these equations take the standard first-order averaging form
\begin{equation}
\label{eq:standard-averaging-periodic}
\frac{\rmd\boldsymbol{u}}{\rmd\theta}
=
\varepsilon\boldsymbol{F}(\theta,\boldsymbol{u})
+
O(\varepsilon^2),
\qquad
\boldsymbol{F}(\theta+2\pi,\boldsymbol{u})
=
\boldsymbol{F}(\theta,\boldsymbol{u}),
\end{equation}
where
\begin{equation*}
\label{eq:F-periodic}
\boldsymbol{F}(\theta,h,z)
=
\begin{pmatrix}
\dfrac{a_1}{\Omega(h)}
\sin^2Y(h,\theta)\sin z
\\[3mm]
\dfrac{b_1}{\Omega(h)}
\bigl(1-2\cos Y(h,\theta)\bigr)
\end{pmatrix}.
\end{equation*}

The first-order averaged vector field is
\begin{equation}
\label{eq:Fbar-definition-periodic}
\overline{\boldsymbol{F}}(h,z)
=
\frac{1}{2\pi}
\int_0^{2\pi}
\boldsymbol{F}(\theta,h,z)\,\rmd\theta.
\end{equation}
Since one period $0\leq\theta\leq2\pi$ corresponds, through
\eqref{eq:s-theta-periodic}, to $0\leq s\leq4K(k)$, we introduce the
orbit average
\begin{equation*}
\label{eq:orbit-average-periodic}
\left\langle g\right\rangle
:=
\frac{1}{4K(k)}
\int_0^{4K(k)}
g(s)\,\rmd s.
\end{equation*}

From~\eqref{eq:XY-elliptic},
\[
\cos Y
=
k'
-
k^2
\frac{\operatorname{sn}(s,k)\operatorname{cn}(s,k)}
{\operatorname{dn}(s,k)}.
\]
The second term has zero average over a complete period because
\[
\frac{\rmd}{\rmd s}
\log\operatorname{dn}(s,k)
=
-k^2
\frac{\operatorname{sn}(s,k)\operatorname{cn}(s,k)}
{\operatorname{dn}(s,k)}.
\]
Consequently,
\begin{equation}
\label{eq:cos-average-periodic}
\left\langle\cos Y\right\rangle
=
k'
=
\frac{h}{2}.
\end{equation}

For the remaining average, we define
\begin{equation*}
\label{eq:A-h-periodic}
A(h)
:=
\left\langle\sin^2Y\right\rangle.
\end{equation*}
Using~\eqref{eq:XY-elliptic}, the mixed term in $\sin^2Y$ again has
zero average, and therefore
\begin{equation*}
\label{eq:A-intermediate-periodic}
A(h)
=
k^2
\left\langle
\operatorname{cn}^2(s,k)
\right\rangle
+
k^2k'^2
\left\langle
\frac{\operatorname{sn}^2(s,k)}
{\operatorname{dn}^2(s,k)}
\right\rangle.
\end{equation*}
The standard complete-period identities
\begin{equation*}
\label{eq:cn2-average-periodic}
\left\langle
\operatorname{cn}^2(s,k)
\right\rangle
=
\frac{1}{k^2}
\left(
\frac{E(k)}{K(k)}-k'^2
\right)
\end{equation*}
and
\begin{equation*}
\label{eq:sndn-average-periodic}
\left\langle
\frac{\operatorname{sn}^2(s,k)}
{\operatorname{dn}^2(s,k)}
\right\rangle
=
\frac{1}{k^2k'^2}
\left(
\frac{E(k)}{K(k)}-k'^2
\right)
\end{equation*}
then yield
\begin{equation}
\label{eq:A-explicit-periodic}
A(h)
=
2\left(
\frac{E(k)}{K(k)}-k'^2
\right)
=
2\left(
\frac{E(k)}{K(k)}-\frac{h^2}{4}
\right).
\end{equation}
These identities follow from standard properties of Jacobi elliptic functions
and complete elliptic integrals~\cite[Chapters 19 and 22]{Olver:2010ouy}.

Since $A(h)$ is the average of $\sin^2Y$ over a nontrivial periodic
orbit,
\begin{equation}
\label{eq:A-positive-periodic}
A(h)>0,
\qquad
0<h<2.
\end{equation}

It is convenient to introduce
\begin{equation}
\label{eq:B-h-periodic}
B(h)
:=
\frac{A(h)}{\Omega(h)}
=
\frac{4}{\pi}
\left(
E(k)-\frac{h^2}{4}K(k)
\right),
\end{equation}
and
\begin{equation}
\label{eq:C-h-periodic}
C(h)
:=
\frac{1-h}{\Omega(h)}
=
\frac{2K(k)}{\pi}(1-h).
\end{equation}
Using~\eqref{eq:cos-average-periodic}, the averaged vector field
\eqref{eq:Fbar-definition-periodic} becomes
\begin{equation}
\label{eq:Fbar-periodic}
\overline{\boldsymbol{F}}(h,z)
=
\begin{pmatrix}
a_1B(h)\sin z\\[1mm]
b_1C(h)
\end{pmatrix}.
\end{equation}

By~\eqref{eq:A-positive-periodic} and~\eqref{eq:T-h-elliptic},
$B(h)>0$ for $0<h<2$, whereas~\eqref{eq:C-h-periodic} implies that
$C(h)$ vanishes only at $h=1$. Hence, modulo $2\pi$ in $z$, the
averaged vector field has precisely two zeros,
\begin{equation}
\label{eq:Pk-periodic}
\boldsymbol{P}_j=(1,j\pi),
\qquad
j=0,1.
\end{equation}

This result has a direct analog in the classical Nos\'e--Hoover
oscillator considered in a previous study~\cite{Llibre:21::}. There, the
unperturbed harmonic oscillator is described by amplitude--phase
coordinates, and the averaged thermostat equation selects the circle
$r^2=2$. In the present system, the harmonic parametrization is
replaced by the nonlinear energy--phase coordinates $(h,\theta)$,
and the corresponding averaged condition
\[
1-2\langle\cos y\rangle_h=0
\]
selects the level $h=1$. Thus, the analogue of the classical circle
$x^2+y^2=2$ is the nonlinear periodic orbit
\begin{equation*}
\label{eq:periodic-analogue-circle}
\cos x+\cos y=1.
\end{equation*}
Moreover, the linear dependence on the thermostat variable in the
classical model is replaced here by the periodic coupling $\sin z$,
which selects the two phases $z=0$ and $z=\pi$ modulo $2\pi$ and
reduces locally to $z$ near $z=0$.

At the selected level $h=1$,
\begin{equation}
\label{eq:k-star-periodic}
k_*=\frac{\sqrt3}{2},
\qquad
k_*'=\frac12.
\end{equation}
For brevity, let
\[
K_*=K(k_*),
\qquad
E_*=E(k_*).
\]
Then
\begin{equation}
\label{eq:T-Omega-star-periodic}
T(1)=4K_*,
\qquad
\Omega(1)=\frac{\pi}{2K_*},
\end{equation}
and
\begin{equation}
\label{eq:A-star-periodic}
A(1)
=
2\left(
\frac{E_*}{K_*}-\frac14
\right).
\end{equation}

To verify that the zeros~\eqref{eq:Pk-periodic} are simple, note that
\begin{equation*}
\label{eq:DFbar-general-periodic}
D\overline{\boldsymbol{F}}(h,z)
=
\begin{pmatrix}
a_1B'(h)\sin z
&
a_1B(h)\cos z
\\[2mm]
b_1C'(h)
&
0
\end{pmatrix}.
\end{equation*}
Since
\begin{equation*}
\label{eq:C-prime-star-periodic}
C'(1)
=
-\frac{1}{\Omega(1)}
=
-\frac{2K_*}{\pi},
\end{equation*}
we obtain
\begin{equation}
\label{eq:DFbar-Pk-periodic}
D\overline{\boldsymbol{F}}(\vP_j)
=
\begin{pmatrix}
0
&
(-1)^j a_1B(1)
\\[2mm]
-\dfrac{b_1}{\Omega(1)}
&
0
\end{pmatrix},\qquad j=0,1.
\end{equation}
Its determinant can be written in the particularly simple form
\begin{equation}
\label{eq:det-DFbar-periodic}
\det D\overline{\boldsymbol{F}}(\vP_j)
=
(-1)^j
\frac{a_1b_1A(1)}{\Omega(1)^2},\qquad j=0,1.
\end{equation}
Since $a_1b_1\neq0$ and $A(1)>0$, both zeros are simple.

\subsection{Periodic orbits near the integrable axes}

The two simple zeros of the averaged vector field identified above
select distinguished members of the two-parameter family of periodic
orbits of the unperturbed system. We now apply the first-order
averaging theorem to show that these selected orbits persist as
periodic solutions of the full system for sufficiently small
$\varepsilon>0$. The following result gives their limiting form as
$\varepsilon\to0$ and the corresponding leading-order periods.

\begin{theorem}
\label{thm:periodic-solutions-periodic-NH}
Consider system~\eqref{eq:periodic_NH} under the scaling
\eqref{eq:ab-scaling-periodic-NH}. Then there exists
$\varepsilon_0>0$ such that, for every
$\varepsilon\in(0,\varepsilon_0)$, the system possesses two periodic
solutions
\begin{equation}
\label{eq:per}
\boldsymbol{x}_j(t,\varepsilon)
=
\bigl(
x_j(t,\varepsilon),
y_j(t,\varepsilon),
z_j(t,\varepsilon)
\bigr),
\qquad
j=0,1.
\end{equation}
We denote the corresponding periodic orbits by
\begin{equation}
\label{eq:Gamma-periodic-NH}
\Gamma_{j,\varepsilon}
=
\left\{
\boldsymbol{x}_j(t,\varepsilon):
t\in\mathbb R
\right\},
\qquad
j=0,1.
\end{equation}
As $\varepsilon\to0$, these orbits bifurcate from
\[
\Gamma_1\times\{0\}
\qquad\text{and}\qquad
\Gamma_1\times\{\pi\},
\]
respectively, where
\[
\Gamma_1
=
\left\{
(x,y)\in\mathbb T^2:
\cos x+\cos y=1
\right\}.
\]
More precisely,
\begin{equation}
\label{eq:periodic-limit-orbits}
\Gamma_{j,\varepsilon}
\longrightarrow
\Gamma_1\times\{j\pi\},
\qquad
j=0,1,
\qquad
\varepsilon\to0.
\end{equation}
Their periods satisfy
\begin{equation}
\label{eq:periods-periodic-NH}
T_{j,\varepsilon}
=
T(1)+O(\varepsilon),
\qquad
j=0,1.
\end{equation}
\end{theorem}
\begin{proof}
Equations~\eqref{eq:dh-dtheta-periodic} and
\eqref{eq:dz-dtheta-periodic} put the slow dynamics into the standard
first-order averaging form~\eqref{eq:standard-averaging-periodic}.
The corresponding averaged vector field is given by in
\eqref{eq:Fbar-periodic}. Its only zeros in the annulus under
consideration are~\eqref{eq:Pk-periodic}. By~\eqref{eq:det-DFbar-periodic}, both
zeros are simple. Hence, by the first-order averaging theorem, for
every sufficiently small $\varepsilon>0$, each $\vP_j$ gives rise to a
$2\pi$--periodic solution of the $\theta$--system and, consequently,
to a periodic solution
$
\boldsymbol{x}_j(t,\varepsilon)$~\eqref{eq:per}
of the original system~\eqref{eq:periodic_NH}.

We now identify the unperturbed periodic orbits from which these
solutions bifurcate. Since both zeros $P_0$ and $P_1$ lie on the
selected level $h=1$, we use the notation introduced in
\eqref{eq:k-star-periodic}--\eqref{eq:A-star-periodic}. Choosing the
origin of time so that the elliptic parameter $s$ coincides with $t$,
the parametrization~\eqref{eq:XY-elliptic} restricted to $h=1$ yields
\begin{align}
\cos x_*(t)
&=
\frac12
+
\frac34
\frac{
\operatorname{sn}(t,k_*)
\operatorname{cn}(t,k_*)
}{
\operatorname{dn}(t,k_*)
},
\label{eq:x0-cos-periodic-NH}
\\[1mm]
\sin x_*(t)
&=
-\frac{\sqrt3}{2}
\operatorname{cn}(t,k_*)
+
\frac{\sqrt3}{4}
\frac{
\operatorname{sn}(t,k_*)
}{
\operatorname{dn}(t,k_*)
},
\label{eq:x0-sin-periodic-NH}
\\[1mm]
\cos y_*(t)
&=
\frac12
-
\frac34
\frac{
\operatorname{sn}(t,k_*)
\operatorname{cn}(t,k_*)
}{
\operatorname{dn}(t,k_*)
},
\label{eq:ystar-cos-periodic-NH}
\\[1mm]
\sin y_*(t)
&=
\frac{\sqrt3}{2}
\operatorname{cn}(t,k_*)
+
\frac{\sqrt3}{4}
\frac{
\operatorname{sn}(t,k_*)
}{
\operatorname{dn}(t,k_*)
}.
\label{eq:y0-sin-periodic-NH}
\end{align}
In particular,
\[
\cos x_*(t)+\cos y_*(t)=1,
\]
so $(x_*(t),y_*(t))$ parametrizes the unperturbed periodic orbit
$\Gamma_1$.

The two limiting periodic solutions corresponding to $P_0$ and
$P_1$ are therefore
\begin{equation}
\label{eq:x00-periodic-NH}
\boldsymbol{x}_0(t,0)
=
\bigl(
x_*(t),
y_*(t),
0
\bigr),
\end{equation}
and
\begin{equation}
\label{eq:x10-periodic-NH}
\boldsymbol{x}_1(t,0)
=
\bigl(
x_*(t),
y_*(t),
\pi
\bigr),
\end{equation}
respectively. The corresponding unperturbed orbits are
\[
\Gamma_{0,0}
=
\Gamma_1\times\{0\},
\qquad
\Gamma_{1,0}
=
\Gamma_1\times\{\pi\}.
\]
Thus, the periodic orbits obtained by averaging satisfy
\[
\Gamma_{j,\varepsilon}
\longrightarrow
\Gamma_1\times\{j\pi\},
\qquad
j=0,1,
\qquad
\varepsilon\to0.
\]
Equivalently, after a suitable choice of phase along each periodic
solution,
\[
\begin{aligned}
x_j(t,\varepsilon)
&=
x_*(t)+O(\varepsilon),
\\
y_j(t,\varepsilon)
&=
y_*(t)+O(\varepsilon),
\\
z_j(t,\varepsilon)
&=
j\pi+O(\varepsilon),
\end{aligned}
\qquad
j=0,1,
\]
uniformly over one period.

Finally, along either branch, the fast angular frequency converges to
$\Omega(1)$. Therefore,
\[
T_{j,\varepsilon}
=
\frac{2\pi}{\Omega(1)}
+
O(\varepsilon)
=
T(1)+O(\varepsilon),
\qquad
j=0,1,
\]
which proves~\eqref{eq:periods-periodic-NH}.

\end{proof}

\section{Phase portrait and bifurcating periodic orbits}
\label{sec:phase-portrait-periodic-NH}

The averaged vector field~\eqref{eq:Fbar-periodic} provides a
two-dimensional description of the slow dynamics near the
intersection $a=b=0$ of the two integrable parameter axes. Its two
equilibria~\eqref{eq:Pk-periodic} correspond, by
Theorem~\ref{thm:periodic-solutions-periodic-NH}, to the two periodic
branches bifurcating from the degenerate unperturbed family. We first
determine the local character of these equilibria.

\begin{theorem}
\label{thm:periodic-NH-averaged-type}
Assume $a_1b_1\neq0$. The eigenvalues of the Jacobian matrix of the
averaged vector field at $\vP_j$, $j=0,1$, are
\begin{equation}
\label{eq:periodic-NH-averaged-eigenvalues}
\lambda_{j,\pm}
=
\pm
\frac{1}{\Omega(1)}
\sqrt{-(-1)^j a_1b_1A(1)}.
\end{equation}
Consequently, if $a_1b_1>0$, then $P_0$ is a center and $P_1$ is a
hyperbolic saddle. If $a_1b_1<0$, their roles are reversed.
\end{theorem}

\begin{proof}
The characteristic equation of the matrix
\eqref{eq:DFbar-Pk-periodic} immediately gives
\eqref{eq:periodic-NH-averaged-eigenvalues}.

To distinguish a nonlinear center from a merely linearly elliptic
equilibrium, we use the Hamiltonian structure of the averaged
dynamics. Since $B(h)>0$ for $0<h<2$, the positive time
reparametrization
\[
\rmd\tau=B(h)\,\rmd\theta
\]
does not change the trajectories in the $(h,z)$--plane. Up to the
common factor $\varepsilon$, the averaged system is therefore
equivalent to
\begin{equation}
\label{eq:averaged-hz-rescaled}
\frac{\rmd h}{\rmd\tau}
=
a_1\sin z,
\qquad
\frac{\rmd z}{\rmd\tau}
=
b_1\frac{C(h)}{B(h)}.
\end{equation}
This system is Hamiltonian with
\begin{equation}
\label{eq:averaged-hz-Hamiltonian}
\mathcal H_{\rm av}(h,z)
=
-a_1\cos z
-
b_1\int_1^h
\frac{C(s)}{B(s)}\,\rmd s.
\end{equation}
Indeed, its equations have the canonical form
\[
\frac{\rmd h}{\rmd\tau}
=
\frac{\partial\mathcal H_{\rm av}}{\partial z},
\qquad
\frac{\rmd z}{\rmd\tau}
=
-\frac{\partial\mathcal H_{\rm av}}{\partial h}.
\]

Using the definitions~\eqref{eq:B-h-periodic} and
\eqref{eq:C-h-periodic}, together with $C(1)=0$, we obtain
\begin{equation}
\label{eq:Hessian-Hav-Pk}
\det D^2\mathcal H_{\rm av}(\vP_j)
=
(-1)^j\frac{a_1b_1}{A(1)},\qquad j=0,1.
\end{equation}
Hence, for $a_1b_1>0$, $P_0$ is a nondegenerate extremum of
$\mathcal H_{\rm av}$ and therefore a center, whereas $P_1$ is a
saddle. For $a_1b_1<0$, the classification is reversed.

\end{proof}

\begin{figure*}[t]
\centering
\subfigure[Phase portrait of the averaged system]{
\includegraphics[width=0.41\textwidth]{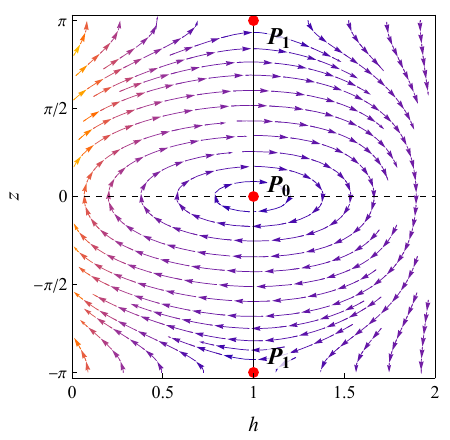}}
\hspace{1cm}
\subfigure[Poincar\'e section of the trigonometric Nos\'e--Hoover
oscillator]{
\includegraphics[width=0.4\textwidth]{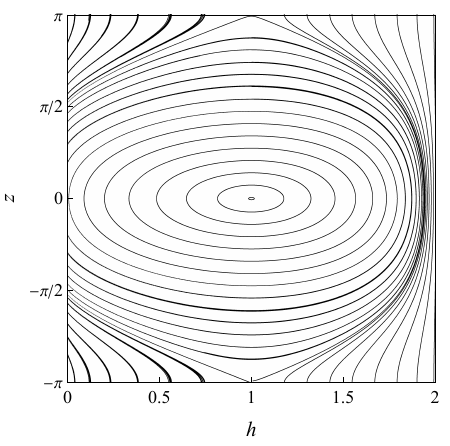}}
\caption{(Color online)
Comparison of the first-order averaged dynamics with the Poincar\'e
section of the trigonometric Nos\'e--Hoover oscillator for the
parameters~\eqref{eq:parameters-numerical-averaging-NH}.
(a) Phase portrait of the averaged system~\eqref{eq:Fbar-periodic}, with a center at $P_0$
and a hyperbolic saddle at $P_1$; the boundaries $z=\pm\pi$ are
identified.
(b) Poincar\'e section of the  trigonometric Nos\'e--Hoover
oscillator~\eqref{eq:periodic_NH}, defined
by~\eqref{eq:Poincare-section-x0-NH} and represented in the slow
coordinates~\eqref{eq:Poincare-hz-map}, showing the corresponding
center--saddle structure.}
\label{fig:Poincare-averaged-comparison-NH}
\end{figure*}
We now compare the phase portrait of the averaged system with a
Poincar\'e section of the full trigonometric Nos\'e--Hoover
oscillator. We take
\begin{equation}
\label{eq:parameters-numerical-averaging-NH}
a=b=10^{-3},
\end{equation}
corresponding through~\eqref{eq:ab-scaling-periodic-NH} to
$\varepsilon=10^{-3}$ and $a_1=b_1=1$. The factor $\varepsilon$
changes only the time scale of the first-order averaged dynamics and
therefore does not affect its phase portrait.

In this case, Theorem~\ref{thm:periodic-NH-averaged-type} gives a
center at $P_0$ and a hyperbolic saddle at $P_1$. Using
\eqref{eq:T-Omega-star-periodic} and~\eqref{eq:A-star-periodic}, the
corresponding eigenvalues can be written explicitly as
\begin{equation}
\label{eq:eigenvalues-star-periodic-NH}
\begin{aligned}
\lambda_{0,\pm}
&=
\pm\rmi\,
\frac{2K_*}{\pi}
\sqrt{
2\left(
\frac{E_*}{K_*}-\frac14
\right)
},
\\[1mm]
\lambda_{1,\pm}
&=
\pm
\frac{2K_*}{\pi}
\sqrt{
2\left(
\frac{E_*}{K_*}-\frac14
\right)
}.
\end{aligned}
\end{equation}
Numerically,
\begin{equation*}
\lambda_{0,\pm}
\simeq
\pm1.08376\,\rmi,
\qquad
\lambda_{1,\pm}
\simeq
\pm1.08376.
\end{equation*}

For the full system, we consider the Poincar\'e section
\begin{equation}
\label{eq:Poincare-section-x0-NH}
x=0\pmod{2\pi}.
\end{equation}
To facilitate a direct comparison with the averaged dynamics, the
intersections are represented in the slow coordinates $(h,z)$.
Since $H=\cos x+\cos y$, on
\eqref{eq:Poincare-section-x0-NH} this gives
\begin{equation}
\label{eq:Poincare-hz-map}
(h,z)=(1+\cos y,z),
\end{equation}
where $z$ is taken modulo $2\pi$ and represented in
$[-\pi,\pi)$.
\begin{figure*}[t]
\centering
\subfigure[Three-dimensional representation.]{
\includegraphics[width=0.45\textwidth]{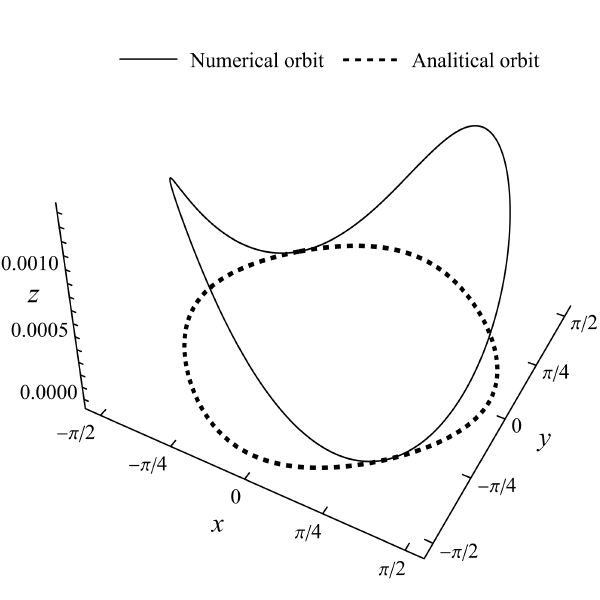}}\hspace{1cm}
\subfigure[Projection onto the $(x,y)$-plane.]{
\includegraphics[width=0.35\textwidth]{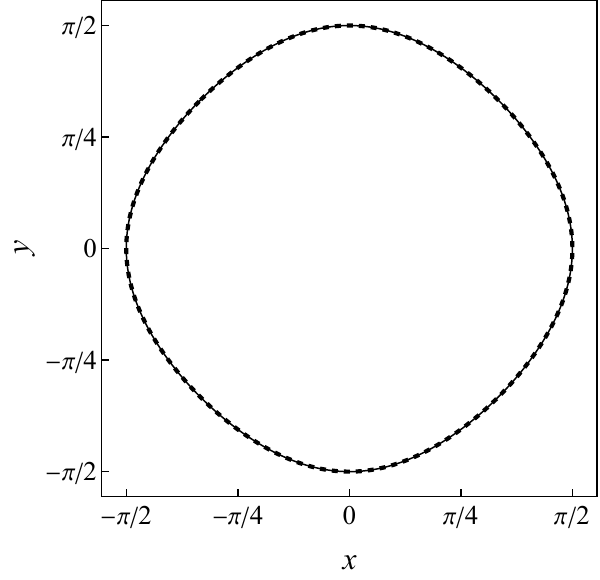}}
\caption{Comparison of a numerical trajectory of the full
trigonometric Nos\'e--Hoover oscillator for $a=b=10^{-3}$ (solid
curve) with the analytically parametrized limiting periodic orbit
$\Gamma_1\times\{0\}$ (dashed curve) associated with the branch
$\Gamma_{0,\varepsilon}$ selected by $P_0$:
(a) three-dimensional representation and
(b) projection onto the $(x,y)$-plane.}
\label{fig:periodic-orbit-NH}
\end{figure*}
The unperturbed orbit $\Gamma_1$ intersects this section at
$y=\pm\pi/2$. Both intersections are mapped to $h=1$, so that the
two crossings of the same unperturbed orbit have the same energy
coordinate and correspond directly to the energy level selected by
the averaged system.

Fig.~\ref{fig:Poincare-averaged-comparison-NH}(a) shows the phase
portrait of the averaged vector field~\eqref{eq:Fbar-periodic} on
the cylinder $(0,2)\times\mathbb S^1$. The center $P_0$ is
surrounded by closed level curves of $\mathcal H_{\rm av}$, whereas
the stable and unstable separatrices of $P_1$ form the corresponding
hyperbolic structure. Since $z$ is an angular variable, $z=\pi$ and
$z=-\pi$ represent the same point on the cylinder.

Fig.~\ref{fig:Poincare-averaged-comparison-NH}(b) shows the corresponding
Poincar\'e section of the original system~\eqref{eq:periodic_NH} in the
same $(h,z)$ coordinates. The regular invariant curves surrounding the
location corresponding to $P_0$ are consistent with the center predicted
by the averaged dynamics. In contrast, the neighborhood of $P_1$
exhibits the hyperbolic structure associated with a saddle, with
separatrix-like branches separating the surrounding regions of motion.
Thus, the Poincar\'e section reproduces the local center--saddle
organization predicted by the averaged system.

The comparison also illustrates the selection mechanism underlying
Theorem~\ref{thm:periodic-solutions-periodic-NH}. At
$\varepsilon=0$, the periodic orbits form the two-parameter family
$\Gamma_h\times\{z_0\}$. The first-order averaged dynamics selects
the energy level $h=1$ and the two phases $z=0$ and $z=\pi$, thereby
selecting the periodic branches $\Gamma_{0,\varepsilon}$ and
$\Gamma_{1,\varepsilon}$ from the degenerate unperturbed family.

For $a_1=b_1=1$, these branches correspond, respectively, to the
center $P_0$ and the hyperbolic saddle $P_1$ of the averaged system.
The elliptic structure around $\Gamma_{0,\varepsilon}$ observed in
the Poincar\'e section is consistent with the first-order averaged
dynamics, but this observation alone does not constitute a complete
nonlinear stability result for the corresponding periodic orbit of
the full three-dimensional system.

Fig.~\ref{fig:periodic-orbit-NH} illustrates the persistence of the
periodic orbit associated with the branch $\Gamma_{0,\varepsilon}$
for the same parameter values as in
Fig.~\ref{fig:Poincare-averaged-comparison-NH}. The solid curve is a
numerical trajectory of system~\eqref{eq:periodic_NH} with the initial
condition
\[
(x(0),y(0),z(0))
=
\left(-\frac{\pi}{3},\frac{\pi}{3},0\right),
\]
which lies on the limiting periodic orbit. The dashed curve is the
analytical parametrization of $\Gamma_1\times\{0\}$ given by in
\eqref{eq:x0-cos-periodic-NH}--\eqref{eq:y0-sin-periodic-NH}.

At $\varepsilon=0$, this orbit has period $T(1)$ given
by~\eqref{eq:T-Omega-star-periodic}. For $a=b=10^{-3}$, the numerical
trajectory remains close to the plane $z=0$, with only a small
oscillation in the $z$ direction, as seen in
Fig.~\ref{fig:periodic-orbit-NH}(a). Its $(x,y)$-projection is
practically indistinguishable, on the scale of
Fig.~\ref{fig:periodic-orbit-NH}(b), from the analytical orbit
$\Gamma_1$. This numerical agreement is consistent with the
persistence of the nearby periodic branch $\Gamma_{0,\varepsilon}$
established in Theorem~\ref{thm:periodic-solutions-periodic-NH}.

\section{$C^1$ non-integrability}
\label{sec:C1-nonintegrability-periodic-NH}

The periodic solutions obtained in
Theorem~\ref{thm:periodic-solutions-periodic-NH} can be used to
derive a local obstruction to the existence of regular first
integrals. The argument combines the relation between the
linearization of the averaged vector field and the characteristic
multipliers of the periodic solutions obtained by averaging
\cite{LlibreZhang:R3} with the Poincar\'e-type criterion of Llibre
and Valls~\cite{Llibre:11a::}.

Since system~\eqref{eq:periodic_NH} is autonomous, every periodic
orbit has the trivial characteristic multiplier $1$ associated with
the tangent direction to the flow. For the periodic solutions
obtained by first-order averaging, the leading-order behavior of the
remaining characteristic multipliers is determined by the
eigenvalues of the Jacobian matrix of the averaged vector field.
Consequently, their monodromy matrices need not be computed
explicitly.

\begin{theorem}
\label{thm:C1-nonintegrability-periodic-NH}
Consider system~\eqref{eq:periodic_NH} under the scaling
\eqref{eq:ab-scaling-periodic-NH}. For all sufficiently small
$\varepsilon>0$, the system admits no nonconstant first integral
$F\in C^1$ defined in a neighborhood of either periodic orbit
$\Gamma_{j,\varepsilon}$, $j=0,1$, of
Theorem~\ref{thm:periodic-solutions-periodic-NH}, such that
\[
\nabla F(\vx)\neq\boldsymbol{0}
\qquad
\text{for all }\vx\in\Gamma_{j,\varepsilon}.
\]
\end{theorem}

\begin{proof}
By~\eqref{eq:periodic-NH-averaged-eigenvalues} and the assumptions
$a_1b_1\neq0$, $A(1)>0$, and $\Omega(1)>0$, both eigenvalues
$\lambda_{j,\pm}$ of the Jacobian matrix of the averaged vector field
at $\vP_j$ are nonzero.

The relation between the averaged linearization and the
characteristic multipliers of the corresponding periodic solution
\cite{LlibreZhang:R3} gives
\[
\mu_{j,\pm}(\varepsilon)
=
1+2\pi\varepsilon\lambda_{j,\pm}
+O(\varepsilon^2),\qquad j=0,1.
\]
Hence
\[
\mu_{j,\pm}(\varepsilon)\neq1
\]
for all sufficiently small $\varepsilon>0$. Since the third
characteristic multiplier is the trivial multiplier $1$ associated
with the flow direction, the multiplier $1$ is simple.

Suppose that a nonconstant first integral $F\in C^1$ is defined in a
neighborhood of $\Gamma_{j,\varepsilon}$ and satisfies
\[
\nabla F(\vx)\neq\boldsymbol{0}
\qquad
\text{on }\Gamma_{j,\varepsilon}.
\]
By the Poincar\'e--Llibre--Valls criterion
\cite{Llibre:11a::}, the existence of such a regular first integral
would imply that the characteristic multiplier $1$ has algebraic
multiplicity at least two. This contradicts its simplicity established
above. Therefore no such first integral exists.

\end{proof}

\begin{remark}
Theorem~\ref{thm:C1-nonintegrability-periodic-NH} provides a local
and perturbative obstruction to $C^1$ integrability near the two
periodic orbits selected by averaging. It applies along the parameter
directions specified by~\eqref{eq:ab-scaling-periodic-NH} as
$\varepsilon\to0^+$.

The regularity condition on the first integral is essential: the
theorem does not exclude a $C^1$ first integral whose gradient
vanishes somewhere on $\Gamma_{j,\varepsilon}$. In the next section,
differential Galois theory provides a complementary obstruction to
meromorphic $B$-integrability in a neighborhood of a particular
non-equilibrium phase curve. This obstruction is non-perturbative
with respect to the parameters and does not require $a$ and $b$ to
be small.
\end{remark}

\section{Meromorphic non-integrability in the Bogoyavlensky sense}
\label{subsec:galois_nonintegrability}
\subsection{Differential-Galois integrability criteria}
\label{subsec:galois_criteria}
Differential Galois theory provides powerful obstructions to
meromorphic integrability through the analysis of variational
equations along a particular non-equilibrium solution. For
Hamiltonian systems, this approach is developed within the
Morales--Ramis theory and its subsequent
extensions~\cite{Morales:99::,Morales:01::}. It has been successfully
used to prove the meromorphic non-integrability of numerous
Hamiltonian systems arising in rigid body dynamics, celestial
mechanics, and other problems of mathematical
physics~\cite{Maciejewski:11::,Maciejewski:05::,10.1063/5.0200592,
mp:13::e,Combot:18::,Szuminski2024,Szuminski2025,Szuminski2026}.

Although the Morales--Ramis theory was originally developed for
Hamiltonian systems, the corresponding Galoisian obstruction was
extended to non-Hamiltonian systems by Ayoul and
Zung~\cite{Ayoul:10::}, using the notion of integrability introduced
by Bogoyavlensky. Since then, differential Galois methods have been
applied to the integrability analysis of various non-Hamiltonian
dynamical systems. In particular, effective integrability obstructions have been established for broad classes of nonlinear three-dimensional differential systems~\cite{Szuminski:20c::,Li:12::,HuangShiLi2017}. These methods have
also been applied directly to the classical Nos\'e--Hoover
oscillator~\cite{LiShiYang2022}.

The trigonometric Nos\'e--Hoover oscillator considered here is also
non-Hamiltonian, so the usual notion of Liouville integrability is not
directly applicable. We therefore work with meromorphic
$B$-integrability in the sense of Bogoyavlensky~\cite{Bogoyavlenski:96a::,Bogoyavlenski:96b::},
which provides the setting for the Ayoul--Zung extension~\cite{Ayoul:10::,MaciejewskiPrzybylska:09::}
used below.

Consider an autonomous system
\begin{equation}
\label{eq:general-nonham-system}
\dot{\boldsymbol{x}}
=
\boldsymbol{v}(\boldsymbol{x}),
\qquad
\boldsymbol{x}\in\mathbb C^n.
\end{equation}
For a non-Hamiltonian system, complete integrability corresponds to
the existence of $n-1$ functionally independent first integrals,
whereas the more general notion of $B$-integrability allows first
integrals to be complemented by commuting vector fields.

More precisely, system~\eqref{eq:general-nonham-system} is said to be
meromorphically $B$-integrable if there exist $k$ functionally
independent meromorphic first integrals
\[
F_1,\ldots,F_k,
\]
and $n-k$ linearly independent meromorphic vector fields
\[
\boldsymbol{u}_1=\boldsymbol{v},
\boldsymbol{u}_2,\ldots,\boldsymbol{u}_{n-k},
\]
such that
\[
[\boldsymbol{u}_i,\boldsymbol{u}_j]=0,
\qquad
\boldsymbol{u}_j[F_i]=0,
\]
for all admissible $i$ and $j$. Thus, the vector fields commute
pairwise and the functions $F_i$ are common first integrals of all
of them.

The connection between $B$-integrability and differential Galois
theory is provided by the result of Ayoul and
Zung~\cite{Ayoul:10::}. They showed that a meromorphically
$B$-integrable non-Hamiltonian system admits a cotangent lift
which is meromorphically integrable in the Liouville sense. As a
consequence, the Morales--Ramis obstruction extends to this class
of systems. A detailed discussion of $B$-integrability and its
relation to differential Galois theory is provided in the Appendix
of our previous work~\cite{Szuminski:18::}.

\begin{theorem}[Ayoul--Zung]
\label{thm:Ayoul-Zung-periodic-NH}
Assume that system~\eqref{eq:general-nonham-system} is
meromorphically $B$-integrable in a neighborhood of a
non-equilibrium phase curve $\Gamma$. Then the identity component
of the differential Galois group of the variational equations along
$\Gamma$ is Abelian.
\end{theorem}

Differential Galois theory also provides an obstruction to the
existence of individual first integrals. In dimension three, the
following criterion will be used~\cite{Li:12::,Szuminski:20c::}.

\begin{theorem}[Li--Shi]
\label{thm:Li-Shi-periodic-NH}
Consider a three-dimensional differential system and a
non-equilibrium phase curve $\Gamma$. If the system possesses a
non-constant meromorphic first integral in a neighborhood of
$\Gamma$, then the identity component of the differential Galois
group of the corresponding normal variational equation is solvable.
\end{theorem}

We shall use both Galoisian obstructions below. By the Ayoul--Zung
theorem, a non-Abelian identity component excludes meromorphic
$B$-integrability. On the other hand, by the Li--Shi theorem, a
non-solvable identity component of the differential Galois group of
the normal variational equation excludes the existence of any
non-constant meromorphic first integral in a neighborhood of
$\Gamma$.

\subsection{Variational equations and differential Galois group}
\label{subsec:NH_NVE_Galois}

For the differential-Galois analysis, we consider the particular
solution from the family~\eqref{eq:particular_periodic} corresponding
to $k=s=0$,
\begin{equation}
\label{eq:NH_particular_solution}
\vvarphi_{0,0}(t)
=
\left(
0,0,z_0-bt
\right)
\pmod{2\pi}.
\end{equation}
For $b\neq0$, this is a non-equilibrium solution lying on the invariant
phase curve
\begin{equation}
\label{eq:NH_invariant_curve}
\Gamma
=
\left\{
(x,y,z)\in\mathbb C^3:
x=y=0
\right\}.
\end{equation}

We can now state the main Galoisian non-integrability result.

\begin{theorem}
\label{th:NH_galois_nonintegrability}
Assume that $ab\neq0$. In any neighborhood of the non-equilibrium
phase curve $\Gamma$, system~\eqref{eq:periodic_NH} is not
meromorphically $B$-integrable and does not possess any non-constant
meromorphic first integral.
\end{theorem}

\begin{proof}
Let $(X,Y,Z)^T$ denote the variations of $(x,y,z)^T$ along the
particular solution~\eqref{eq:NH_particular_solution}. The first
variational equations along the phase curve $\Gamma$ are
\begin{equation}
\label{eq:NH_variational_equations}
\frac{\rmd}{\rmd t}
\begin{pmatrix}
X\\Y\\Z
\end{pmatrix}
=
\begin{pmatrix}
0&1&0\\
-1&-a\sin(z_0-bt)&0\\
0&0&0
\end{pmatrix}
\begin{pmatrix}
X\\Y\\Z
\end{pmatrix}.
\end{equation}
The variation $Z$ is tangent to the invariant curve, while the normal
variations $(X,Y)$ satisfy
\begin{equation*}
\label{eq:NH_NVE_system}
\dot X=Y,
\qquad
\dot Y=-X-a\sin(z_0-bt)Y.
\end{equation*}
Hence the normal variational equation can be written as
\begin{equation}
\label{eq:NH_NVE_t}
\ddot X
+
a\sin(z_0-bt)\dot X
+
X
=
0.
\end{equation}
To rationalize the trigonometric coefficients, we introduce directly
the new independent variable
\begin{equation}
\label{eq:NH_galois_s}
s=e^{\mathrm{i}(z_0-bt)}.
\end{equation}
The derivatives transform according to
\begin{equation*}
\label{eq:NH_derivative_transformation}
\frac{\rmd}{\rmd t}
=
-\mathrm{i}\,b\,s\frac{\rmd}{\rmd s},
\qquad
\frac{\rmd^2}{\rmd t^2}
=
-b^2\left(
s^2\frac{\rmd^2}{\rmd s^2}
+
s\frac{\rmd}{\rmd s}
\right).
\end{equation*}
Therefore,
\[
\dot X
=
-\mathrm{i}bsX',
\qquad
\ddot X
=
-b^2\left(s^2X''+sX'\right),
\]
where the prime denotes differentiation with respect to~$s$.
Moreover,
\[
\sin(z_0-bt)
=
\frac{1}{2\mathrm{i}}
\left(
s-\frac{1}{s}
\right).
\]
Substituting these expressions into
Eq.~\eqref{eq:NH_NVE_t} and dividing by $-b^2s^2$, we obtain
\begin{equation}
\label{eq:NH_NVE_rational}
X''
+
\left(
\frac{a}{2b}
+\frac{1}{s}
-\frac{a}{2bs^2}
\right)X'
-
\frac{1}{b^2s^2}X
=
0.
\end{equation}
Thus, the normal variational equation is transformed into a
second-order linear differential equation with rational coefficients.
Removing the first-derivative term by the standard transformation
\begin{equation}
\label{eq:NH_reduced_transformation}
X(s)
=
w(s)s^{-1/2}
\exp\left[
-\frac{a}{4b}
\left(
s+\frac{1}{s}
\right)
\right]
\end{equation}
gives the reduced equation
\begin{equation}
\label{eq:NH_reduced_DCHE}
w''=r(s)w,
\end{equation}
with
\begin{equation}
\label{eq:NH_r_DCHE}
r(s)
=
\alpha^2
+\frac{\alpha}{s}
+
\frac{b^{-2}-2\alpha^2-\frac14}{s^2}
+
\frac{\alpha}{s^3}
+
\frac{\alpha^2}{s^4},
\end{equation}
where \[\alpha=\frac{a}{4b}.\]
Equation~\eqref{eq:NH_reduced_DCHE} belongs to the family of double
confluent Heun equations.
 To study
its differential Galois group, we apply the Kovacic
algorithm~\cite{Kovacic:86::}.

\begin{theorem}[Kovacic]
\label{th:Kovacic_NH}
Let $\mathcal G$ be the differential Galois group of
\[
w''=r(s)w,
\qquad
r(s)\in\mathbb C(s).
\]
Then one of the following four cases occurs:
\begin{enumerate}
\item
$\mathcal G$ is triangularizable and the equation admits a
Liouvillian solution with rational logarithmic derivative;

\item
$\mathcal G$ is conjugate to a subgroup of the infinite dihedral
group and the logarithmic derivative of a solution is algebraic of
degree two over $\mathbb C(s)$;

\item
$\mathcal G$ is finite and all solutions are algebraic;

\item
$\mathcal G=\mathrm{SL}(2,\mathbb C)$ and the equation has no
Liouvillian solutions.
\end{enumerate}
\end{theorem}

The possible cases can be restricted by examining the orders of the
poles of $r$. Let $\Sigma'$ denote the set of finite poles and put
$\Sigma=\Sigma'\cup\{\infty\}$. If
\[
r(s)=\frac{p(s)}{q(s)},
\]
with relatively prime polynomials $p$ and $q$, we use
\[
o(\infty)=\deg q-\deg p.
\]

The following necessary conditions will be sufficient for our
purposes.

\begin{theorem}[Kovacic]
\label{th:Kovacic_NC_NH}
The following conditions are necessary for the first three cases of
Theorem~\ref{th:Kovacic_NH}:
\begin{enumerate}
\item
every finite pole has even order or order one, while
$o(\infty)$ is even or greater than two;

\item
there exists at least one finite pole of order two or of odd order
greater than two;

\item
every finite pole has order at most two and $o(\infty)\geq2$.
\end{enumerate}
\end{theorem}

We now apply these conditions to~\eqref{eq:NH_r_DCHE}. Throughout
the proof we assume
\[
a\neq0,
\qquad
b\neq0,
\]
and hence $\alpha\neq0$. The only finite pole of $r(s)$ is $s=0$ of order
$
o(0)=4.
$
Moreover, the expansion at infinity is
\begin{equation*}
\label{eq:NH_Laurent_infty}
r(s)
=
\alpha^2+\frac{\alpha}{s}
+O\left(\frac{1}{s^2}\right).
\end{equation*}
Thus,  the degree of infinity is $
o(\infty)=0$.

Hence, taking into account the character of these singularities, we
conclude that the necessary conditions for Case~1 are satisfied.
Case~2 is excluded because $r(s)$ has no finite pole of order two or
of odd order greater than two, whereas Case~3 is excluded because
$s=0$ is a pole of order four. Consequently, only Case~1 or Case~4
of Theorem~\ref{th:Kovacic_NH} can occur.

\begin{lemma}
\label{lem:NH_Galois}
For $a\neq0$ and $b\neq0$, the differential Galois group of
equation~\eqref{eq:NH_reduced_DCHE} is
$\operatorname{SL}(2,\mathbb C)$.
\end{lemma}

\begin{proof}
We apply the first case of the Kovacic algorithm.

The finite singularity $s=0$ is a pole of order
\[
o(0)=4=2\nu,
\qquad
\nu=2.
\]
Following the algorithm, we determine the principal part
$[\sqrt r]_0$ from the Laurent expansion of $\sqrt r$ at $s=0$.
From~\eqref{eq:NH_r_DCHE}, we have
\[
[\sqrt r]_0=\frac{\alpha}{s^2}.
\]
Moreover, the coefficient of $s^{-3}$ in $r(s)$ is $\alpha$, whereas
$([\sqrt r]_0)^2=\alpha^2/s^4$ has no term proportional to
$s^{-3}$. Hence,
\begin{equation*}
\label{eq:NH_E0}
E_0
=
\left\{
\alpha_0^+,\alpha_0^-
\right\}
=
\left\{
\frac32,\frac12
\right\}.
\end{equation*}

At infinity, we have
\[
o(\infty)=0=-2\nu,
\qquad
\nu=0.
\]
The corresponding polynomial part of $\sqrt r$ is
\[
[\sqrt r]_\infty=\alpha.
\] 
The coefficient of $s^{-1}$ in $r(s)$ is again $\alpha$, while
$([\sqrt r]_\infty)^2=\alpha^2$ contains no term proportional to
$s^{-1}$. Hence, for a singularity at infinity of order $-2\nu\leq0$, the algorithm
gives
\begin{equation*}
\label{eq:NH_Einfty}
E_\infty
=
\left\{
\alpha_\infty^+,\alpha_\infty^-
\right\}
=
\left\{
\frac12,-\frac12
\right\}.
\end{equation*}

We now pass to Step~2 of the Kovacic algorithm. We calculate the
Cartesian product
\[
E=E_0\times E_\infty
\]
and consider only those elements
\[
e=(e_0,e_\infty)\in E
\]
for which
\begin{equation*}
\label{eq:NH_de}
d(e)=e_\infty-e_0
\end{equation*}
is a non-negative integer. In the present case, there exists only one
element satisfying this condition, namely
\begin{equation*}
\label{eq:NH_admissible_element}
e=(\alpha_0^-,\alpha_\infty^+),
\qquad
d(e)=0.
\end{equation*}
For this choice, the corresponding rational function entering
Step~3 of the Kovacic algorithm is
\begin{equation*}
\label{eq:NH_Kovacic_omega}
\begin{split}
\omega(s)
&=
-[\sqrt r]_0
+\frac{\alpha_0^-}{s}
+[\sqrt r]_\infty
=
-\frac{\alpha}{s^2}
+\frac{1}{2s}
+\alpha.
\end{split}
\end{equation*}

We now pass to Step~3 of the algorithm. We look for a monic polynomial
$P(s)$ of degree $d(e)$ satisfying
\begin{equation}
\label{eq:NH_Kovacic_P}
P''
+
2\omega P'
+
\left(
\omega'
+
\omega^2
-
r(s)
\right)P
=
0.
\end{equation}
Since $d(e)=0$, we have $P=1$, and hence
Eq.~\eqref{eq:NH_Kovacic_P} reduces to
\begin{equation*}
\label{eq:NH_Kovacic_final_condition}
\omega'+\omega^2=r(s).
\end{equation*}
A direct computation gives
\begin{equation}
\label{eq:NH_Kovacic_residual}
\omega'(s)+\omega(s)^2-r(s)
=
-\frac{1}{b^2s^2}.
\end{equation}
Since $b\neq0$, this expression does not vanish identically.
Consequently, the polynomial required in Step~3 does not exist, and
Case~1 of the Kovacic algorithm is excluded.

Cases~2 and~3 have already been excluded by the necessary conditions
on the orders of the poles. Hence, only Case~4 remains, and therefore
the differential Galois group of the reduced equation
~\eqref{eq:NH_reduced_DCHE} is
$
\mathcal G
=
\operatorname{SL}(2,\mathbb C).
$
\end{proof}

Since
$
\mathcal G=\operatorname{SL}(2,\mathbb C),
$
its identity component is non-Abelian and non-solvable. The
reduction~\eqref{eq:NH_reduced_transformation} preserves solvability
of the identity component. Moreover, non-Abelianity of the identity
component of the reduced equation implies non-Abelianity of the
identity component of the unreduced normal variational equation
~\cite{HuangShiLi2017,Szuminski:20c::}. Hence, the identity component
of the differential Galois group of the normal variational equation
is both non-Abelian and non-solvable.

By the Ayoul--Zung theorem, its non-Abelianity excludes meromorphic
$B$-integrability in any neighborhood of $\Gamma$. Moreover, by
Theorem~\ref{thm:Li-Shi-periodic-NH}, its non-solvability excludes
the existence of any non-constant meromorphic first integral there.
Hence, for $ab\neq0$, system~\eqref{eq:periodic_NH} is not
meromorphically $B$-integrable and possesses no non-constant
meromorphic first integral in such a neighborhood.

\end{proof}

 \section{Conclusions}
\label{sec:conclusions}
In this paper, we have studied a trigonometric generalization of the
classical Nos\'e--Hoover oscillator, motivated by the relation between
the harmonic oscillator and its periodic pendulum counterpart. Although
the resulting system reproduces the classical model locally to leading
order, its natural phase space is the compact three-dimensional torus
$\mathbb{T}^3$. Our results show that this change in global geometry
has significant consequences for both the dynamics and integrability
properties of the system.

The numerical analysis reveals a transition from predominantly regular
dynamics to the coexistence of regular and chaotic motion. For weak
coupling $a$, the trigonometric and classical Nos\'e--Hoover oscillators
remain qualitatively similar near the central regular region, whereas
away from it the invariant structures of the trigonometric model are
destroyed more rapidly. As $a$ increases, chaotic layers grow and merge
into extended chaotic regions. Bifurcation diagrams identify a
period-doubling route to chaos, while the Lyapunov spectra and
Kaplan--Yorke dimensions characterize the transition and indicate
phase-space contraction in parts of the chaotic regime. Together with
the Poincar\'e sections and Lyapunov maps, these results reveal a
nontrivial organization of regular, chaotic, and contracting dynamics.

The limiting cases $a=0$ and $b=0$ differ substantially. For $a=0$,
the dynamics separates into an integrable planar subsystem and a driven
angular equation, and two functionally independent first integrals can
be constructed on the corresponding regular domains. For $b=0$, the
variable $z$ is constant and the phase space is foliated by invariant
two-dimensional tori. On each torus, the reduced dynamics can be
transformed into a polynomial vector field and analyzed using Darboux
theory. Thus, the integrable structures of the classical model are only
partially inherited by its trigonometric counterpart.

Near the intersection of the parameter axes, first-order averaging
yields two periodic solutions of the full system for sufficiently small
nonzero parameters. Their different linear types reproduce the
center--saddle organization observed in the Poincar\'e sections, while
the averaged linearization determines the leading behavior of their
nontrivial characteristic multipliers. Since these multipliers differ
from unity, the Poincar\'e--Llibre--Valls criterion provides an
obstruction to the existence of a $C^1$ first integral whose gradient
does not vanish along the corresponding periodic orbit.

This perturbative result is complemented by a differential-Galois
analysis. Along a suitable non-equilibrium particular solution, the
normal variational equation reduces to a second-order linear equation
of double confluent Heun type. The Kovacic algorithm shows that its
differential Galois group is $\operatorname{SL}(2,\mathbb C)$. The
Ayoul--Zung theorem therefore excludes meromorphic $B$-integrability,
while the Li--Shi criterion excludes any non-constant meromorphic first
integral. For $ab\neq0$, both obstructions hold in a neighborhood of
the corresponding non-equilibrium phase curve.

The two analytical approaches are complementary. Averaging provides a
local $C^1$ non-integrability obstruction in the perturbative regime,
whereas differential Galois theory gives an independent meromorphic
obstruction for arbitrary nonzero $a$ and $b$. To the best of our
knowledge, combining these two types of non-integrability obstructions
within the same non-Hamiltonian system has rarely been considered.

Taken together, our results show that the trigonometric Nos\'e--Hoover
oscillator is not merely a bounded periodic counterpart of the
classical model. It preserves the local leading-order structure while
changing the global geometry, the organization of regular and chaotic
motion, and the structure of the limiting cases. The model therefore
provides a simple setting in which numerical dynamics, periodic
solutions, and complementary notions of integrability and
non-integrability can be studied within a unified framework.

Although the present study is quite extensive, several open problems remain. It
would be interesting to study the statistical properties of the
trigonometric thermostat and, in particular, to determine whether it
possesses a physically meaningful invariant measure. The chaotic
attracting sets observed numerically also deserve a more detailed
analysis. In particular, the system shows an interesting transition as
the parameter $a$ is increased. For small $a$, the phase space has a
largely conservative-like structure. At intermediate values of $a$,
strongly dissipative dynamics develops, together with chaotic attracting
sets. Surprisingly, for larger $a$, the phase-space structure becomes
conservative-like again, with the attracting structures gradually
disappearing. Our additional numerical experiments indicate that this
behavior persists even for very large values of $a$, although these
results are not presented in the present paper. Understanding the
mechanism behind the transition from conservative-like to strongly
dissipative and then back to conservative-like dynamics remains an
interesting problem for further study.

Another open problem concerns the reduced two-dimensional system for
$b=0$. Although we have identified its qualitative phase-space structure
and analyzed its Darboux polynomials and exponential factors, a complete
classification of its integrability remains open. In particular, it
would be interesting to determine whether further exceptional integrable
cases exist beyond those identified here. More generally, these problems
may help to clarify which properties of the classical Nos\'e--Hoover
oscillator survive when its unbounded phase space is replaced by the
periodic geometry considered in this work.

\section*{Author contributions}

W.~S. conceived and formulated the research problem, developed the
analytical framework, performed the analytical and numerical computations,
prepared the figures, and wrote the original draft. J.~L. contributed to
the analytical framework, theoretical discussion, and interpretation of the
non-integrability results, and critically revised the manuscript.
Both authors read and approved the final version of the manuscript.
\newpage
\section*{Code availability}
The Mathematica code used to perform the numerical analysis presented
in this manuscript, including the computation of Lyapunov spectra,
Lyapunov maps, the Kaplan--Yorke dimension, bifurcation diagrams, and
the Lyapunov Integrability Test (LIT), is publicly available, together
with subsequent updates, at:

\url{https://github.com/WojciechSzuminski/LIT}.

\section*{Acknowledgements}

The present work was   carried out during a research stay at the Department of Mathematics, Universitat Autònoma de Barcelona, at the invitation of Professor Jaume Llibre.

\section*{Funding}
W.~S. was supported by the Polish National Agency for Academic Exchange
(NAWA) through the Bekker Programme (grant no.~BPN/BEK/2025/1/00055).
J.~L. was partially supported by the Agencia Estatal de Investigación of Spain
(grant no. PID2022-136613NB-100).
		\section*{Data availability }
The numerical data generated in this study can be reproduced using the
publicly available code described in the Code availability section.
Additional data supporting the findings of this study are available
from the corresponding author upon reasonable request.
	\section*{Declarations}
	\subsection*{Conflicts of interest} The author  declares no conflict of interest.\\
	
\bibliographystyle{unsrt}

\end{document}

%% file: mathdef.tex
\newcommand\mvector{\boldsymbol}

\newcommand\vv{\mvector{v}}

\newcommand\vx{\mvector{x}}

\newcommand\vP{\mvector{P}}

\newcommand\valpha{\mvector{\alpha}}

\newcommand\vvarphi{\mvector{\varphi}}

\newcommand\field{\mathbb}

\newcommand\Z{\field{Z}}

\newcommand\T{\field{T}}

\newcommand\rmd{\mathrm{d}}

\newcommand\rmi{\mathrm{i}\mspace{1mu}}